\documentclass[12pt]{article}

\usepackage[ruled,vlined]{algorithm2e}
\usepackage{graphicx}
\usepackage{multirow}
\usepackage{amssymb} 
\usepackage{amsthm}
\usepackage{amsmath}
\usepackage{fullpage}
\usepackage[longnamesfirst]{natbib}
\usepackage{enumerate}
\usepackage{amsfonts, bm, amsmath, color, natbib, rotating, amsthm, subfigure, lscape, multicol, epstopdf, verbatim, hyperref,soul}  
\usepackage{authblk}
\usepackage{xr}
\bibpunct{(}{)}{;}{a}{,}{,}

\usepackage{times}
    \usepackage{paralist}
    \usepackage[compact]{titlesec}
   \titlespacing{\section}{0pt}{2ex}{1ex}
    \titlespacing{\subsection}{0pt}{1ex}{0ex}
    \titlespacing{\subsubsection}{0pt}{0.5ex}{0ex}

\usepackage{setspace}
\theoremstyle{plain} \newtheorem{theorem}{Theorem} \newtheorem{proposition}{Proposition} \newtheorem{lemma}{Lemma} \newtheorem{corollary}{Corollary}

\theoremstyle{definition}    

\theoremstyle{remark} \newtheorem*{remark}{Remark}  
\begin{document}

\newif\ifblinded

\ifblinded
\title{Semiparametric Functional Factor Models with \\  Bayesian  Rank Selection} 
\author{}
\else

\title{\Large Semiparametric Functional Factor Models with\\  Bayesian  Rank Selection} 
\author{\large  Daniel R. Kowal\thanks{Dobelman Family Assistant Professor, Department of Statistics, Rice University (\href{mailto:Daniel.Kowal@rice.edu}{daniel.kowal@rice.edu}).} \ and Antonio Canale\thanks{Associate Professor, Department of Statistics,  University of Padova (\href{mailto:antonio.canale@unipd.it}{antonio.canale@unipd.it}).}}

\fi

\date{}

\maketitle
  
  \vspace{-12mm}

\begin{abstract}
Functional data are frequently accompanied by a parametric template that describes the typical shapes of the functions. However, these parametric templates can incur significant bias, which undermines both utility and interpretability. To correct for model misspecification, we augment the parametric template with an infinite-dimensional nonparametric functional basis. The nonparametric basis functions are learned from the data and constrained to be orthogonal to the parametric template, which preserves distinctness between the parametric and nonparametric terms. This distinctness is essential to prevent functional confounding, which otherwise induces severe bias for the parametric terms. The nonparametric factors are regularized with an ordered spike-and-slab prior that provides consistent rank selection and satisfies several appealing theoretical properties. The versatility of the proposed approach is illustrated through applications to synthetic data, human motor control data, and dynamic yield curve data. Relative to parametric and semiparametric alternatives, the proposed semiparametric functional factor model eliminates bias, reduces excessive posterior and predictive uncertainty, and provides reliable inference on the effective number of nonparametric terms---all with minimal additional computational costs. 


\end{abstract}

\thispagestyle{empty}

\clearpage
\setcounter{page}{1} 

\section{Introduction}
\subsection{Setting and goals} 
As high-resolution monitoring and measurement systems generate vast quantities of complex and highly correlated data, \emph{functional data analysis} has become increasingly vital for many scientific, medical, business, and industrial applications. Functional data are (noisy) realizations of random functions $\{Y_i\}_{i=1}^n$ observed over a continuous domain, such as time, space, or wavelength, and exhibit a broad variety of shapes. The concurrence of \emph{complex} and \emph{voluminous} data prompts the common use of nonparametric models for functional data. Yet in many cases, there is valuable information regarding the functional form of $Y_i$.  Template curves that describe the shape of $Y_i$ are derived from fundamental scientific laws or motivated by extensive empirical studies, and often are the focal point of the analysis. Prominent examples include human motor control \citep{Ramsay2000,goldsmith2016assessing}, interest rates \citep{nelson1987parsimonious, diebold2006forecasting}, and basal body temperature \citep{scarpa2009bayesian,scarpa2014enriched,canale2017pitman}.  

Our goal is to construct a functional data modeling framework that simultaneously (i) incorporates parametric templates in a coherent and interpretable manner, (ii) maintains the modeling flexibility of nonparametric methods, and (iii) provides computationally scalable inference with reliable uncertainty quantification. The approach is fully Bayesian, accompanied by  a highly efficient Markov chain Monte Carlo (MCMC) algorithm for posterior and predictive inference, and equally applicable to both densely-observed and sparsely- or irregularly-sampled functional data. 

The parametric templates are represented as a spanning set $\mathcal{H}_0 = \mbox{span}\{g_1(\cdot; \gamma),\ldots,g_L(\cdot; \gamma)\}$ of functions  $\{g_\ell(\cdot; \gamma)\}_{\ell=1}^L$ known up to  $\gamma$. Any function belonging to $\mathcal{H}_0$ is a linear combination of $\{g_\ell\}$; the corresponding coefficients and the parameters $\gamma$ must be learned. Important examples are presented in Table~\ref{tab-ex}.  The linear basis is routinely used for longitudinal data analysis and here is equivalent to a random slope model \citep{Molenberghs2000}. A change in slope, $(\tau - \gamma)_+$ with $(x)_+ = \max\{0,x\}$, is useful for modeling structural changes, such as a change in disease transmissions due to policy interventions \citep{Wagner2020}. Cosinor functions model circadian rhythms \citep{Mikulich2003} and other periodic behaviors \citep{Welham2006}. Biphasic curves offer utility in modeling  basal body temperature of women during the menstrual cycle \citep{scarpa2009bayesian,scarpa2014enriched,canale2017pitman}. In general, interest centers on learning the linear coefficients associated with each $g_\ell$, the nonlinear parameters $\gamma$, and an adequate yet interpretable model for the functions $Y_i$.

\begin{table}[ht]
\centering
\begin{tabular}{ c | c  | c | c   }
Linear &  Linear change &    Cosinor & Biphasic \\ \hline
 $\{1, \tau\}$&  $\{1, \tau, (\tau - \gamma)_+ \}$ &  $\{1, \sin(2\pi\tau/\gamma), \cos(2\pi\tau/\gamma)\}$ &  $\{1, \exp(\gamma\tau)/\{1 + \exp(\gamma \tau)\}\}$ 
\end{tabular}
\caption{\small Examples of parametric templates, or spanning sets, for $\mathcal{H}_0$. 
\label{tab-ex}}
\end{table}

The advantages of the parametric templates are clear: they incorporate domain knowledge, lend interpretability to the model, and often produce low-variance estimators relative to nonparametric alternatives. These templates are restrictive by design, and therefore can incur significant bias and other model misspecifications. Such effects erode model interpretability and can induce variance inflation. For illustration, we present model fits for two datasets in Figure~\ref{fig:fit} using a parametric model and our proposed semiparametric alternative; details and analyses of these data are in Section~\ref{apps}. In both instances, the templates capture the general shape of the data. However, regions of substantial bias are present, which produce uniformly inflated prediction bands over the domain. By comparison, the proposed approach preserves the essential shape of the curves, yet corrects the bias and shrinks the prediction bands appropriately---and crucially does so without overfitting. 

\begin{figure}[h]
\begin{center}
\includegraphics[width=.49\textwidth]{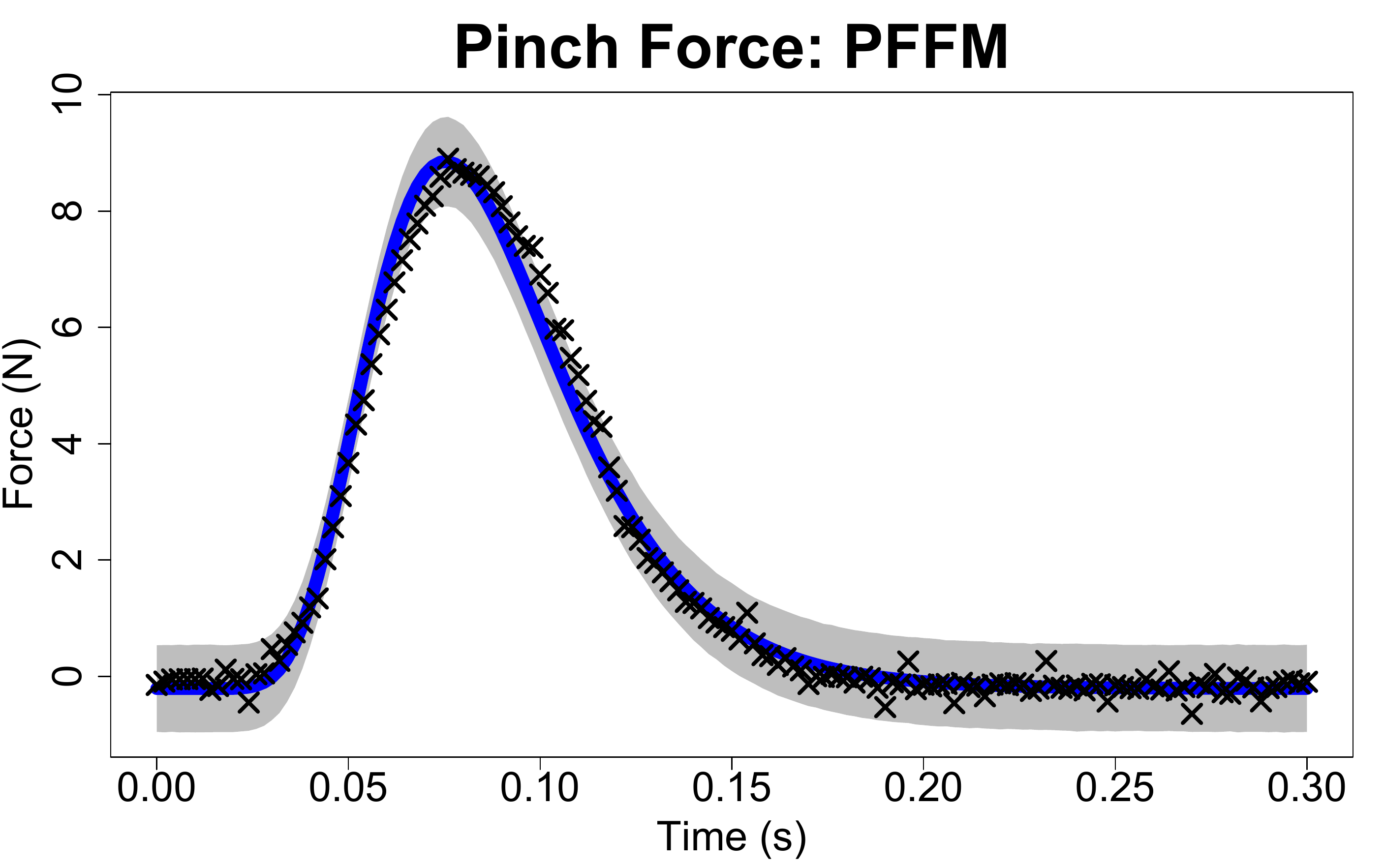}
\includegraphics[width=.49\textwidth]{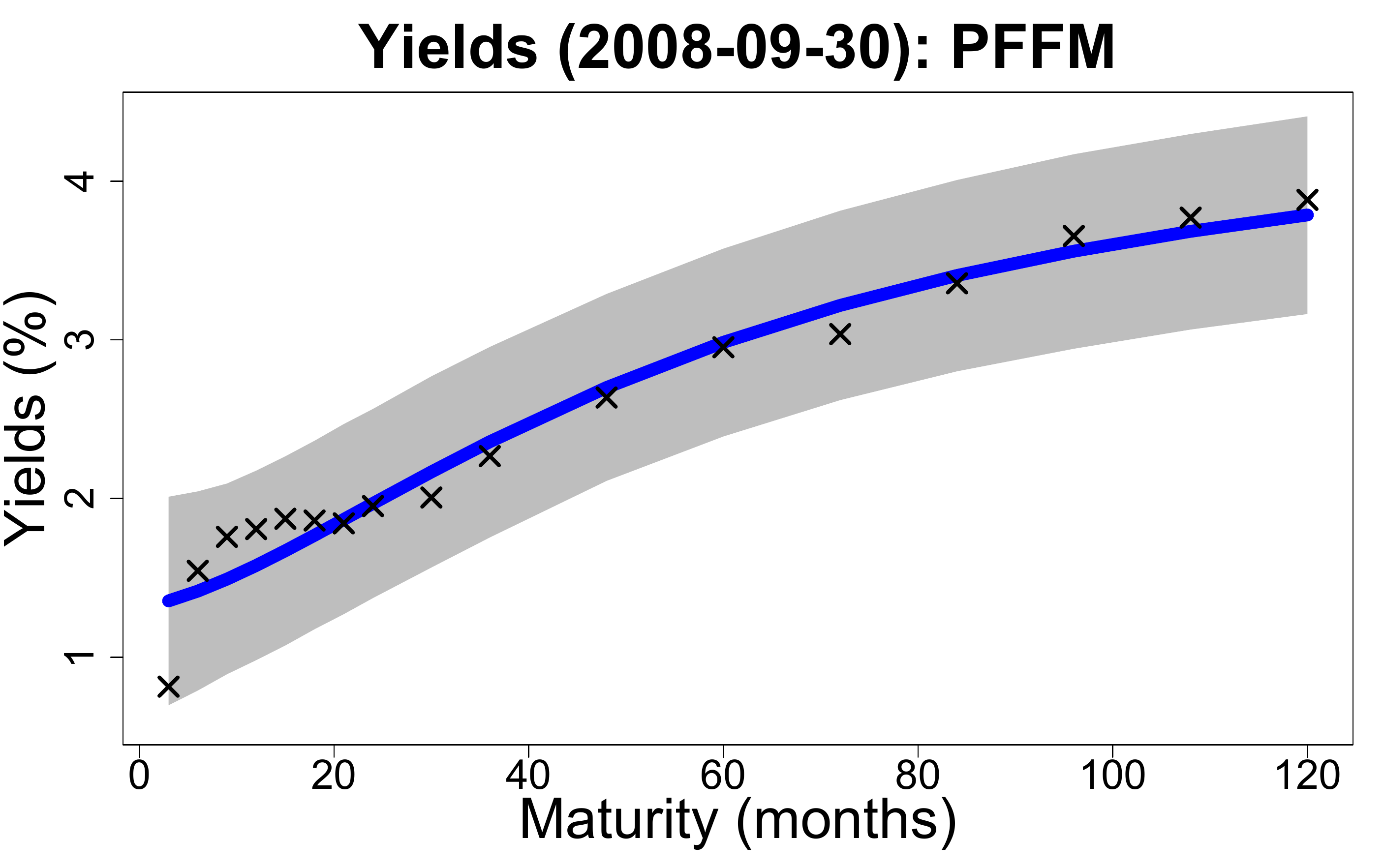}
\includegraphics[width=.49\textwidth]{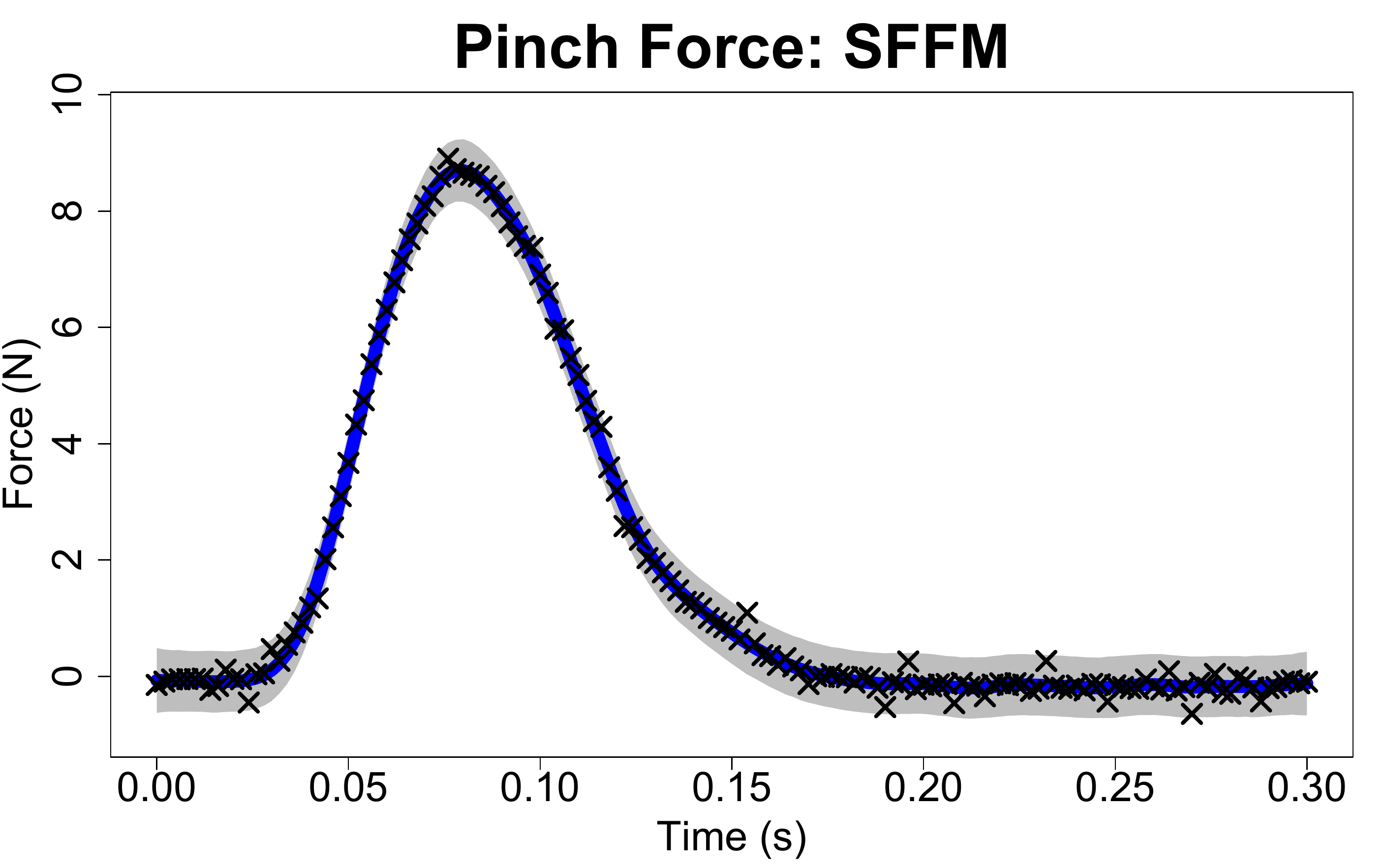}
\includegraphics[width=.49\textwidth]{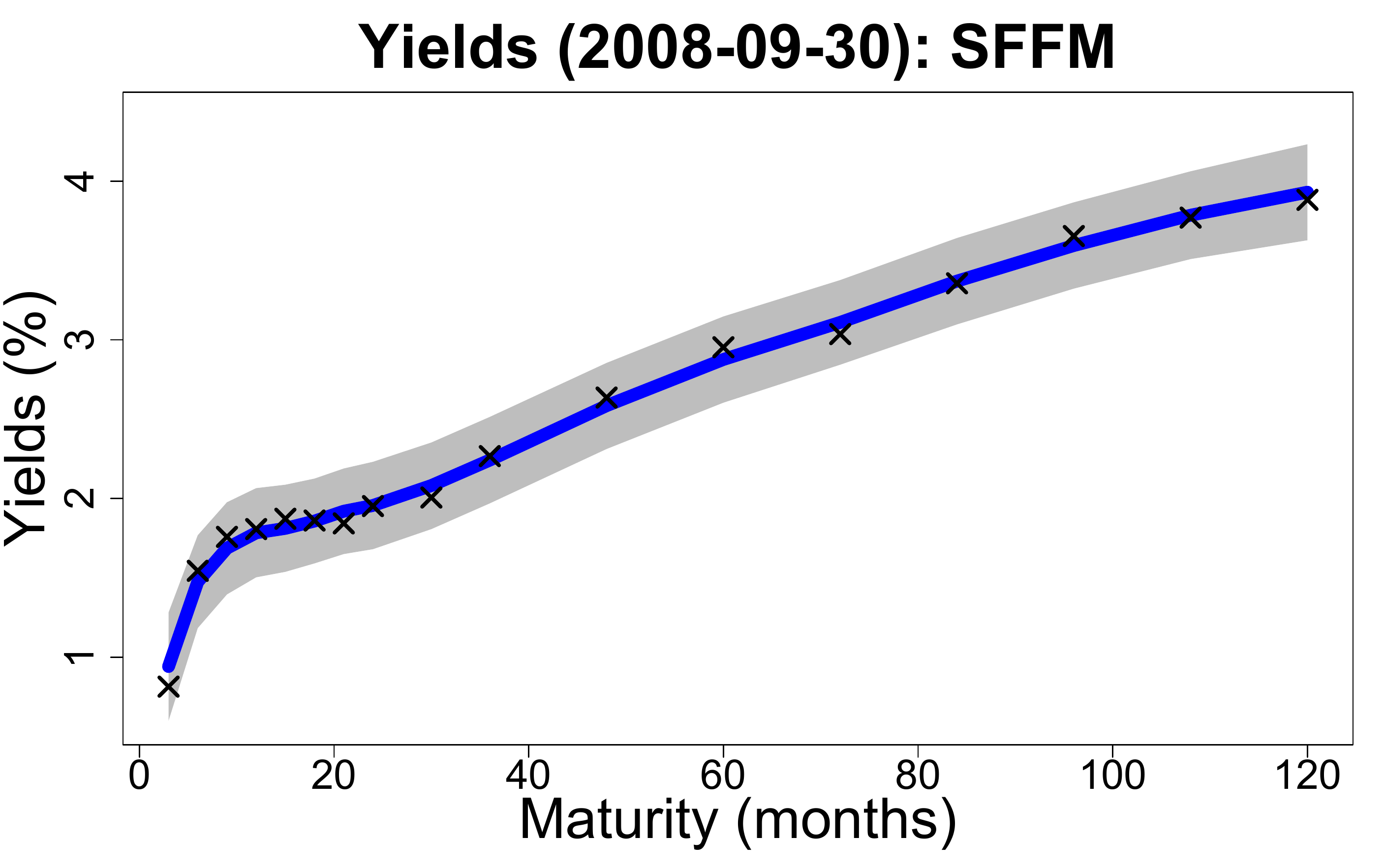}
\caption{\small Posterior expectations and 95\% simultaneous prediction bands for the parametric (PFFM, top) and semiparametric functional factor model (SFFM, bottom) for one replicate of the pinch force data (left; see Section~\ref{pinch}) and one yield curve (right; see Section~\ref{yields}). The proposed SFFM corrects the bias of the PFFM and offers narrower prediction bands. 
 \label{fig:fit}}
\end{center}
\end{figure}

\subsection{Overview of the proposed approach} 
The \emph{semiparametric functional factor model} (SFFM) bridges the gap between parametric and nonparametric functional data models. The SFFM augments a parametric template $\mathcal{H}_0$ 
 with a nonparametric and infinite-dimensional basis expansion for the functions $Y_i \in L^2(\mathcal{T})$:
 \begin{equation}\label{sffm}
 Y_i(\tau) = \sum_{\ell=1}^L \alpha_{\ell,i} g_\ell(\tau; \gamma)  + \sum_{k=1}^{\infty} \beta_{k,i}f_k(\tau), \quad \tau \in \mathcal{T},
\end{equation}
where $\{f_k\}$ are unknown functions, $\{\alpha_{\ell,i}\}$ and $\{\beta_{k,i}\}$ are unknown factors for the parametric and nonparametric components, respectively, and $\mathcal{T}  \subset \mathbb{R}^d$ is the domain, usually with $d=1$ for curves or $d=2$ for images. Without the nonparametric terms, \eqref{sffm} is a \emph{parametric functional factor model} (PFFM), which serves as our parametric baseline.

Our implementation of \eqref{sffm} is fully Bayesian with three unique and essential features.

First, the nonparametric basis $\{f_k\}$ is treated as unknown. The SFFM pools information across all functions $\{Y_i\}$ to learn the key functional features in the data---in particular, the systemic biases unresolved by $\mathcal{H}_0$. The functions $\{f_k\}$ are learned jointly with the parametric terms $\gamma$ and $\{\alpha_{\ell,i}\}$ and the nonparametric factors $\{\beta_{k,i}\}$, so uncertainty about $\{f_k\}$ is automatically absorbed into the joint posterior distribution.  

Second, the nonparametric basis $\{f_k\}$ is constrained to be orthogonal to the templates $\{g_\ell\}$. This important constraint enforces distinctness between the parametric and nonparametric components of the SFFM. We show that without such a constraint, inference on the parametric terms becomes severely biased and loses interpretability due to \emph{functional confounding}.  The nonredundancy of $\{f_k\}$ is also essential for valid inference on the effective number of nonparametric terms. Our chosen constraint produces computational simplifications that improve both the algorithmic efficiency and the ease of implementation of the MCMC algorithm. 


Third, the nonparametric factors $\{\beta_{k,i}\}$ are endowed with an ordered spike-and-slab prior distribution. The proposed prior is critical for coherence of the infinite-dimensional basis expansion in \eqref{sffm}: it provides much-needed regularization, encourages selection of a finite number of factors, and provides posterior inference for the effective number of nonparametric terms---including an assessment of whether any nonparametric component is needed at all. The prior admits a parameter expansion that offers substantial improvements in MCMC efficiency and satisfies several key properties that broaden applicability beyond the SFFM.

The SFFM in \eqref{sffm} is accompanied by an observation error equation to accommodate noisy and sparsely- or irregularly-sampled functional data. The observed data $\bm y_i = (y_{i,1},\ldots, y_{i,m_i})'$ are modeled as noisy realizations of $Y_i$ on a discrete set of points $\{\tau_{i,j}\}_{j=1}^{m_i} \subset \mathcal{T}$ for $i=1,\ldots,n$:
\begin{equation}\label{obs}
y_{i,j} = Y_i(\tau_{i,j}) + \epsilon_{i,j}, \quad \epsilon_{i,j} \stackrel{indep}{\sim}N(0, \sigma_\epsilon^2),
\end{equation}
although non-Gaussian versions are available. We proceed using common observation points $\tau_{i,j} = \tau_j$ and $m_i = m$ for notational simplicity, but this restriction may be relaxed.

Although we focus on the parametric and semiparametric versions of \eqref{sffm}, the proposed modeling framework remains useful without any  parametric template ($L=0$). In this case, \eqref{sffm} resembles a Karhunen-Lo\`{e}ve decomposition, and $\{f_k\}$ correspond to the eigenfunctions of the covariance function of $\{Y_i\}$ (assuming the functions $Y_i$ have been centered). The Karhunen-Lo\`{e}ve decomposition provides the theoretical foundation for functional principal components analysis (FPCA). As such, model \eqref{sffm}---together with the model for $\{f_k\}$ and the ordered spike-and-slab prior for $\{\beta_{k,i}\}$---constitutes a new Bayesian approach FPCA. 

\subsection{Review of related approaches}
Semiparametric models for functional data are predominantly non-Bayesian. For modeling a single function, L-splines combine a goodness-of-fit criterion with a penalty on deviations from $\mathcal{H}_0$ \citep{ramsay1991some,heckman2000penalized}. Although L-splines elegantly admit analytic expressions  for  $\{f_k\}$, 
these derivations are highly challenging for all but the simplest choices of $\mathcal{H}_0$, which inhibits widespread practical use. 
In addition, L-splines offer limited direct inference on the adequacy of the parametric templates $\mathcal{H}_0$, and typically require hypothesis tests with asymptotic validity. 
Other non-Bayesian approaches for semiparametric functional data analysis seek to replace nonparametric functions with parametric alternatives. \cite{sang2017parametric} attempt to simplify FPCA by using polynomials for each FPC instead of splines or Fourier functions. In functional regression analysis, \cite{Chen2019} develop hypothesis tests to determine whether an unknown regression function deviates from a parametric template. 


Bayesian semiparametric functional data models are less common. The ``semiparametric" model of 
\cite{Lee2018} refers to additive rather than linear effects, but does not include a parametric template like $\mathcal{H}_0$. \cite{scarpa2009bayesian} construct a Dirichlet process mixture of a parametric function and a Gaussian process contamination, which is generalized by \cite{scarpa2014enriched}  to include prior information on the frequencies of certain functional features. These methods are designed primarily for clustering: they identify individual curves $Y_i$ that deviate substantially from the parametric model, while the remaining curves are presumably well modeled parametrically. 
The SFFM is capable of modeling total deviations from $\mathcal{H}_0$ for a particular $Y_i$, but also captures---and corrects---partial deviations from $\mathcal{H}_0$ that persist for some or many $Y_i$. 
Unlike the mixture models, the SFFM is well suited for including additional layers of dependence, such as hierarchical (Section~\ref{pinch}) or dynamic (Section~\ref{yields}) models, while maintaining efficient posterior computing.  Notably, these existing methods do not address functional confounding.


The remainder of the paper is organized as follows. The model for the parametric and nonparametric functions  is in Section~\ref{all-np}. The ordered spike-and-slab prior is introduced and studied in Section~\ref{order-prior}. The MCMC algorithm is discussed in Section~\ref{mcmc}. A simulation study is in Section~\ref{sims}. The model is applied to real datasets in Section~\ref{apps}. We conclude in Section~\ref{dis}. Online supplementary material includes   proofs of all results, the full MCMC algorithm, additional simulations,  and \texttt{R} code for reproducibility.









\section{Modeling the nonparametric functions} \label{all-np}
The error-free latent functions $\{Y_i\}$ belong to the space spanned by the template parametric curves $\{g_\ell\}$ and the nonparametric curves $\{f_k\}$. Any systemic bias resulting from the inadequacies of $\{g_\ell\}$ must be corrected by  $\{f_k\}$, which demands substantial flexibility of the nonparametric basis $\{f_k\}$.  However, interpretability of the parametric terms $\{g_\ell\}$ and $\{\alpha_{\ell,i}\}$ requires a strict distinction between the parametric and nonparametric components.. In conjunction, these points demand both \emph{flexibility} and \emph{restraint} from the nonparametric term. 

To emphasize the importance of the distinction between $\{g_\ell\}$ and $\{f_k\}$, consider a seemingly reasonable alternative to \eqref{sffm} that replaces the nonparametric term with a Gaussian process (PFFM+gp):
\begin{equation}\label{pffm+gp}
Y_i(\tau)  = \sum_{\ell=1}^L \alpha_{\ell,i} g_\ell(\tau; \gamma) + h_i(\tau), \quad h_i  \sim \mathcal{GP}(0, \mathcal{K}_h).
\end{equation}
The PFFM+gp is a special case of the  SFFM \eqref{sffm}, and implies 
that each latent curve $Y_i$ is a Gaussian process centered at the parametric term. Naturally, it may be expected that each $h_i$ corrects for the biases of the parametric component. However, there is a significant cost: despite being centered at zero, the Gaussian process $h_i$ is not constrained to be sufficiently distinct from the parametric term, and therefore introduces a \emph{functional confounding} that biases the parametric factors $\{\alpha_{\ell,i}\}$. 

To illustrate this point, we fit the PFFM, PFFM+gp, and SFFM   to the pinch force data (see Section~\ref{pinch}) with identical parametric models (see \eqref{par-pinch}). The resulting posterior expectations of each $\{\alpha_{\ell,i}\}$ are presented in Figure~\ref{fig:alphapm}. Most striking, the point estimates of the parametric factors are nearly identical between the PFFM and SFFM, while the PFFM+gp estimates are substantially different. The discrepancy of the PFFM+gp is concerning: under mild conditions, the PFFM will produce (approximately) unbiased estimates of the linear coefficients $\{\alpha_{\ell,i}\}$ even in the presence of dependent errors, so the point estimates are a reliable benchmark. These results are also confirmed in the simulation study using the ground truth values of $\{\alpha_{\ell, i}\}$ (Section~\ref{sims}). 
We emphasize that  despite this agreement between the PFFM and the SFFM for point estimation of $\{\alpha_{\ell,i}\}$,  the SFFM \emph{does} identify inadequacies of the PFFM and provides key improvements in bias removal  for $\{Y_i\}$ and uncertainty quantification for both $\{Y_i\}$ and $\{\alpha_{\ell,i}\}$ (see Sections~\ref{sims}-\ref{apps}).

\begin{figure}[h]
\begin{center}
\includegraphics[width=.49\textwidth]{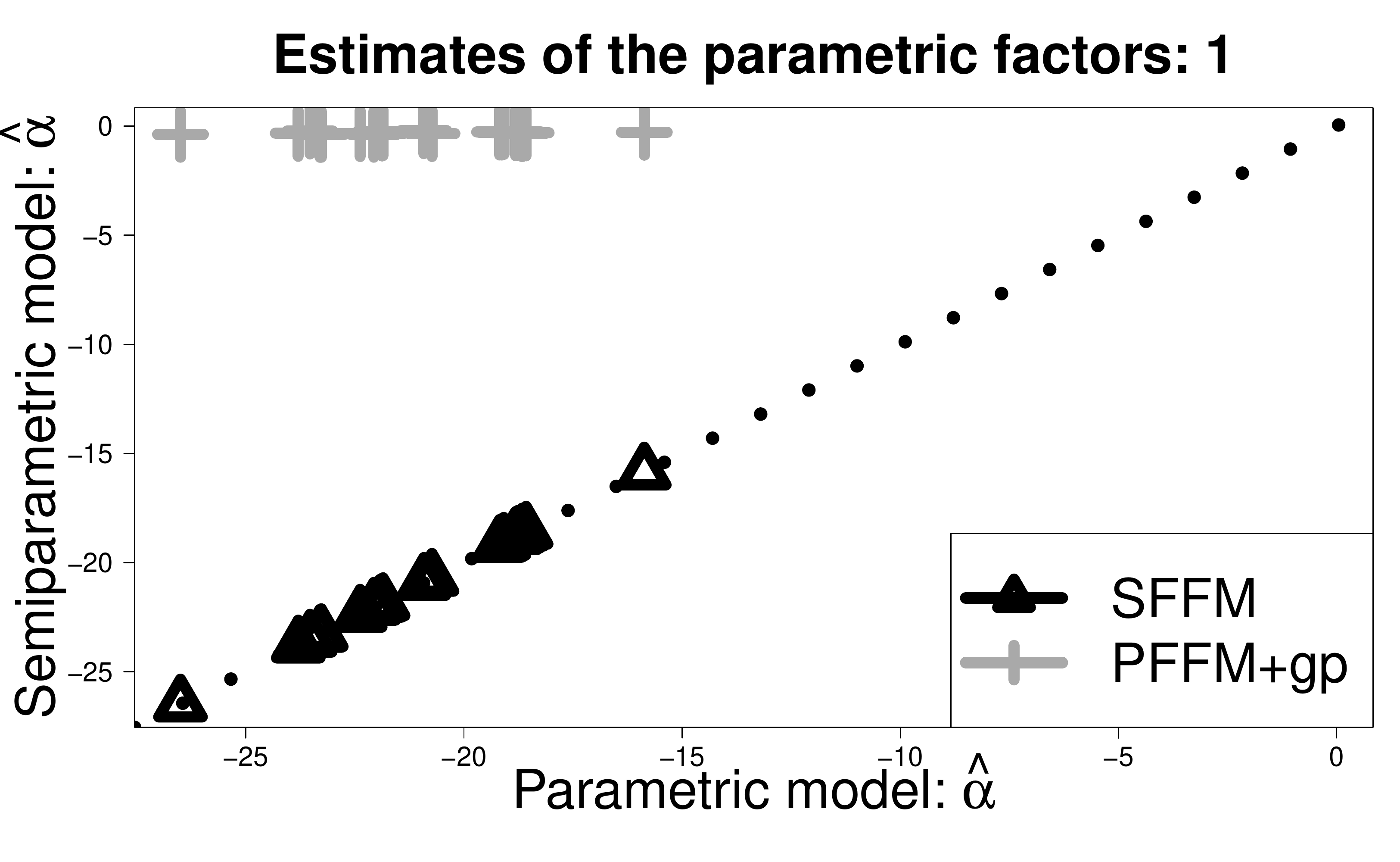}
\includegraphics[width=.49\textwidth]{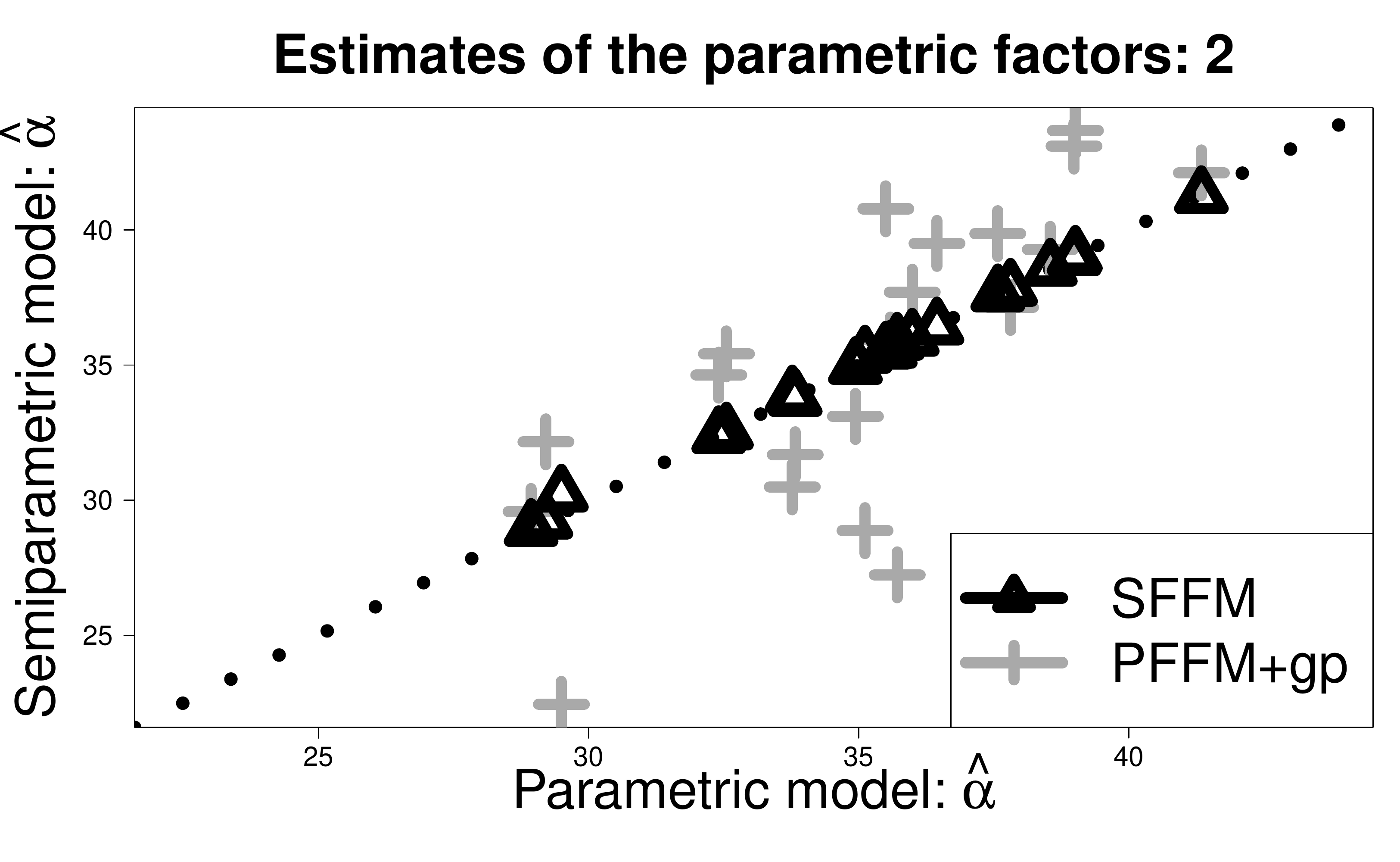}
\caption{\small Posterior expectations of $\{\alpha_{\ell,i}\}_{i=1}^n$ for $\ell=1$ (left) and $\ell =2$ (right)  for the pinch force data (see Section~\ref{pinch}) comparing the PFFM estimates (x-axis) against both the  SFFM and  PFFM+gp  (y-axis). The point estimates of the parametric factors are nearly identical between the PFFM and SFFM, while the PFFM+gp estimates are substantially different. The dotted line denotes $y=x$.
 \label{fig:alphapm}}
\end{center}
\end{figure}

The simple augmentation of a Gaussian process in \eqref{pffm+gp} offers  flexibility but lacks restraint: $\{h_i\}$ absorbs variability otherwise explained by $\{g_\ell\}$, which ultimately corrupts inference and interpretation of the parametric term. This functional confounding is not simply an artifact of the PFFM+gp sampling algorithm, which uses a \emph{joint} sampler for  $\{\alpha_{\ell,i}, h_i\}$  (see the supplement), and is robust to the choice of the covariance function $\mathcal{K}_h$. 
This issue is related to spatial confounding (e.g., \citealp{Reich2006}) but has not been thoroughly explored or resolved for functional data.

The proposed SFMM achieves both flexibility and distinctness by (i) modeling each $f_k$ in \eqref{sffm} as a  
smooth unknown function and (ii) constraining each $f_k$ to be orthogonal to $\{g_\ell\}$. The model for $f_k$ can be any Bayesian curve-fitting model, such as splines, wavelets, or Gaussian processes, usually with a prior that encourages smoothness; our specifications are discussed subsequently.  Crucially, our choice of constraints not only ensure   nonredundancy of $\{f_k\}$, but also offer key computational simplifications that improve scalability and increase MCMC efficiency. These results maintain regardless of the specification of $\{g_\ell\}$ or $\{f_k\}$. 

Consider model \eqref{sffm}-\eqref{obs} evaluated at the observation points $\{\tau_j\}_{j=1}^m$:
\begin{equation}\label{like}
\bm y_i =  \bm G_\gamma \bm \alpha_i + \bm F\bm \beta_i + \bm \epsilon_i, \quad \bm \epsilon_i \stackrel{indep}{\sim}N(\bm 0, \sigma_\epsilon^2 \bm I_m), \quad i=1,\ldots,n
\end{equation}
where $\bm G_{\gamma} = (\bm g_{1; \gamma},\ldots, \bm g_{L; \gamma})$ is the $m \times L$ parametric basis matrix with $\bm g_{\ell; \gamma} = (g_\ell(\tau_1; \gamma),\ldots, g_\ell(\tau_m; \gamma))'$, $\bm \alpha_i = (\alpha_{1,i},\ldots,\alpha_{L,i})'$ is the vector of parametric factors, $\bm F = (\bm f_1, \bm f_2, \ldots )$ is the $m \times \infty$ nonparametric basis matrix with $\bm f_k = (f_k(\tau_1),\ldots, f_k(\tau_m))'$, and $\bm \beta_i = (\beta_{1,i}, \beta_{2,i},\ldots,)'$ is the vector of nonparametric factors. 
By constraining each $\bm f_k$ to be orthogonal to each $\bm g_{\ell;\gamma}$, the likelihood  \eqref{like} factors into parametric and nonparametric terms: 
\begin{lemma}\label{like-decomp}
When $\bm G_\gamma ' \bm f_k = \bm 0_{L}$ for all $k=1,2,\ldots$ and conditional on $\sigma_\epsilon^2$, the likelihood \eqref{like} factorizes:
$p(\bm y \mid  \gamma , \{\bm \alpha_i\}, \bm F,  \{\bm \beta_i\}) =  p_{0}(\bm y \mid  \gamma ,\{\bm \alpha_i\}) p_{1}(\bm y \mid  \bm F , \{\bm \beta_i\})$,
where $p_0$ depends \emph{only} on the parametric terms $\gamma$  and $\{\bm \alpha_i\}$ and  $p_1$ depends \emph{only} on  the nonparametric terms $\bm F$  and $\{\bm \beta_i\}$. 
\end{lemma}
   
Within a Bayesian model,  a reasonable notion of distinctness between parametric and nonparametric components is (conditional) independence in the posterior. By assuming independence in the prior, Lemma~\ref{like-decomp} ensures this result:
\begin{corollary}\label{post-indep}
Under model \eqref{like} with prior independence $p(\{\bm \alpha_i\},  \{\bm \beta_i\} ) =  p( \{\bm \alpha_i\}) p(\{\bm\beta_i\})$ and  $\bm G_\gamma ' \bm f_k = \bm 0_{L}$ for all $k=1,2,\ldots$, the parametric and nonparametric factors are \emph{a posteriori} independent: 
$p( \{\bm \alpha_i\}, \{\bm \beta_i\} \vert \bm y, \sigma_\epsilon^2,  \bm F, \gamma) =  p(\{\bm \alpha_i\}\vert \bm y, \sigma_\epsilon^2,  \bm F, \gamma) p( \{\bm\beta_i\}\vert \bm y, \sigma_\epsilon^2,  \bm F, \gamma)$. 
\end{corollary}
Corollary~\ref{post-indep} is crucial not only for preserving \emph{distinctness} between the parametric and nonparametric terms, but also for designing an \emph{efficient} MCMC sampler. The  specific constraint $\bm G_\gamma' \bm f_k = \bm 0_L$ ensures that the sampling steps for all parametric factors  $\{\bm \alpha_i\}$ do \emph{not} depend on the nonparametric parameters in any way: the sampling steps for $\{\bm \alpha_i\}$ are identical to the fully parametric case, while the sampling steps for  $\{\bm\beta_i\}$ proceed exactly as in a fully nonparametric setting (see Section~\ref{order-prior}). This decoupling of parametric and nonparametric sampling steps is particularly useful for more complex models. At the same time, Corollary~\ref{post-indep} implies that \emph{separate} sampling steps for $\{\bm \alpha_i\}$ and $\{\bm \beta_i\}$ are equivalent to a \emph{joint} sampler for all $\{\bm \alpha_i, \bm \beta_i\}$. Joint sampling steps are typically preferred for greater MCMC efficiency, but often increase the computational burden. For instance, a joint sampler for $\{\bm \alpha_i, h_i\}$ under the PFFM+gp requires either a large block sampler or marginalization over either $\bm\alpha_i$ or $h_i$, both of which become more challenging when the model for $\bm \alpha_i$ is complex. The SFFM---with the proposed orthogonality constraint---completely avoids this tradeoff and guarantees  joint sampling steps even for complex models. Such a reduction is not available in general, for example $L^2$ orthogonality $\int_\mathcal{T} g_\ell(\tau; \gamma) f_k(\tau) \, d\tau = 0$ for $\ell=1,\ldots,L$.

To enforce the proposed constraint, we \emph{condition} on orthogonality $\bm G_\gamma ' \bm f_k = \bm 0_{L}$ in the model for $\bm f_k$.  Conditioning is a natural Bayesian mechanism for incorporating such information into estimation and inference. We model the unknown functions $f_k(\tau) = \bm b'(\tau) \bm \psi_k$ using known basis functions $\bm b$ and unknown coefficients $\bm \psi_k$, so a  prior on $\bm \psi_k$ implies  a prior on the function $f_k$, for example to encourage smoothness. Our default specification uses low-rank thin plate splines $\bm b$, which are flexible, computationally efficient, and well-defined on $\mathcal{T} \subset \mathbb{R}^d$ for $d \in \mathbb{Z}^+$, along with the smoothness prior $[\bm \psi_k \mid \lambda_{f_k}] \stackrel{indep}{\sim} N(\bm 0, \lambda_{f_k}^{-1} \bm \Omega^{-})$, where  $\bm \Omega$ is the (known) matrix of integrated squared second derivatives, $[\bm \Omega]_{j, j'} = \int \ddot{b}_j(\tau) \ddot{b}_{j'}(\tau) d\tau$, which implies that the prior distribution represents the classical roughness penalty $-2\log p(\bm \psi_k \mid \lambda_{f_k}) \stackrel{c}{=} \lambda_{f_k} \bm \psi_k' \bm \Omega \bm \psi_k = \lambda_{f_k} \int_\mathcal{T} \{\ddot{f}_k(\tau)\}^2 d\tau$ and the smoothing parameter $\lambda_{f_k}^{-1/2} \stackrel{indep}{\sim} \mbox{Uniform}(0, 10^4)$ is learned from the data. 

Absent the orthogonality constraint, the full conditional distribution of each vector of coefficients is $[\bm \psi_k \mid -] \sim N\left(\bm Q_{\psi_k}^{-1} \bm \ell_{\psi_k}, \bm Q_{\psi_k}^{-1}\right)$, where 
$\bm Q_{\psi_k} =  \sigma_\epsilon^{-2}(\bm B'\bm B)\sum_{i=1}^n \beta_{k,i}^2 +   \lambda_{f_k} \bm \Omega $, $\bm \ell_{\psi_k} =   \sigma_\epsilon^{-2}\bm B' \sum_{i=1}^n\big\{\beta_{k,i} \big( \bm y_i  - \bm G_\gamma \bm\alpha_i - \sum_{\ell \ne k} \bm f_\ell \beta_{\ell,i}\big)\big\}$, and $\bm B = (\bm b'(\tau_1), \ldots, \bm b'(\tau_m))'$. Samples from the full conditional distribution of $\bm \psi_k$ \emph{conditional} on $\bm G_\gamma ' \bm f_k = \bm 0_{L}$ can be obtained by (i) sampling from the unconstrained distribution $[\bm \psi_k \mid -]$ and (ii) applying a simple shift:
  \begin{lemma}
	\label{lem:constraint}
	Suppose $\bm G_\gamma$ is $L\times m$ with rank $L$. Denote $\bm\psi_k^0 \sim N\left(\bm Q_{\psi_k}^{-1} \bm \ell_{\psi_k}, \bm Q_{\psi_k}^{-1}\right)$ with $\bm f_k^0 = \bm B \bm \psi_k^0$. The shifted term  $\bm f_k = \bm B \bm \psi_k$ with $\bm \psi_k = \bm \psi_k^0 - \bm Q_{\psi_k}^{-1} \bm B' \bm G_\gamma \left(\bm G_\gamma' \bm B\bm Q_{\psi_k}^{-1}\bm B' \bm G_\gamma\right)^{-1} \bm G_\gamma' \bm B\bm \psi_k^0$  satisfies  $\bm f_k \stackrel{d}{=}[\bm f_k^0 \mid \bm G_\gamma' \bm f_k^0 = \bm 0]$ and $ \mathbb{P}(\bm G_\gamma' \bm f_k = \bm 0) = 1$. 
\end{lemma}
Lemma~\ref{lem:constraint} provides a simple and effective mechanism to adapt any standard (basis) regression sampling step for $[\bm \psi_k \mid -] $ to accommodate this crucial orthogonality constraint, including in the case when the nonlinear term $\gamma$ is unknown. Note that it is possible to adapt Lemma~\ref{lem:constraint} for the PFFM+gp; however, the resulting model still does not provide direct inference on the necessity of the nonparametric factors  (i.e., rank selection) or describe the systemic biases of $\mathcal{H}_0$, which are captured by $\{f_k\}$ in the SFFM. With or without additional constraints, the PFFM+gp sampling algorithm lacks the same scalability in $m$ (see Section~\ref{apps}).

For further computational simplifications along with parameter identifiability, we additionally constrain $\bm G_\gamma$ and $\bm F$ to be \emph{orthonormal}, which   implies that the joint basis $\{g_1,\ldots, g_L, f_1, f_2, \ldots\}$ in \eqref{sffm} is orthonormal. First, $\bm G_\gamma$ is orthogonalized using a QR decomposition. For any value of $\gamma$, let $\bm G_\gamma^0 = \bm Q_\gamma \bm R_\gamma$ be the QR decomposition of the initial basis matrix   $\bm G_{\gamma}^0 = (\bm g_{1; \gamma},\ldots, \bm g_{L; \gamma})$. By setting $\bm G_\gamma = \bm Q_\gamma$, we ensure that $\bm G_\gamma' \bm  G_\gamma = \bm I_L$ and the columns of $\bm G_\gamma$ span the same space as the columns of $\bm G_\gamma^0$. When $\gamma$ is unknown and endowed with a prior distribution, the QR decomposition is incorporated into the likelihood evaluations of \eqref{like} for posterior sampling of $\gamma$. The orthonormality of $\bm G_\gamma$ provides critical simplifications for efficient posterior inference: 
\begin{corollary}\label{cor-alpha}
Under model \eqref{like} and subject to $\bm G_\gamma'\bm f_k = \bm 0_L$ for all $k$ and $\bm G_\gamma'\bm G_\gamma = \bm I_L$, the likelihood for the parametric factors is proportional to $p(\bm y \mid  \gamma , \{\bm \alpha_i\}, \bm F,  \{\bm \beta_i\}) \propto  p_{0, G}(\bm y \mid  \gamma ,\{\bm \alpha_i\})$ where $p_{0,G}$ is the likelihood defined by $[\bm g_{\ell; \gamma}'\bm y_i \mid -] \stackrel{indep}{\sim} N(\alpha_{\ell,i} , \sigma_\epsilon^2)$ for $\ell=1,\ldots,L$.
\end{corollary}
Corollary~\ref{cor-alpha} implies that posterior inference for the parametric factors only requires a simple independent Gaussian likelihood with pseudo-data  $\{\bm g_{\ell; \gamma}'\bm y_i\}$. Hence, the posterior sampling steps for $\{\alpha_{\ell,i}\}$ are often convenient: for example, under the prior $\alpha_{\ell,i} \stackrel{indep}{\sim} N(0, \sigma_{\alpha_\ell}^2)$, the joint full conditional distribution of  $\{\alpha_{\ell_i}\}$ is decomposable into independent univariate Gaussian distributions with $\mathbb{E}(\alpha_{\ell,i}\mid -) = \sigma_{\alpha_\ell}^2/(\sigma_{\alpha_\ell}^2 + \sigma_\epsilon^2)\bm g_{\ell; \gamma}'\bm y_i$ and  $\mbox{Var}(\alpha_{\ell,i}\mid -) = \sigma_\epsilon^2\sigma_{\alpha_\ell}^2/(\sigma_{\alpha_\ell}^2 + \sigma_\epsilon^2)$, which may be sampled efficiently---and is identical for both the PFFM and the SFFM. 
Similar simplifications apply for more complex (e.g., dynamic) models for the parametric factors, and are invariant to the parametrization of the remaining model. 
We emphasize that the constraint on $\bm G_\gamma$  is simply a computationally favorable reparametrization of the parametric term and does not hinder the interpretation: since $\bm G_\gamma^0 \bm \alpha_i^0 = \bm G_\gamma \bm \alpha_i$ for $\bm \alpha_i = \bm R_\gamma \bm \alpha_i^0$ under the QR decomposition, we  recover the parametric factors on the original scale by setting $\bm \alpha_i^0 = \bm R_\gamma^{-1} \bm \alpha_i$, which can be computed draw by draw within the MCMC sampler.

For each nonparametric term $\bm f_k$, we incorporate the \emph{orthogonality} constraints via conditioning on $\bm f_{k'}'\bm f_k = 0$ for all $k' \ne k$ and the \emph{unit-norm} constraint by rescaling $\bm f_k$ appropriately. These steps follow \cite{Kowal2020b} and  are described briefly. 
The linear orthogonality constraint is enforced using a straightforward modification of Lemma~\ref{lem:constraint}, which augments the parametric orthogonality of $\bm G_\gamma$ with the nonparametric terms $\bm f_{k'}$ for $k'\ne k$. The unit-norm constraint is enforced by suitably rescaling $\bm \psi_k$ after sampling; we equivalently rescale $\beta_{k,i}$ to preserve the  the product $\bm f_k \beta_{k,i}$ in the likelihood \eqref{like}. This operation does not change the shape of the curve $f_k$ nor the likelihood \eqref{like}. 
 
 Naturally, Corollary~\ref{cor-alpha} may be adapted for the nonparametric factors:
 \begin{corollary}\label{cor-beta}
Under model \eqref{like} and subject to $\bm G_\gamma'\bm f_k = \bm 0_L$ for all $k$ and $\bm F'\bm F = \bm I_\infty$, the likelihood for the nonparametric factors is proportional to $p(\bm y \mid  \gamma , \{\bm \alpha_i\}, \bm F,  \{\bm \beta_i\})  \propto  p_{1,F}(\bm y \mid  \bm F , \{\bm \beta_i\}) $ where $p_{1,F}$ is the likelihood defined by $[\bm f_k'\bm y_i \mid -] \stackrel{indep}{\sim} N(\beta_{k,i}, \sigma_\epsilon^2)$ for all $k$.
\end{corollary}
As with the parametric factors, Corollary~\ref{cor-beta} implies that posterior inference for the nonparametric factors only requires a simple independent Gaussian likelihood with pseudo-data $\{\bm f_k'\bm y_i\}$. As a result, the SFFM is capable of incorporating  complex prior specifications for $\{\beta_{k,i}\}$---which are essential for inference in the SFFM with the \emph{infinite} expansion  in \eqref{sffm}.

\section{Modeling the nonparametric factors with rank selection}\label{order-prior}
Regularization of the SFFM is critical: the  infinite-dimensional nonparametric  term in \eqref{sffm}  is clearly overparametrized for finite $m$. The nonparametric  components increase model complexity and should be removed whenever the added complexity is not supported by the data. This \emph{rank selection} is complicated by the presence of the parametric template with unknown $\{\alpha_{\ell,i}\}$ and $\gamma$. 
Information criteria offer metrics for model comparisons, but require separate model fits for each rank.
Ordered shrinkage priors (e.g., \citealp{bhattacharya2011sparse})  can reduce sensitivity to the choice of the rank, but often induce undesirable overshrinkage. 
Instead, we design an \emph{ordered spike-and-slab} prior  that applies joint shrinkage and selection to the nonparametric \emph{factors} $\{\beta_{k,i}\}_{i=1}^n$ and therefore removes unnecessary nonparametric terms from the SFFM. Because of the proposed orthogonality constraints on $\{g_\ell\}$ and $\{f_k\}$ and the accompanying simplifications from Corollaries~\ref{cor-alpha}~and~\ref{cor-beta}, we are able to introduce a sophisticated prior for $\{\beta_{k,i}\}$  within the broader SFFM framework of \eqref{sffm}---and do so with minimal impact on computational  cost and MCMC efficiency relative to the PFFM (see Section~\ref{apps}).

\subsection{Ordered spike-and-slab priors}
Let $\theta_k$ denote the parameter that determines the inclusion of the $k$th nonparametric factor $\{\beta_{k,i}\}_{i=1}^n$; the connection is made explicit below. We apply a spike-and-slab prior for  $\{\theta_k\}_{k=1}^\infty$ that imposes a particular \emph{ordering} structure on the spike probabilities $\{\pi_k\}_{k=1}^\infty$:
\begin{align}
\label{csp}
[\theta_k \mid \pi_k] &\sim P_k, \quad P_k = (1-\pi_k) P_{slab} + \pi_k P_{spike} \\
\label{pi-prior}
&\pi_k = \sum_{h=1}^k \omega_h, \quad \omega_h = \nu_h \prod_{\ell=1}^{h-1} (1-\nu_\ell), \quad  \nu_\ell \stackrel{iid}{\sim} \mbox{Beta}(\iota, \iota \kappa),
\end{align}
where $P_{slab}$ is distribution of the slab (active) component and $P_{spike}$ is the distribution of the spike (inactive) component. The  cumulative summation for $\pi_k$ ensures an increasing sequence of (spike) probabilities, $\pi_k < \pi_{k+1}$, that converges: $\lim_{k\rightarrow\infty} \pi_k = 1$. The hyperparameters are determined by the distributions $P_{slab} $ and $P_{spike}$ and the scalars  $\iota,\kappa > 0$.  The prior in \eqref{csp}-\eqref{pi-prior}  generalizes the \emph{cumulative shrinkage process} (CUSP) introduced by \cite{Legramanti2020}, who specified $P_{spike} = \delta_{\theta_\infty}$ to be a point mass at  $\theta_\infty$ with  $\iota = 1$ and fixed $\kappa$ (using $\kappa = 5$ in their examples). The CUSP in \eqref{csp}-\eqref{pi-prior} will be further adapted for the SFFM model \eqref{sffm} in Section~\ref{px-models}.


The ordering of the spike probabilities $\{\pi_k\}_{k=1}^\infty$ in \eqref{pi-prior} implies an ordering for the prior distribution of the parameters $\{\theta_k\}_{k=1}^\infty$ in \eqref{csp}:
\begin{proposition}\label{csp-shrink}
For $\varepsilon>0$ and fixed $\theta_0$, let $\mathbb B_\varepsilon(\theta_0) = \{\theta_k: | \theta_k - \theta_0| < \varepsilon\}$. Prior  \eqref{csp}-\eqref{pi-prior} implies that
$\mathbb{P}(|\theta_k - \theta_0| \le \varepsilon) < \mathbb{P}(|\theta_{k+1} - \theta_0| \le \varepsilon) $
whenever $P_{slab}\{\mathbb B_\varepsilon(\theta_0)\} < P_{spike}\{\mathbb B_\varepsilon(\theta_0)\}$. 
\end{proposition}
Intuitively, for any $\theta_0$ that is favored under the spike distribution $P_{spike}$ relative to the slab distribution $P_{slab}$, the CUSP prior  \eqref{csp}-\eqref{pi-prior} places greater mass around $\theta_0$ as $k$ increases. For spike-and-slab priors, we are most interested in $\theta_0 = 0$. The special case of Proposition~\ref{csp-shrink} with $P_{spike} = \delta_{\theta_\infty}$, $\iota = 1$, and $\theta_0 = 0$ is proved by \cite{Legramanti2020}.

To specify $P_{slab}$ and $P_{spike}$, we apply a normal mixture of inverse-gamma (NMIG) prior:
\begin{equation}
\label{nmig}
[\eta_k |  \theta_k, \sigma_k^2]  \stackrel{indep}{\sim}N(0, \theta_k \sigma_k^2), \quad 
[\theta_k | \pi_k]  \sim (1 - \pi_k) \delta_1 + \pi_k \delta_{v_0}, \quad 
[\sigma_k^{-2}]  \stackrel{iid}{\sim} \mbox{Gamma}(a_1, a_2)
\end{equation}
where the scalars $\{\eta_k\}$ are linked to the factors $\{\beta_{k,i}\}$ in Section~\ref{px-models} and $v_0$, $a_1$ and $a_2$ are hyperparameters. The NMIG prior  incorporates variable selection by assigning the variance scale parameter $\theta_k$ to the slab component, $\theta_k =1$, or the spike component, $\theta_k = v_0$. 
By design, the NMIG prior produces a continuous distribution for the conditional variance of $\eta_k$, which is preferable for variable selection and risk properties \citep{ishwaran2005spike}. The NMIG prior is  often less  sensitive  to hyperparameter choices compared to other spike-and-slab priors
\citep{ishwaran2005spike,Scheipl2012}.

  By coupling the NMIG  prior for $\{\eta_k\}$ in \eqref{nmig}  with the CUSP prior  for $\{\pi_k\}$ in \eqref{pi-prior} and marginalizing over $\{\theta_k, \sigma_k^2,\pi_k\}$, we obtain the following  
   ordering result for the marginal prior on $\{\eta_k\}_{k=1}^\infty$:
   \begin{corollary}\label{eta-shrink}
For $\varepsilon>0$, $\mathbb{P}(|\eta_k| \le \varepsilon) < \mathbb{P}(|\eta_{k+1}| \le \varepsilon)$ whenever $v_0 < 1$.
\end{corollary}
Corollary~\ref{eta-shrink} follows from Proposition~\ref{csp-shrink} by observing that $[\eta_k | \pi_k] \sim (1-\pi_k) t_{2a_1}(0, \sqrt{a_2/a_1}) +\pi_k t_{2a_1}(0, \sqrt{v_0a_2/a_1})$, where $t_d(m, s)$ denotes a $t$-distribution with mean $m$, standard deviation $s$, and degrees of freedom $d$. This mixture is a special case of \eqref{csp}, and the densities of the $t$-distributions  place greater mass near zero when the scale parameter is smaller. Hence, the condition  $v_0 < 1$ ensures that the spike distribution is indeed more concentrated around zero.

For interpretability and efficient MCMC sampling, there is a convenient data augmentation of the CUSP prior \eqref{csp}. Let $z_k \in \{1,\ldots,\infty\}$  denote a categorical variable with $\mathbb{P}(z_k = h | \omega_h) = \omega_h$. The specification
$
[\theta_k | z_k] \sim (1 - \mathbb{I}\{z_k \le k\}) P_{slab} + \mathbb{I}\{z_k \le k\} P_{spike}
$ 
induces \eqref{csp} via marginalization over $z_k$. The number of active (slab) terms is therefore 
$
K^* = \sum_{k=1}^\infty \mathbb{I}\{z_k > k\}, 
$ 
with posterior inference available through the proposed MCMC sampling algorithm (see Algorithm~\ref{alg:MCMC}). For model \eqref{sffm},  $K^*$ is the effective number of nonparametric terms, and therefore is an important inferential target to assess the adequacy of the parametric model. 

The prior expected number of slab terms is $\mathbb{E}(K^*) = \kappa$, which appears in \eqref{pi-prior} and allows for the inclusion of prior information regarding the number of factors. The choice of $\kappa$ can influence the posterior distribution for $K^*$, and thus subsequent inference on the number of nonparametric factors in \eqref{sffm}. To mitigate this effect, we propose the hyperprior 
$ \kappa \sim \mbox{Gamma}(a_{\kappa} , b_{\kappa})$ for
 $a_{\kappa}, b_{\kappa} > 0$, which is conditionally conjugate to \eqref{pi-prior}. The hyperparameters may be selected to provide weak prior information regarding the number of factors: in practice, we set $a_{\kappa} = 2$ and $b_{\kappa}=1$ so that $\mathbb{E}(\kappa) = \mbox{Var}(\kappa) = 2$.  Similarly, we select the default value $\iota = 1$ for simplicity.

CUSPs are intrinsically linked to Indian buffet processes (IBPs; \citealp{Griffiths2011}), which are often used for sparse (non-functional) factor models with rank selection \citep{Ohn2021}. IBPs supply a prior over binary matrices $\{b_{j,k}\}$ with $m$ rows and infinitely many columns, and are  usually applied to $\bm F$ in non-functional versions of \eqref{like}. IBPs are obtained by 
establishing the conditional distributions  
$ [b_{j,k} \mid \mu_k] \stackrel{indep}{\sim}  \mbox{Bernoulli}(\mu_k) $ and  $ [\mu_k] \stackrel{iid}{\sim} \mbox{Beta}(\iota \kappa/K, \iota)$,
 usually with $\iota = 1$, and then integrating over $\{\mu_k\}$ with $K\rightarrow \infty$. IBPs admit a stick-breaking construction  \citep{Teh2007}, which can be mapped to the CUSP probability sequence $\{\pi_k\}_{k=1}^\infty$:
\begin{proposition}\label{ibp-connect}
The CUSP \eqref{pi-prior} satisfies $(1 - \pi_k) = \mu_{(k)}$ where $\mu_{(1)} > \cdots > \mu_{(K)}$ are the ordered (slab) probabilities from the \cite{Teh2007} stick-breaking construction of the IBP, i.e., $\mu_{(k)} = \prod_{\ell = 1}^k \nu_\ell'$ with   $\nu_\ell' \stackrel{iid}{\sim} \mbox{Beta}(\iota \kappa, \iota).$
\end{proposition}
Conventionally, IBPs apply the multiplicative beta process to the slab probability $(1-\pi_k)$, which explains the complement in Proposition~\ref{ibp-connect}. 

Despite these fundamental stick-breaking connections, CUSPs and IBPs remain distinct due to \eqref{csp}: CUSPs define a spike-and-slab prior for a sequence of parameters $\{\theta_k\}_{k=1}^\infty$, while IBPs define a prior  over binary matrices. In the context of \eqref{sffm}, elementwise sparsity of $\bm F$ is unwarranted: although we are interested in rank selection, we prefer priors that encourage smoothness rather than sparsity in $f_k(\tau)$ over $\tau \in \mathcal{T}$.

\subsection{Consistent rank selection}\label{consist}
To motivate  the use of  this  prior within the SFFM, we investigate the asymptotic behavior of the generalized CUSP as a  standalone prior for rank selection. Specifically, we consider the  modified setting with
\begin{equation}
y_{i} = \eta_{0i} +  \epsilon_{i}, \quad {\epsilon_{i}} \stackrel{iid}{\sim} N(0,1), 
\label{eq:truedistribution}
\end{equation}
where $\eta_{01}, \dots, \eta_{0n}$ are ``true'' mean parameters. This setting is applicable to the SFFM via Corollary~\ref{cor-beta} and is related to the problem of estimating a high-dimensional mean vector from a single multivariate observation \citep{CvdV,rock}, but with ordered sparsity  $\eta_{0k} \neq 0$ for $k \leq K_{0n}$ and  $\eta_{0k} = 0$ for $k > K_{0n}$. As such, $K_{0n}$ is the true number of nonzero means and corresponds to the true rank in model \eqref{sffm}. 
As customary in posterior asymptotics, we allow $K_{0n}$ to grow with $n$.

We study the posterior of $\eta_i$ under the CUSP \eqref{csp}-\eqref{pi-prior} and NMIG \eqref{nmig} prior. A key term is the ``remainder" $R_n = \sum_{k\geq  K_{0n}} \omega_k$, which  represents an upper bound for the probability mass assigned to the slab for the parameters that are null. 
Under the CUSP, we show that $R_n$ is far from zero with vanishing probability  \emph{a priori}:

\begin{lemma}
	\label{lem:priorR}
	Let $\varepsilon_n\to 0$ with $\varepsilon_n^{1/K_{0n}} > \kappa/(\kappa+1)$. For the CUSP prior  \eqref{csp}-\eqref{pi-prior} and a positive constant  $C>1$, the remainder term $R_n = \sum_{k\geq  K_{0n}} \omega_k$ satisfies
$	\mathbb{P}(R_{n} >\varepsilon_n) \leq \exp(-C K_{0n}).$
\end{lemma}
	
According to Lemma~\ref{lem:priorR}, the prior probability that $R_n$ exceeds a small threshold is exponentially small. Such exponential decay in the prior is  essential to obtain optimal  posterior behavior in sparse settings \citep{CvdV}. This result does not require the NMIG prior  \eqref{nmig}.

To connect the prior behavior in Lemma~\ref{lem:priorR} to the posterior, first consider the following mild conditions on the prior parameters:
(C1) $a_1/(a_2v_0)	> n^2/K_{0n}$;  
(C2) $AK_{0n}/n < 1/2$ with $A>1/2$;
(C3) $a_1>1/2$, and $a_1 \leq a_2$. 
Condition (C1)  implies  that  the precision of the spike component is increasing with $n$; condition (C2) requires that  $K_{0n}$ is strictly less than the total number of elements (or the truncation limit); and condition (C3) implies that the values of the $\sigma_k^{2}$ parameters in \eqref{nmig} are not too large.  Now let $\mathbb{P}_0$ and $\mathbb{E}_0$  denote the probability and the expectation, respectively, under the true distribution of the data \eqref{eq:truedistribution}. Under (C1)-(C3),  the posterior probability that the remainder exceeds a shrinking $\varepsilon_n$, for $n\to \infty$,  goes to zero in expectation $\mathbb{E}_0$. This is formalized in the following result:
\begin{theorem}
Let $\varepsilon_n\to 0$ with $\varepsilon_n^{1/K_{0n}} > \kappa/(\kappa+1)$ and
	assume (C1)-(C3) and $C >2Ae$. For the CUSP  \eqref{csp}-\eqref{pi-prior} and NMIG \eqref{nmig} priors, the posterior distribution satisfies 
		\[
	\lim_{n \to \infty} \mathbb{E}_0 \left\{ \mathbb{P} \left( R_n> \varepsilon_n \mid y_1,\ldots, y_n \right) \right\} = 0.
\]
\end{theorem}
In the context of the SFFM, \eqref{eq:truedistribution} can be obtained by fixing $\bm F$  and setting $\beta_{k,i} = \eta_{k,i}$. For simplicity, set the observation error variance in \eqref{sffm} to be $\sigma_\epsilon^2 =n$, which grows with the number of curves. Applying  Corollary~\ref{cor-beta},  the likelihood for $\{\eta_k\}$ in \eqref{eq:truedistribution} is equivalent to the likelihood implied by $y_{k} = \bar{\eta}_k +  \epsilon_{k}$ with $\epsilon_k \stackrel{iid}{\sim}N(0,1)$, where $y_{k} = n^{-1} \sum_{i=1}^n \bm f_k'\bm y_i$  and  $\bar{\eta}_k = n^{-1}\sum_{i=1}^n \eta_{k,i}$. The above results then apply for the NMIG prior on $\{\bar{\eta}_k\}$ and by interchanging the notation $k$ for $n$.


\subsection{Parameter-expanded spike-and-slab models}\label{px-models}
The generalized CUSP  provides an appealing mechanism for introducing ordered shrinkage and selection into model \eqref{sffm}.  A subtle yet important feature of this framework is that spike-and-slab selection is applied jointly to all $n$ factors $\{\beta_{k,i}\}_{i=1}^n$. 
However, parametrization of a \emph{collection} of variables via a spike-and-slab prior on a \emph{scalar} $\eta_k$ is known to introduce significant challenges for MCMC sampling.   In a related setting,  \cite{Scheipl2012}  showed that as the size of this collection   increases (here, $n$), it becomes more difficult for the MCMC algorithm to switch between the slab and spike components. The consequences are nontrivial: without the ability to traverse between the slab and spike components, the uncertainty quantification for each factor's activeness is invalidated, and inference for the effective number of nonparametric factors $K^*$ is not reliable.

To illustrate the implications, we implemented the CUSP prior and sampling algorithm from  \cite{Legramanti2020},  which does not address these issues. 
The analogous prior in \cite{Legramanti2020}  is $[\beta_{k,i} \mid \theta_k] \stackrel{indep}{\sim} N(0, \theta_k)$ and a CUSP prior for $\{\theta_k\}_{k=1}^\infty$ with $P_{spike} = \delta_{\theta_\infty}$ and $P_{slab} = \mbox{Inv-Gamma}(a_\theta, b_\theta)$. We use the hyperparameters from \cite{Legramanti2020}: $\theta_\infty = 0.05$, $a_\theta = b_\theta = 2$, and $\kappa=5$, but the subsequent results persist for other choices. For this example, we use the simulation design from Section~\ref{sims} with $K_{true} = 6$ nonparametric factors. Figure~\ref{fig:postk} (left) displays the MCMC samples of $K^*$ across five different MCMC runs with random initializations. 
 The \cite{Legramanti2020} prior and sampler produces MCMC draws that are essentially stuck at a single value: $K^*$ changes values in only 0.004\% of the MCMC iterations. Clearly, this result cannot be used to provide reliable estimation or  uncertainty quantification for the effective number of nonparametric components $K^*$.

 \begin{figure}[h]
\begin{center}
\includegraphics[width=.49\textwidth]{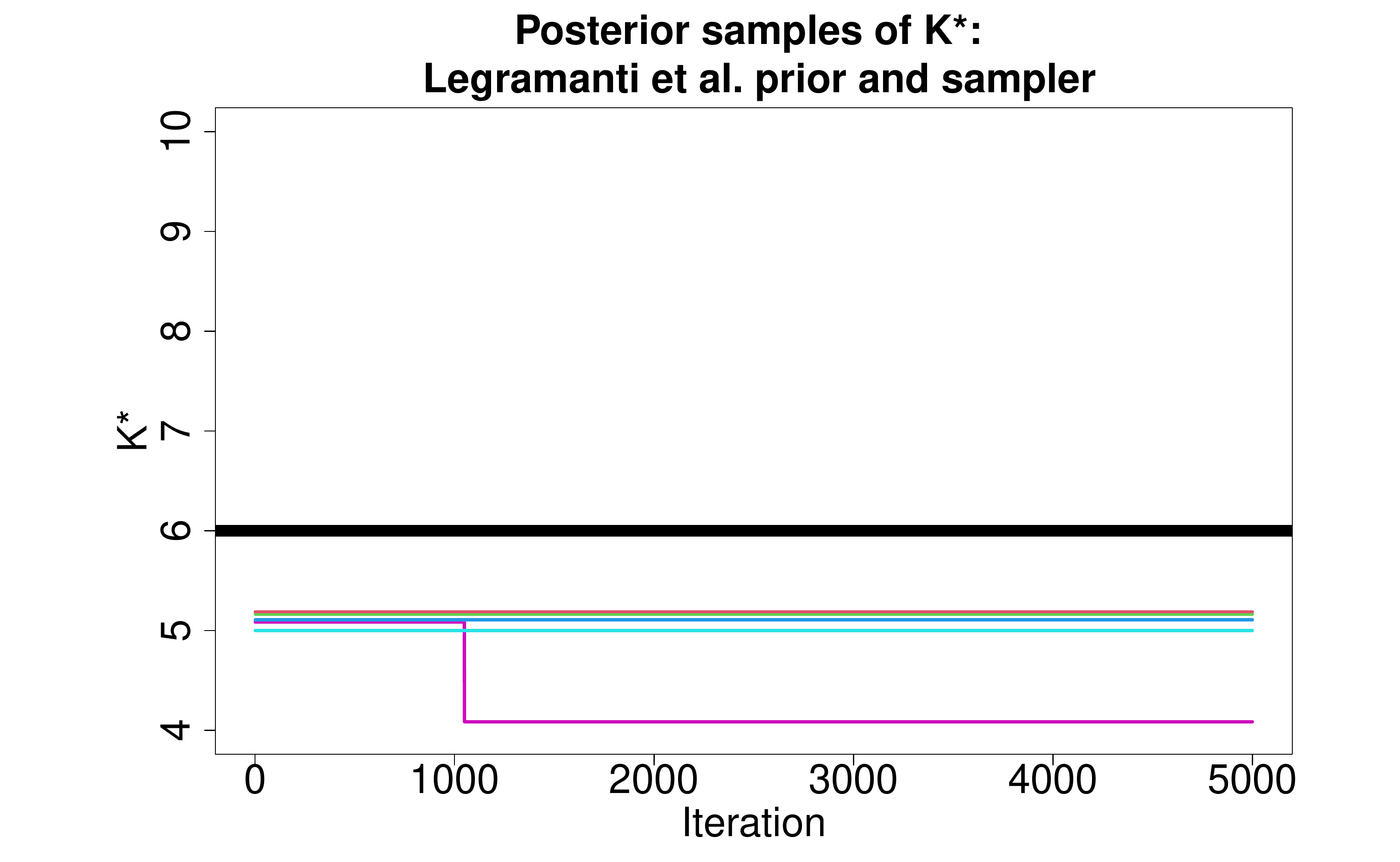}
\includegraphics[width=.49\textwidth]{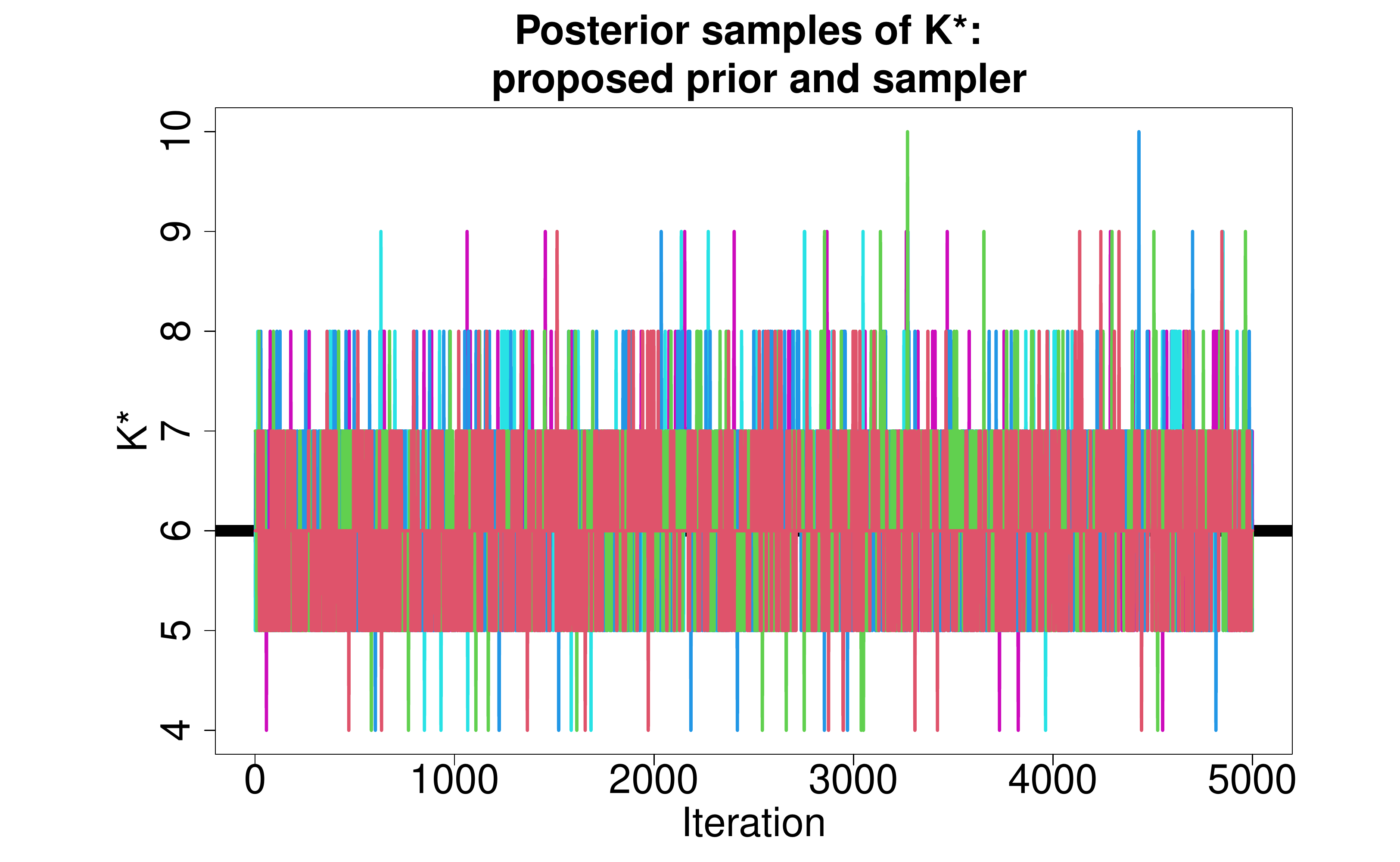}
\caption{\small MCMC samples of $K^*$ across 5 chains for \cite{Legramanti2020} (left) compared  to the proposed prior and algorithm (right). The horizontal black line denotes the true value $K_{true} = 6$. The proposed approach provides substantial improvements in mixing and produces more accurate estimates. 
\label{fig:postk}}
\end{center}
\end{figure}

Our solution is to introduce a redundant parameter expansion for the ordered spike-and-slab prior. 
We induce a prior on each factor $\beta_{k,i}$ via the following parameter expansion:
\begin{equation} \label{px}
\beta_{k,i} = \eta_k \xi_{k,i}, \quad [\xi_{k,i} \mid m_{\xi_{k,i}}] \stackrel{indep}{\sim}N(m_{\xi_{k,i}}, 1), \quad 
m_{\xi_{k,i}}  \stackrel{indep}{\sim} \frac{1}{2}\delta_1 + \frac{1}{2} \delta_{-1}
\end{equation}
where $\eta_k$ applies the ordered spike-and-slab shrinkage for factor $k$ via the NMIG prior \eqref{nmig} and $\{\xi_{k,i}\}_{i=1}^n$ disperses this shrinkage throughout $\{\beta_{k,i}\}_{i=1}^n$. Here, $\xi_{k,i}$ is nonidentifiable and exists primarily for the purpose of improving MCMC performance.   The Gaussian parameter expansion of $\xi_{k,i}$ is centered around $1$ or $-1$ with variance $1$. As a result, the parameter $\eta_k$ retains the interpretability as a selection parameter and assigns the magnitude to $\{\beta_{k,i}\}_{i=1}^n$. 


The impact of the parameter expansion is striking: Figure~\ref{fig:postk} (right)  shows  the MCMC samples of $K^*$ from the proposed method (again across five different MCMC runs with random initializations). The proposed approach produces draws that are centered at the true value with some variability: $K^*$ changes values in 32\% of the MCMC iterations,  which is a significant improvement over  \cite{Legramanti2020}.  Clearly, the MCMC-based uncertainty quantification from the proposed approach is more plausible---which is essential for inference on effective number of nonparametric terms in the SFFM \eqref{sffm}.

The parameter-expanded prior \eqref{px} also maintains the required ordering on the factors:
\begin{proposition}\label{nmig-px}
For $\varepsilon>0$, $\mathbb{P}(|\beta_{k,i}| \le \varepsilon) < \mathbb{P}(|\beta_{k+1,i}| \le \varepsilon)$ whenever $v_0 < 1$.
\end{proposition} 
Proposition~\ref{nmig-px} requires only that $\{\xi_{k,i}\}$ are marginally independent and identically distributed. Therefore, other parameter expansion schemes may be substituted in \eqref{px}  without disrupting the probabilistic ordering of $|\beta_{k,i}|$. By a similar argument, the ordering also is preserved across observations: $\mathbb{P}(|\beta_{k,i}| \le \varepsilon) < \mathbb{P}(|\beta_{k+1,i'}| \le \varepsilon)$ for $\varepsilon>0$ and $v_0 < 1$.

\section{MCMC for posterior inference}\label{mcmc}
We design an efficient MCMC algorithm for posterior inference. The proposed Gibbs sampler is both computationally scalable and MCMC efficient (see Section~\ref{apps}) due to two key features: the orthogonality constraints on $\{g_\ell\}$ and $\{f_k\}$, which produce important model simplifications, and the parameter-expanded sampler for the ordered spike-and-slab prior, which substantially improves MCMC mixing (see Figure~\ref{fig:postk}). We implement the following sampling steps, with details provided in the supplement:  (i) imputation  of $y_i(\tau^*)$  for any unobserved $\tau^*$; (ii) sampling the  nonlinear parameter $\gamma$ if unknown; (iii) sampling the constrained nonparametric factors $\{f_k\}$ (see Section~\ref{all-np}) along with the smoothing parameters $\{\lambda_{f_k}\}$; (iv) sampling the ordered spike-and-slab parameters (Algorithm~\ref{alg:MCMC}); (v) sampling the parametric factors $\{\alpha_{\ell,i}\}$ using the simplifications from Corollary~\ref{cor-alpha}; and (vi)  sampling any variance components, such as $\sigma_\epsilon^2$.

The implementation incorporates a truncation  of the infinite summation in \eqref{sffm}, which provides simpler and faster computations. The following result shows that finite approximations are accurate for sufficiently large truncation $K$:
\begin{proposition}\label{truncate}
Let $\theta^{(K)} =  \{\theta_k\}_{k=1}^K$ denote the sequence $\{\theta_k\}_{k=1}^\infty$ truncated at $K$. Under the CUSP \eqref{csp}-\eqref{pi-prior} and NMIG \eqref{nmig} priors with $0 < v_0 < \varepsilon < 1$, we have $\mathbb{P}\left\{d_\infty(\theta, \theta^{(K)}) > \varepsilon\right\}  \le \kappa \{\kappa/(1+\kappa)\}^K$. 
\end{proposition}
The approximation error induced by truncating $\{\theta_k\}_{k=1}^\infty$ to $K$ terms decreases rapidly in $K$, which suggests that the proposed infinite-dimensional ordered spike-and-slab prior is accurately approximated by a conservative truncation.

We emphasize that the  sampling steps for the ordered spike-and-slab parameters, detailed in Algorithm~\ref{alg:MCMC}, consist of simple, fast, and closed form updates. Because of the orthogonality constraints and Corollary~\ref{cor-beta}, these distributions depend on the data only through $\{\bm f_k'\bm y_i\}$.  In addition, we retain posterior samples of $K^* = \sum_{k=1}^K \mathbb{I}\{z_k > k\}$, which is the effective number of nonparametric terms. If the posterior distribution of $K^*$ places mass on values near $K$, then the truncation level $K$ should be increased. Adaptations of this sampling algorithm for use of the ordered spike-and-slab prior with other models, such as (non-functional) factor models, are available by replacing $y_{k,i}^F$ with the appropriate term.

\begin{algorithm}[h]
\SetAlgoLined  \small
Let $y_{k,i}^F = \bm f_k'\bm y_i$ for $i=1,\ldots,n$ and $k=1,\ldots,K$:
\begin{enumerate}
\item {\bf Sample} $[m_{\xi_{k,i}} \vert -]$ from $\{-1,1\}$ with $\mathbb{P}(m_{\xi_{h,i}} = 1 \vert -) = 1/\{1 + \exp(-2\xi_{h,i})\}$; 

\item {\bf Sample} $[\xi_{k,i} \vert -] \sim N(Q_{\xi_{k,i}}^{-1} \ell_{\xi_{k,i}},   Q_{\xi_{k,i}}^{-1})$ where $Q_{\xi_{k,i}} = \eta_k^2/\sigma_\epsilon^2 + 1$ and $\ell_{\xi_{k,i}} = \eta_k y_{k,i}^F/\sigma_\epsilon^2 + m_{\xi_{k,i}} $;

\item {\bf Sample} $[\eta_{k} \vert -] \sim N(Q_{\eta_k}^{-1} \ell_{\eta_k},   Q_{\eta_k}^{-1} )$ where $Q_{\eta_k} = \sum_{i=1}^n \xi_{k,i}^2/\sigma_\epsilon^2 + (\theta_k \sigma_k^2)^{-1}$ and $\ell_{\eta_k} = \sum_{i=1}^n \xi_{k,i} y_{k,i}^F/\sigma_\epsilon^2$;

\item {\bf Rescale} $\eta_k \rightarrow (\sum_{i=1}^n \vert\xi_{k,i}\vert/n) \eta_k$
 and $\bm \xi_k \rightarrow (n/\sum_{i=1}^n \vert\xi_{k,i}\vert) \bm \xi_k$ 
 and {\bf update}   $\beta_{k,i} = \xi_{k,i}\eta_k$;

\item {\bf Sample} $[\sigma_{k}^{-2} \mid -] \sim \mbox{Gamma}\left\{a_1 + 1/2, a_2 + \eta_k^2/(2\theta_k)\right\}$; 

\item {\bf Sample}  $[\nu_k \mid -] \sim \mbox{Beta}(1 + \sum_{h=1}^K \mathbb{I}\{z_h = k\}, \kappa + \sum_{h=1}^K \mathbb{I}\{z_h > k\})$ for $k=1,\ldots,K-1$ 
and {\bf update} $\pi_k$ and $\omega_k$ from \eqref{pi-prior};

\item {\bf Sample} $[\kappa \mid -] \sim \mbox{Gamma}\{a_{\kappa}+ K-1, b_{\kappa} - \sum_{k=1}^{K-1} \log(1 - \nu_k)\}$;

\item {\bf Sample} $[z_k \mid -]$ from
$$
\mathbb{P}(z_k=h \mid -) 
\propto 
\begin{cases}
\omega_h t_{2a_1}( \eta_k; 0, \sqrt{v_0a_2/a_1})&  h \le k \\
\omega_h t_{2a_1}( \eta_k; 0, \sqrt{a_2/a_1}) & h > k 
\end{cases}
$$
where $t_d(x; m, s)$ is the density of the $t$-distribution evaluated at $x$ with mean $m$, standard deviation $s$, and degrees of freedom $d$;

\item {\bf Update} $\theta_k =  1$ if $z_k >k$ and $\theta_k = v_0$ if $z_k \le k$.


\end{enumerate}
 \caption{MCMC sampling steps for the ordered spike-and-slab prior} \label{alg:MCMC}
\end{algorithm}


 \section{Simulation study}\label{sims}
 A simulation study is performed to assess model performance for point and interval estimation of $Y_i$, $\bm y_i$, and $\{\bm \alpha_i\}$ and uncertainty quantification for the number of nonparametric terms.  We focus on the linear template $\mathcal{H}_0 = \mbox{span}\{1, \tau\}$ and present results for the \cite{nelson1987parsimonious} template with unknown $\gamma$ in the supplementary material. 

Synthetic functional data with $n=100$ curves and $m=25$ equally-spaced observation points in $[0,1]$ are generated as follows.  The parametric and nonparametric factors are simulated as $\alpha_{\ell,i}^* \stackrel{iid}{\sim}N(0,1)$ and $\beta_{k,i}^* \stackrel{indep}{\sim}N(0, 1/(k+1)^2)$, respectively, for $k=1,\ldots, K_{true}$ and  $K_{true} \in \{0,1,3, 8\}$; results for $K_{true} = 1$ are nearly identical to the $K_{true} = 3$ case and are omitted, while $K_{true} = 8$ is presented in the supplement. By design, the variability in the parametric factors outweighs the variability in the nonparametric factors. The parametric basis matrix $\bm G_\gamma^*$ is constructed by evaluating  $\{g_\ell\}$ at each observation point and QR-decomposing the resulting matrix as in Section~\ref{all-np}. For the nonparametric functions $f_k^*$, we use orthogonal polynomials of degree $k+1$, which are orthogonal to the linear template; the Nelson-Siegel version applies an additional orthogonalization step. The  error-free latent functions are $\bm Y_i^* = \bm G_\gamma^* \bm \alpha_i^* + \bm F^* \bm \beta_i^*$ and the functional observations are generated as $\bm y_i = \bm Y_i^* + \sigma^* \bm \epsilon_i^*$ where $\sigma^* = \mbox{sd}(\bm Y_i^*)/\mbox{RSNR}$ for sample standard deviation $\mbox{sd}(\cdot)$, root signal-to-noise ratio $\mbox{RSNR} = 3$, and $\bm\epsilon_i^*\stackrel{iid}{\sim}N(\bm 0, \bm I_m)$. This process was repeated to create 100 synthetic datasets. 
 
We focus primarily on the PFFM, PFFM+gp, and SFFM. 
 For all models, we assume the conditionally Gaussian likelihood \eqref{like} with the hierarchical priors
 \begin{equation}
 \label{priors-alpha}
[ \alpha_{\ell,i} \mid \sigma_{\alpha_\ell}] \stackrel{indep}{\sim} N(0, \sigma_{\alpha_\ell}^2), \quad \sigma_{\alpha_\ell} \stackrel{iid}{\sim} C^+(0,1), \quad p(\sigma_\epsilon^2) \propto 1/\sigma_\epsilon^2,
 \end{equation}
 and the SFFM hyperparameters $a_1 = 5$,  $a_2 = 25$, and $v_0 = 0.001$, with an upper bound $K=10$ on the number of nonparametric factors. 
 Sensitivity analyses were conducted for $(a_1, a_2) \in \{(5,25), (5,50), (10, 30)\}$ and $v_0 \in \{0.01, 0.005, 0.00025\}$.  The results for point and interval predictions and estimates of $\bm Y_i^*$ and $\{\alpha_{\ell,i}^*\}$ are highly robust to these hyperparameters. Inference on $K^*$ is typically robust with the exception of $v_0 = 0.00025$, for which the posterior of $K^*$ becomes more sensitive to $(a_1,a_2)$. Hence, we select a larger value of $v_0$.

 

First, we evaluate point prediction  of the  latent  curves $\bm  Y_i^*$ and point estimation of the parametric factors $\{\alpha_{\ell,i}^*\}$. Posterior expectations from each model are assessed using root mean squared (prediction) error, and presented in  Figure~\ref{fig:linear-rmse}. Although the methods perform similarly when the parametric model is true ($K_{true} = 0$), the SFFM offers  substantial improvements in the semiparametric setting ($K_{true} > 0$). Notably, when $K_{true} > 0$, the PFFM+gp offers better prediction than the  PFFM---as expected---but conversely cannot estimate the parametric factors accurately.

  \begin{figure}[h]
\begin{center}
\includegraphics[width=.49\textwidth]{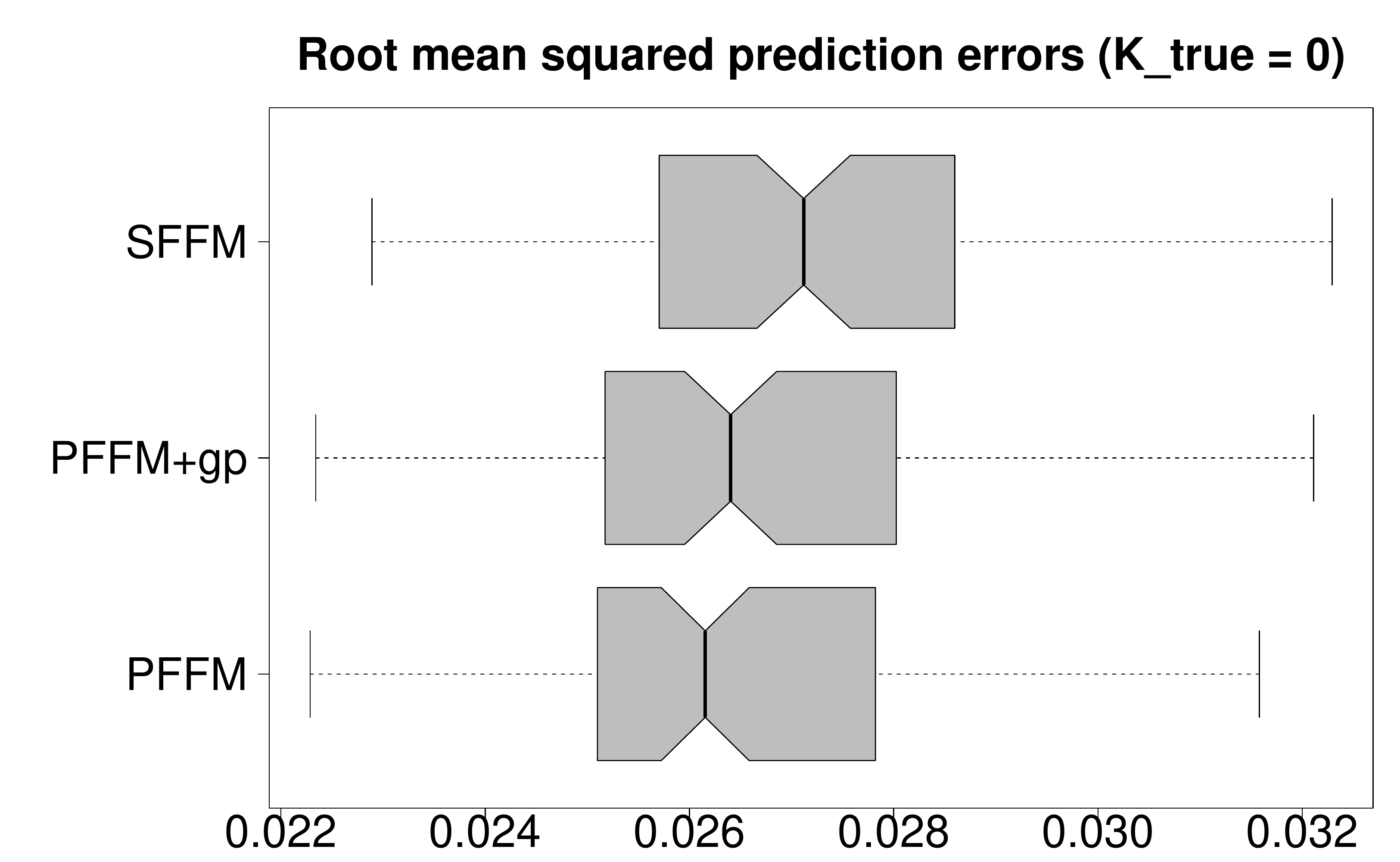}
\includegraphics[width=.49\textwidth]{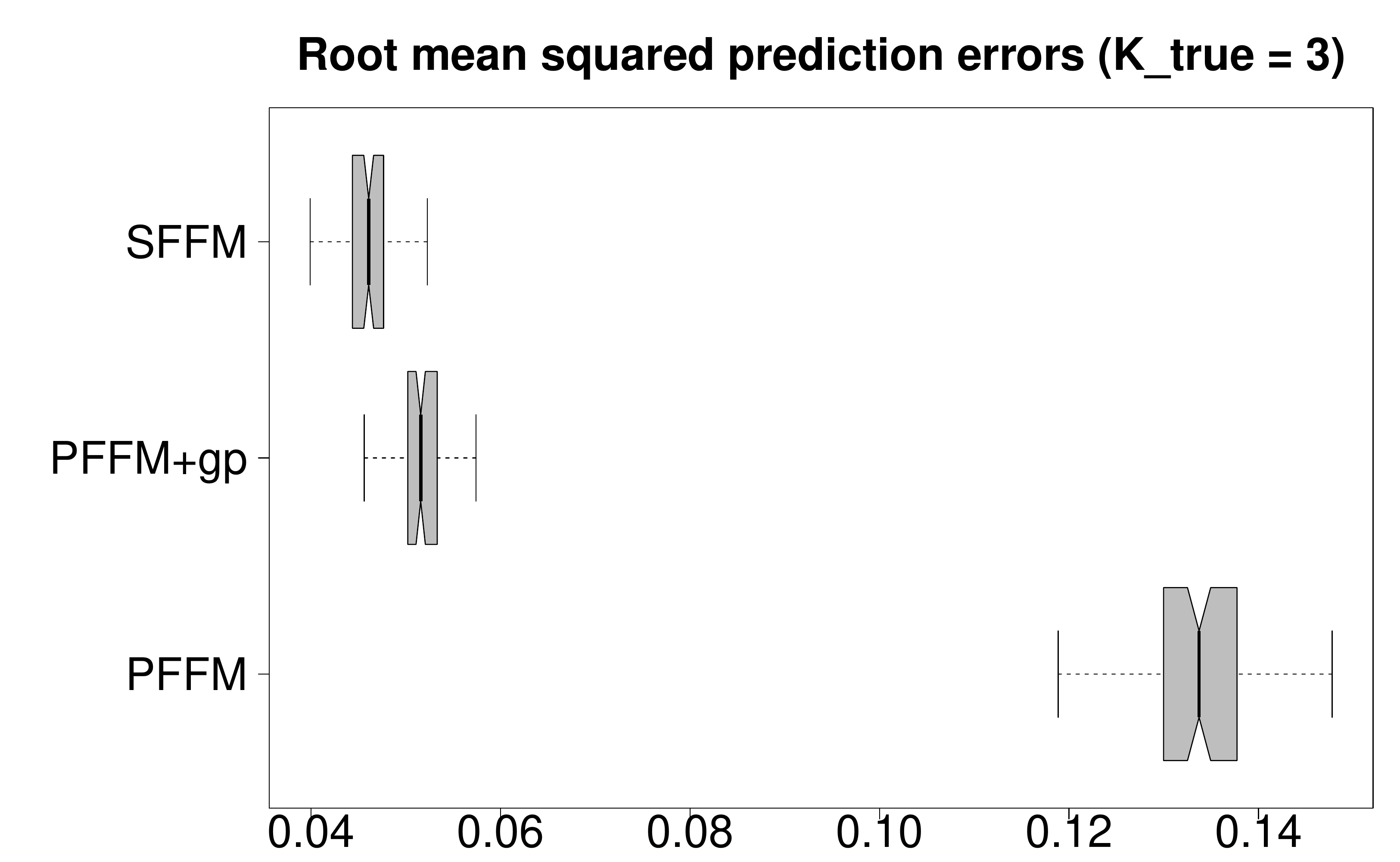}
\includegraphics[width=.49\textwidth]{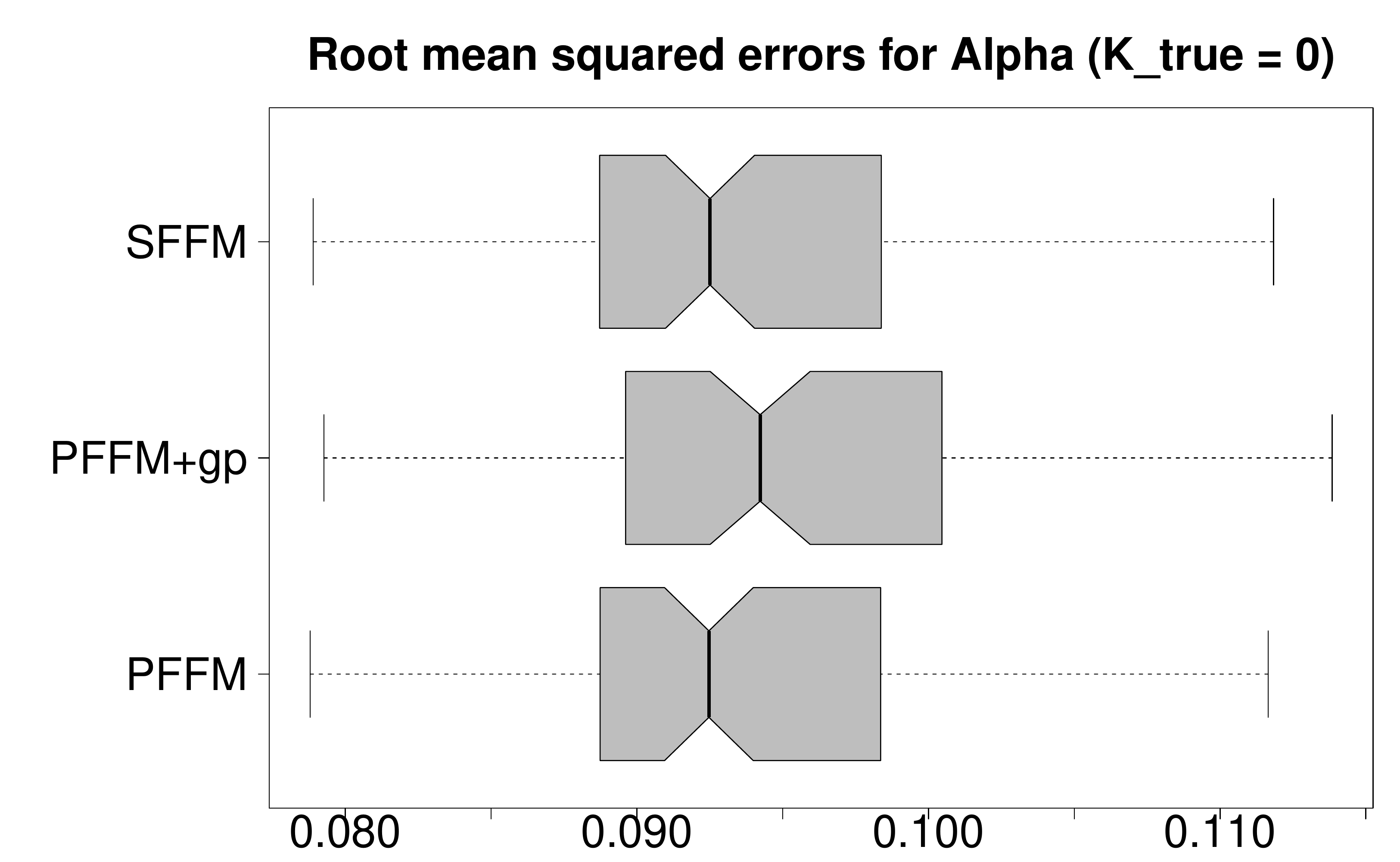}
\includegraphics[width=.49\textwidth]{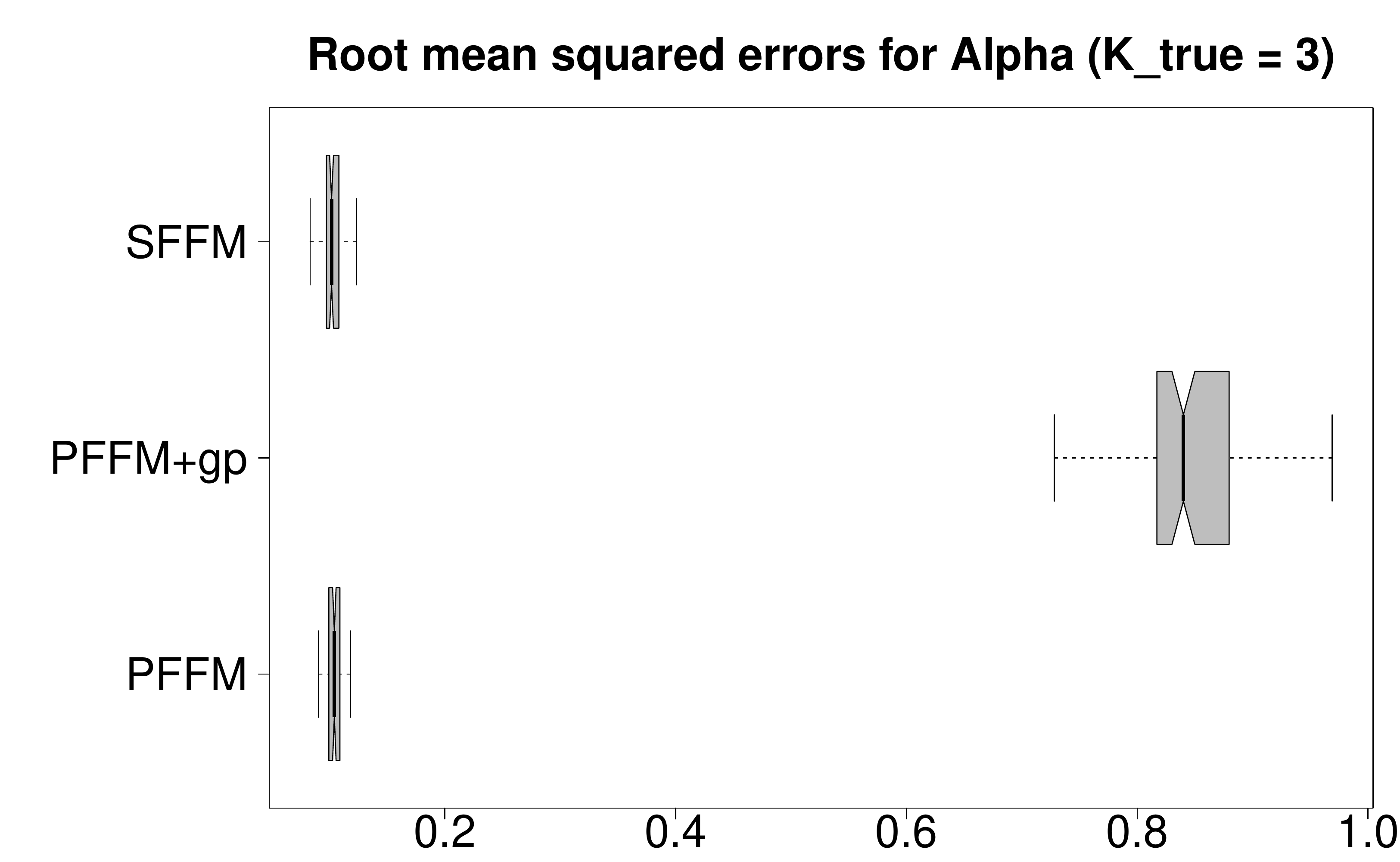}
\caption{\small Root mean squared (prediction) errors for $\bm Y^*$ (top) and   $\{\alpha_{\ell, i}^*\}$   (bottom) for $K_{true} = 0$ (left) and $K_{true} = 3$ (right). All methods perform similarly for $K_{true} = 0$, while the SFFM outperforms all competitors when $K_{true}  >  0$. In this setting, the PFFM+gp offers better prediction than the  PFFM, but cannot accurately estimate the parametric factors.
\label{fig:linear-rmse}}
\end{center}
\end{figure}

Next, we assess  prediction and credible intervals for $\bm y_i^*$ and $\{\alpha_{\ell,i}^*\}$, respectively, where  $\bm y_i^*$ is distributed identically to $\bm y_i$. The 95\% interval estimates are  evaluated using mean interval width and empirical coverage, which are presented in Figure~\ref{fig:linear-int}.  Most notably, the SFFM provides uniformly narrower interval estimates---and therefore more precise uncertainty quantification---with approximately the correct nominal coverage. The PFFM prediction and credible interval estimates are excessively wide when $K_{true} > 0$, which is reflected in Figure~\ref{fig:fit} and subsequently in Figure~\ref{fig:sd}. The PFFM+gp interval estimates  for $\{\alpha_{\ell,i}^*\}$ are quite poor, especially for $K_{true} > 0$, where the interval estimates are excessively wide yet still fail to contain the true parameters  $\{\alpha_{\ell,i}^*\}$ at close to the nominal level. In conjunction with Figure~\ref{fig:alphapm}, these results show that functional confounding adversely impacts both point estimation and uncertainty quantification.

  \begin{figure}[h]
\begin{center}
\includegraphics[width=.49\textwidth]{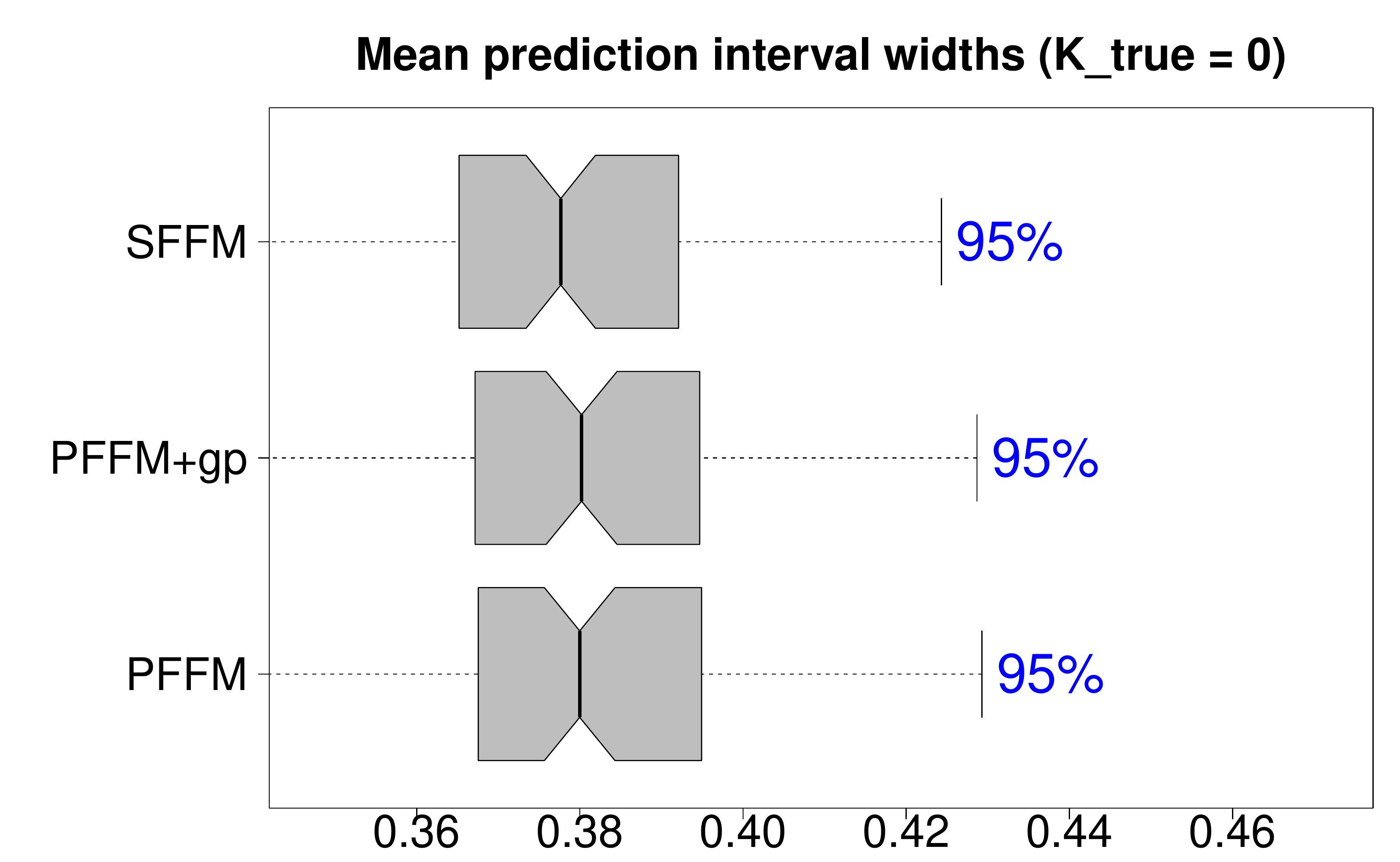}
\includegraphics[width=.49\textwidth]{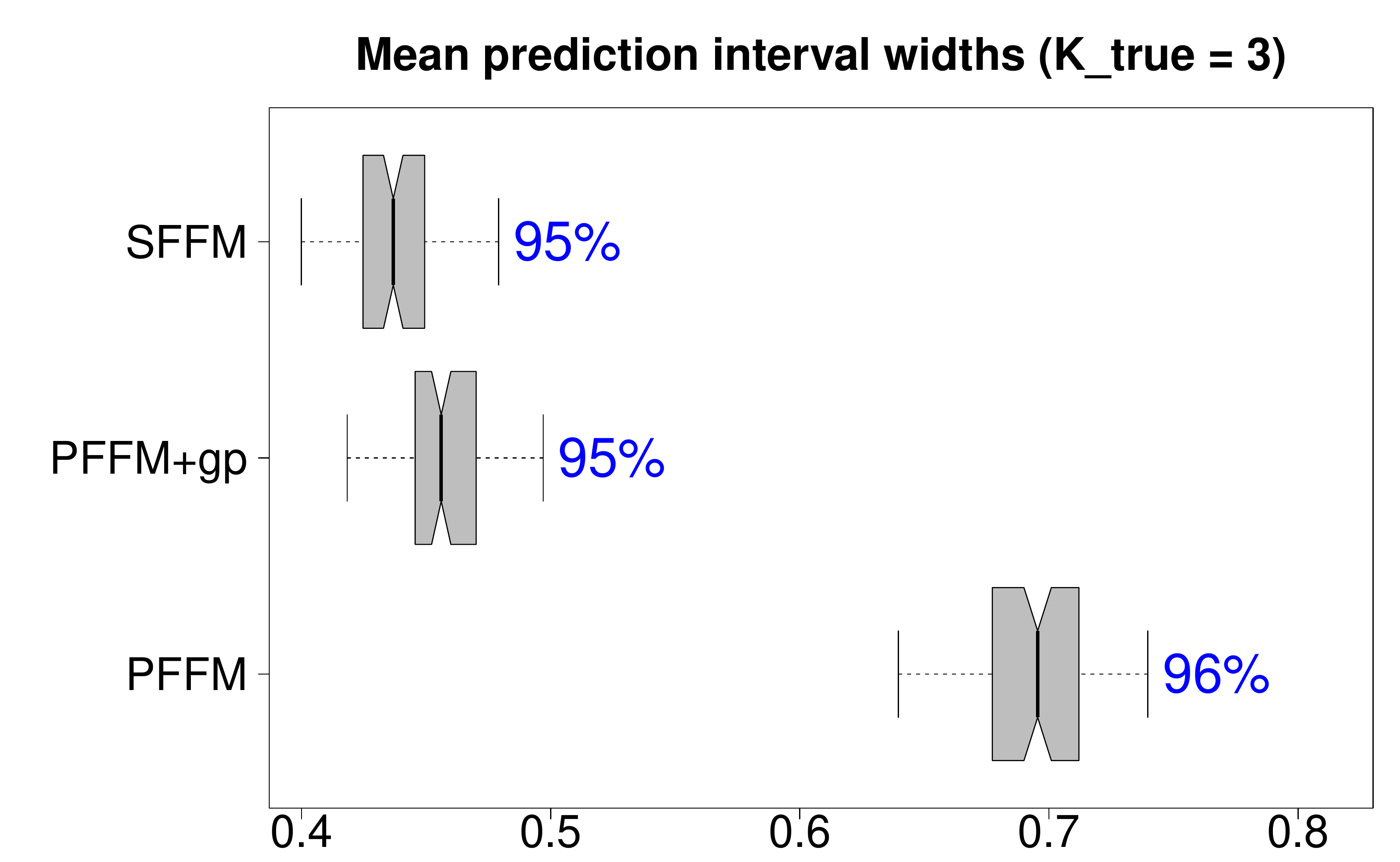}
\includegraphics[width=.49\textwidth]{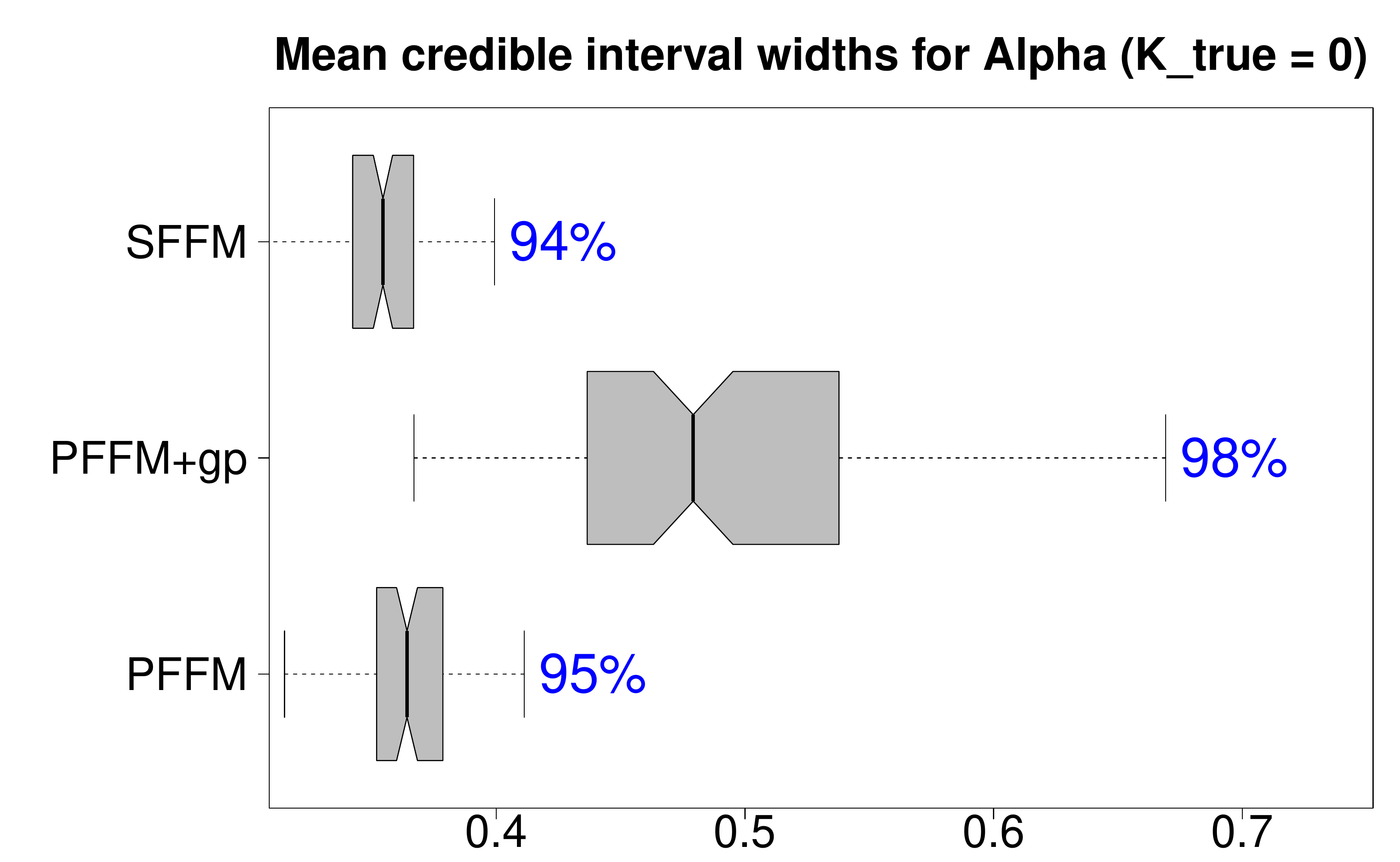}
\includegraphics[width=.49\textwidth]{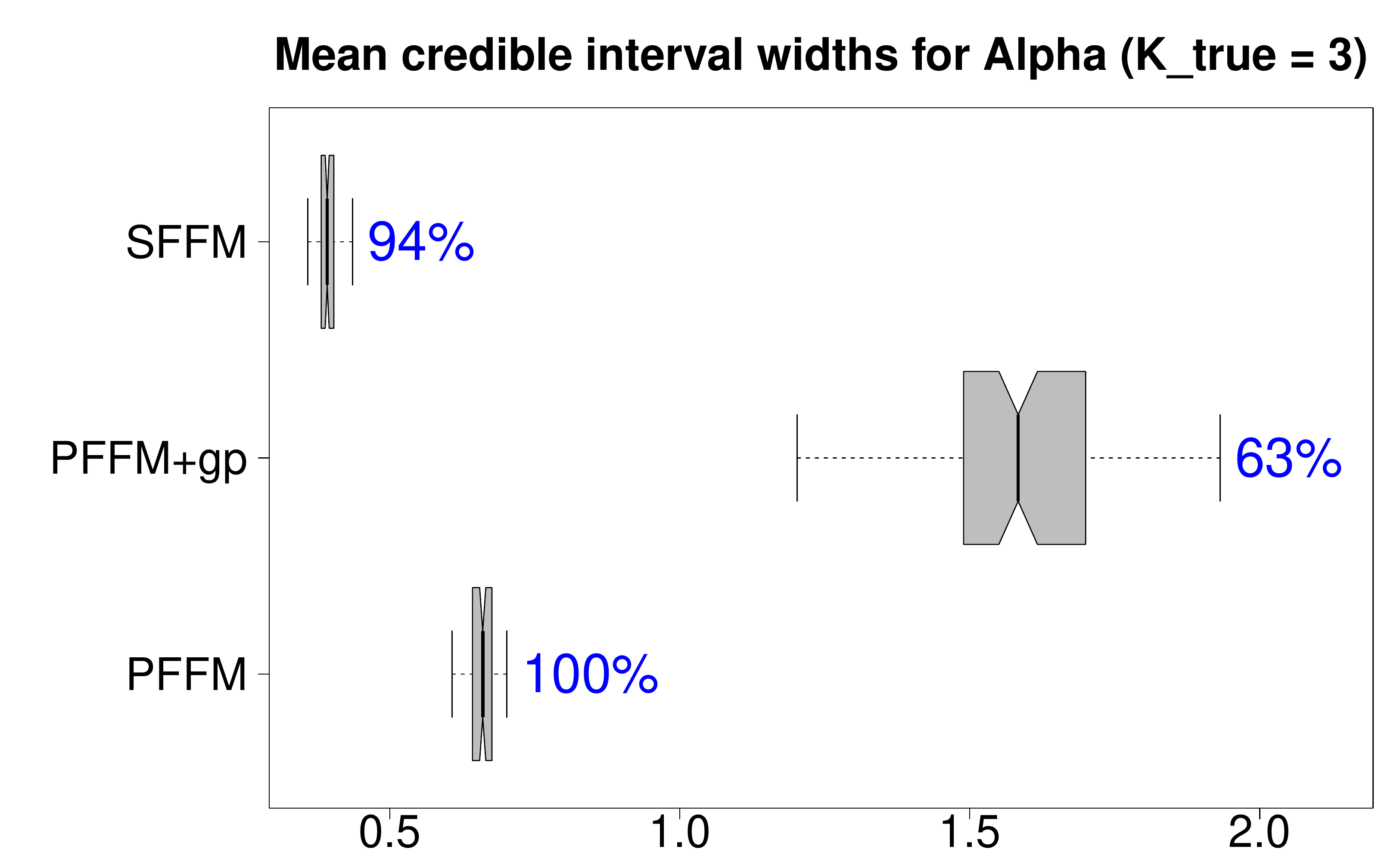}
\caption{\small Mean 95\% prediction (top) and credible (bottom) interval widths  for $\bm y_i^*$ and $\{\alpha_{\ell,i}^*\}$, respectively, with empirical coverage (annotations) for $K_{true} = 0$ (left) and $K_{true} = 3$ (right). The SFFM provides narrow interval estimates with (nearly) the correct nominal coverage in all cases. The PFFM intervals are excessively wide for $K_{true} > 0$, while the PFFM+gp performs poorly for $\{\alpha_{\ell,i}^*\}$ in all cases. 
\label{fig:linear-int}}
\end{center}
\end{figure}

Lastly,  we study the posterior distribution of $K^* = \sum_{k=1}^K \mathbb{I}\{z_k > k\}$ for estimating $K_{true}$ and  
quantifying the necessity of the nonparametric terms. Since the PFFM and PFFM+gp are not competitors for rank selection, we include a  finite mixture model variant of the SFFM (SFFM-fmm) that instead assumes $\omega_{1:K} \sim  \mbox{Dirichlet}(\kappa/K, \ldots, \kappa/K)$. We use the same choice of $K$ as in the SFFM and  set the prior expected number of factors to be $\kappa=1$ to encourage fewer factors.
Figure~\ref{fig:kprob} compares the quality of the posterior distribution $\mathbb{P}(K^*\mid \bm y)$ for the SFFM and the finite mixture alternative. Notably, the ordered spike-and-slab prior results in much larger probabilities on the true rank $K_{true}$. Despite fixing $\kappa = 1$, the SFFM-fmm  repeatedly overestimates the ranks: the simulation average of $\mathbb{P}(K^* > K_{true} \mid \bm y)$ is 0.30 for $K_{true} = 0$, 0.23 for $K_{true} = 3$, and 0.22 for $K_{true} = 8$, while the comparable values for the SFFM are 0.15, 0.16,  and 0.15, respectively. Hence, the finite mixture alternative introduces unnecessary parameters and relays deceptively strong evidence for additional nonparametric terms. 


 \begin{figure}[h]
\begin{center}
\includegraphics[width=.49\textwidth]{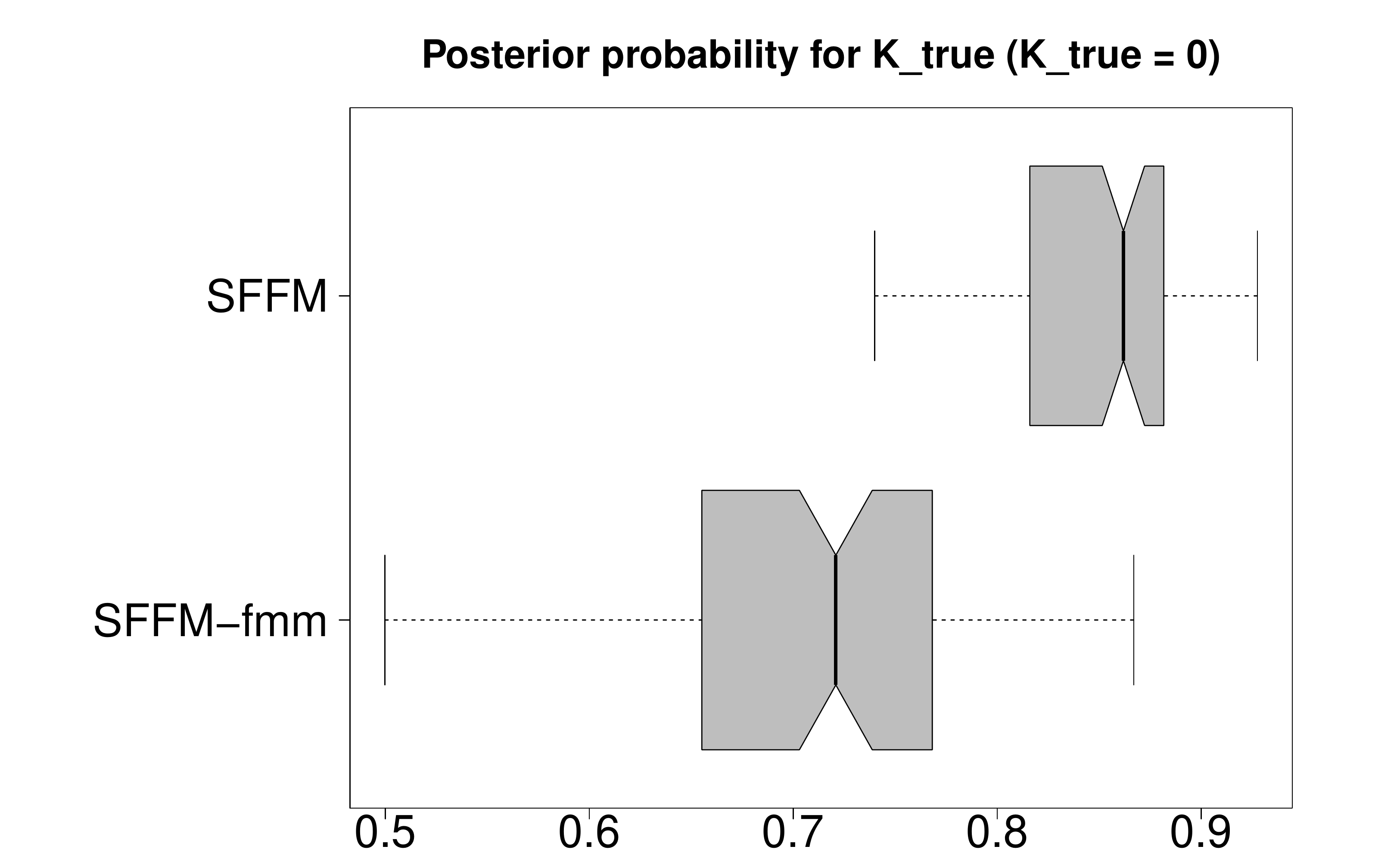}
\includegraphics[width=.49\textwidth]{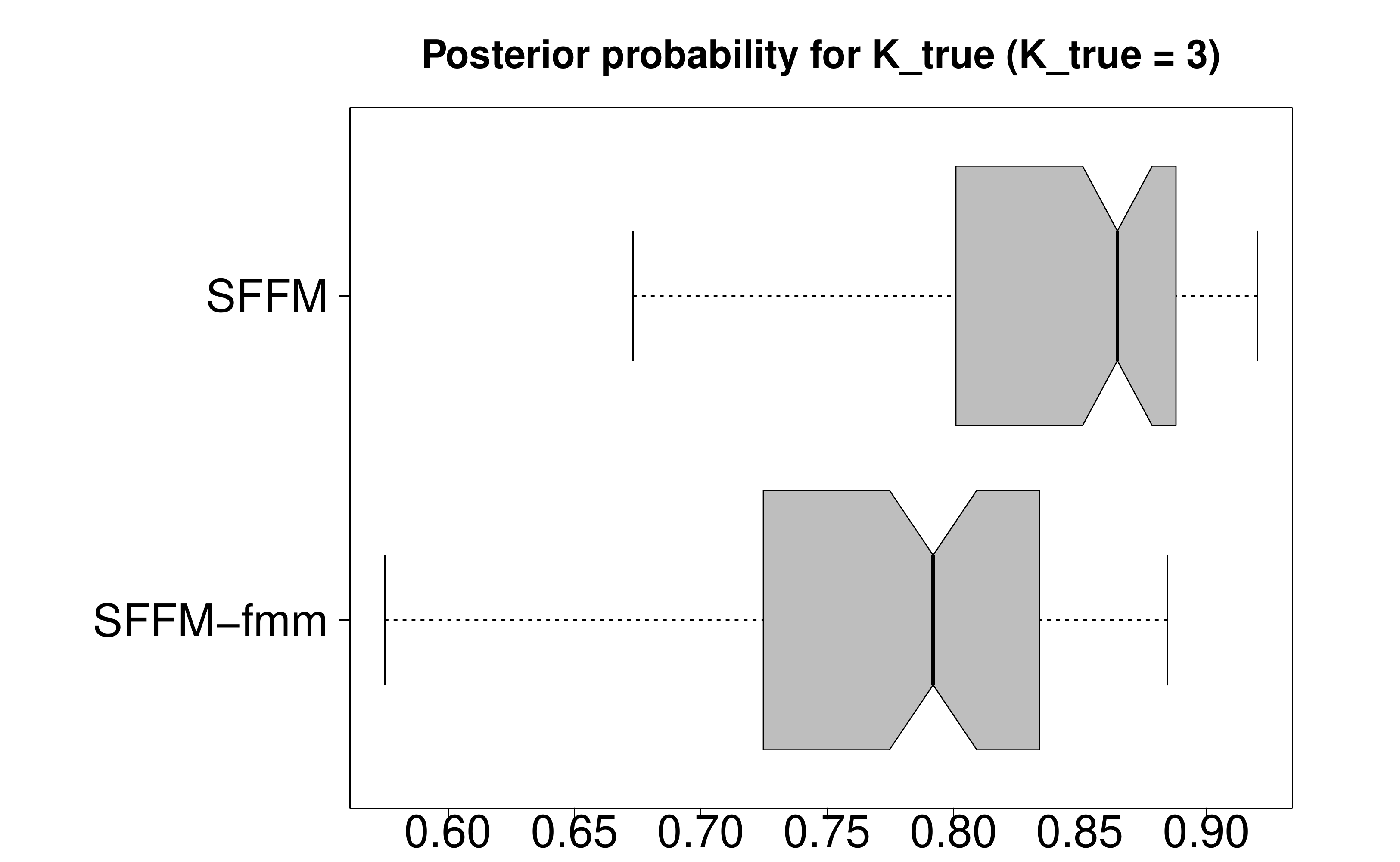}
\caption{\small   Probability score $\mathbb{P}(K^* = K_{true} \mid \bm y)$ 
for  $K_{true} = 0$ (left) and $K_{true} = 3$ (right). The proposed ordered spike-and-slab prior provides better rank selection and inference than the finite mixture alternative.  
\label{fig:kprob}}
\end{center}
\end{figure}
 
Additional simulation results in the supplement include $K_{true} = 8$, the \cite{nelson1987parsimonious} template, and further assessments of nonparametric detection via $\mathbb{P}(K^* > 0\mid \bm y)$, including ranked probability scores for the distribution $\mathbb{P}(K^*\mid \bm y)$. 
These results confirm the results presented here.
 
\clearpage

\section{Applications}\label{apps}

\subsection{Pinch force data}\label{pinch}
Human motor control is a critical area of research with implications for human physiology, monitoring and mitigating muscle degeneration, and designing robotic devices. Motor control data are recorded at high resolutions and often modeled as functional data \citep{Ramsay2000,goldsmith2016assessing}. 
Human physiology and the laws of motion dictate a parametric template, which is crucial for understanding the underlying process. However, these parametric models may not adequately describe the observed data, which  undermines the interpretability of the key parameters.

We analyze human pinching data from \cite{Ramsay1995}, which reports the force measured by pinching the thumb and forefinger on opposing sides of a 6cm force meter. The subject was instructed to maintain a background level of constant force, then increase the pinching force to a predetermined maximum level, and finally return to the original background level of constant force.  We use data from the \texttt{fda} package in \texttt{R}, which consists of $n=20$ replicate force curves over time each with $m=151$ observations selected such that the maximum of each curve occurred at 0.076 seconds. An example curve is in Figure~\ref{fig:fit}.

To model the pinch force over time $\tau$, we adapt a parametric model from 
 \cite{Ramsay1995}: 
\begin{equation}\label{par-pinch}
g_1(\tau; \gamma_i) = 1, \quad g_2(\tau; \gamma_i) = \exp[-(\log\tau - c_i)^2/\{2\exp(\gamma_i)\}],
\end{equation}
where $\exp(c_i)$ is the time of the maximum force and $\gamma_i \in \mathbb{R}$ is a shape parameter. \cite{Ramsay1995} argue that the unnormalized log-normal density $g_2$  matches the shape of the observed data and offers plausible scientific explanations. For computational convenience, we estimate $c_i$ as in \cite{Ramsay1995} by fitting a quadratic regression in $\log(\tau)$ for the response variable $\log(\bm y_i)$ restricted to observations $y_i(\tau) > 0.5$. The estimate of each $c_i$ can be recovered from the estimated regression coefficients and is subsequently treated as fixed. 

We fit the PFFM, PFFM+gp, and SFFM 
to the data using the template \eqref{par-pinch}. 
For partial pooling among subjects, we specify a hierarchical prior on the shape parameters: $\gamma_i \stackrel{iid}{\sim} N(\mu_\gamma, \sigma_\gamma^2)$, $\mu_\gamma \sim N(0, 10)$, and $\sigma_\gamma \sim C^+(0,1)$, with the priors from \eqref{priors-alpha} on the remaining parameters. Since the nonlinear parameters $\gamma_i$ are curve-specific, we construct each $\bm G_{\gamma_i}$ for $i=1,\ldots,n$ using a QR decomposition as in Section~\ref{all-np} and modify the orthogonality constraint to be $\bm{\bar G}'\bm f_k = \bm 0_{L}$ for $\bm{\bar G} = n^{-1}\sum_{i=1}^n \bm G_{\gamma_i}$. This constraint no longer preserves the posterior factorization of Corollaries~\ref{post-indep}-\ref{cor-beta}, but nonetheless maintains sufficient distinctness between the parametric and nonparametric terms. The SFFM hyperparameters are set to $a_1 = 5$,  $a_2 = 25$,  $v_0 = 0.001$, and $K=10$ as in Section~\ref{sims}.

An example of the fitted values with 95\% simultaneous prediction bands for the PFFM and SFFM is in Figure~\ref{fig:fit}. Although the PFFM captures the general shape of the data, it suffers from clear bias around the peak and produces unnecessarily wide prediction bands. The SFFM corrects both issues: the bias is removed and the prediction bands are more precise. Notably, the fitted SFFM curve preserves the same general shape as the PFFM and avoids overfitting despite the increase in modeling complexity.

Posterior uncertainty quantification is also more precise under the SFFM for the parametric factors $\{\alpha_{\ell, i}\}$. Figure~\ref{fig:sd} presents the posterior standard deviations of $\{\alpha_{\ell, i}\}_{i=1}^n$ for each $\ell = 1,\ldots,L$ under the PFFM and SFFM. Although the posterior expectations  are similar (see Figure~\ref{fig:alphapm}), the posterior standard deviations are substantially smaller under the SFFM. This reduction is even more dramatic compared to the PFFM+gp, which is omitted from Figure~\ref{fig:sd} due to excessively large posterior standard deviations that  range from 1.42 to 4.43. For model \eqref{par-pinch}, the linear coefficients $\{\alpha_{\ell,i}\}$ determine the maximum of the force curve, which is the most prominent feature in the data. The SFFM provides more precise posterior inference for these key parametric factors.

\begin{figure}[h]
\begin{center}
\includegraphics[width=.49\textwidth]{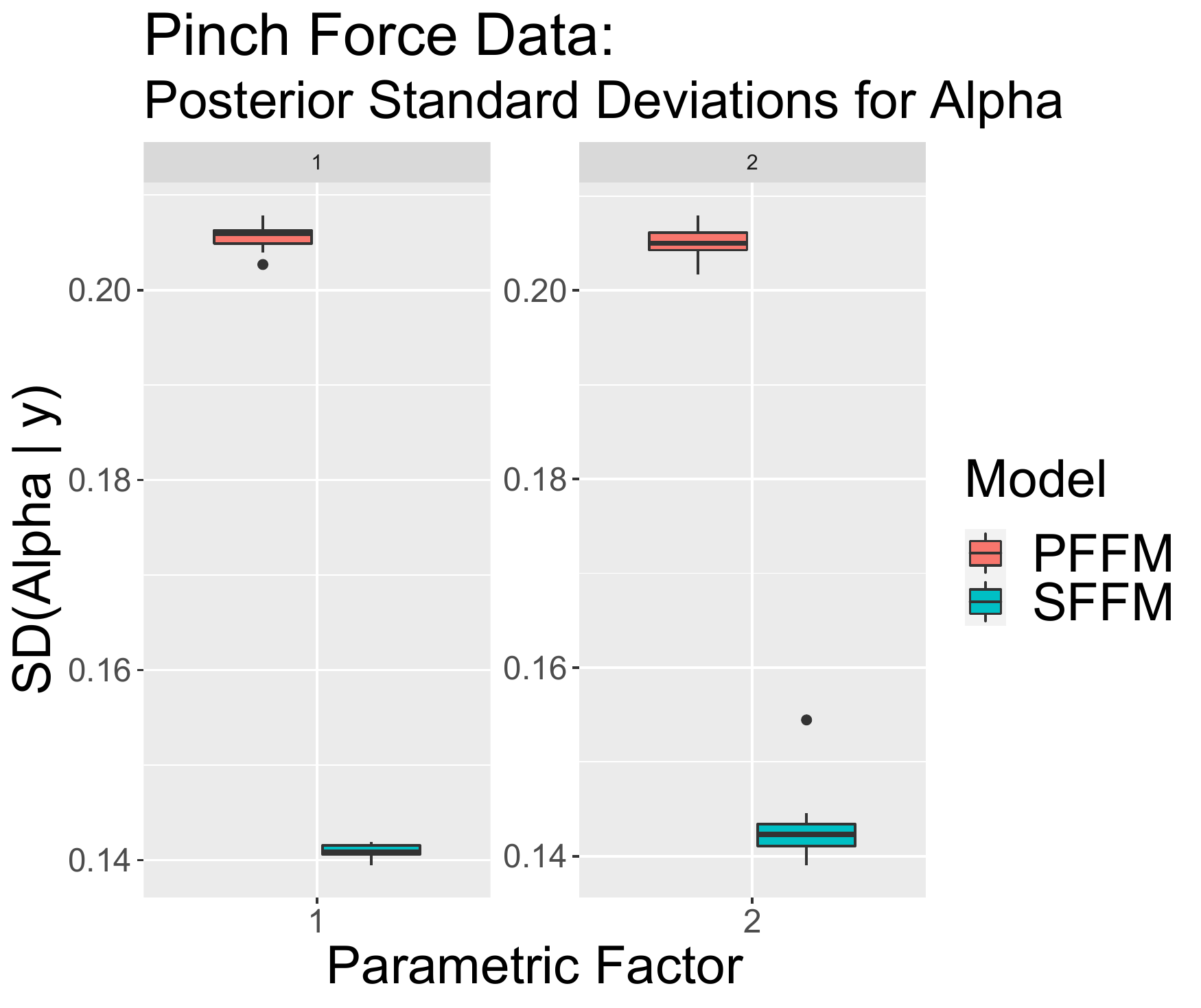}
\includegraphics[width=.49\textwidth]{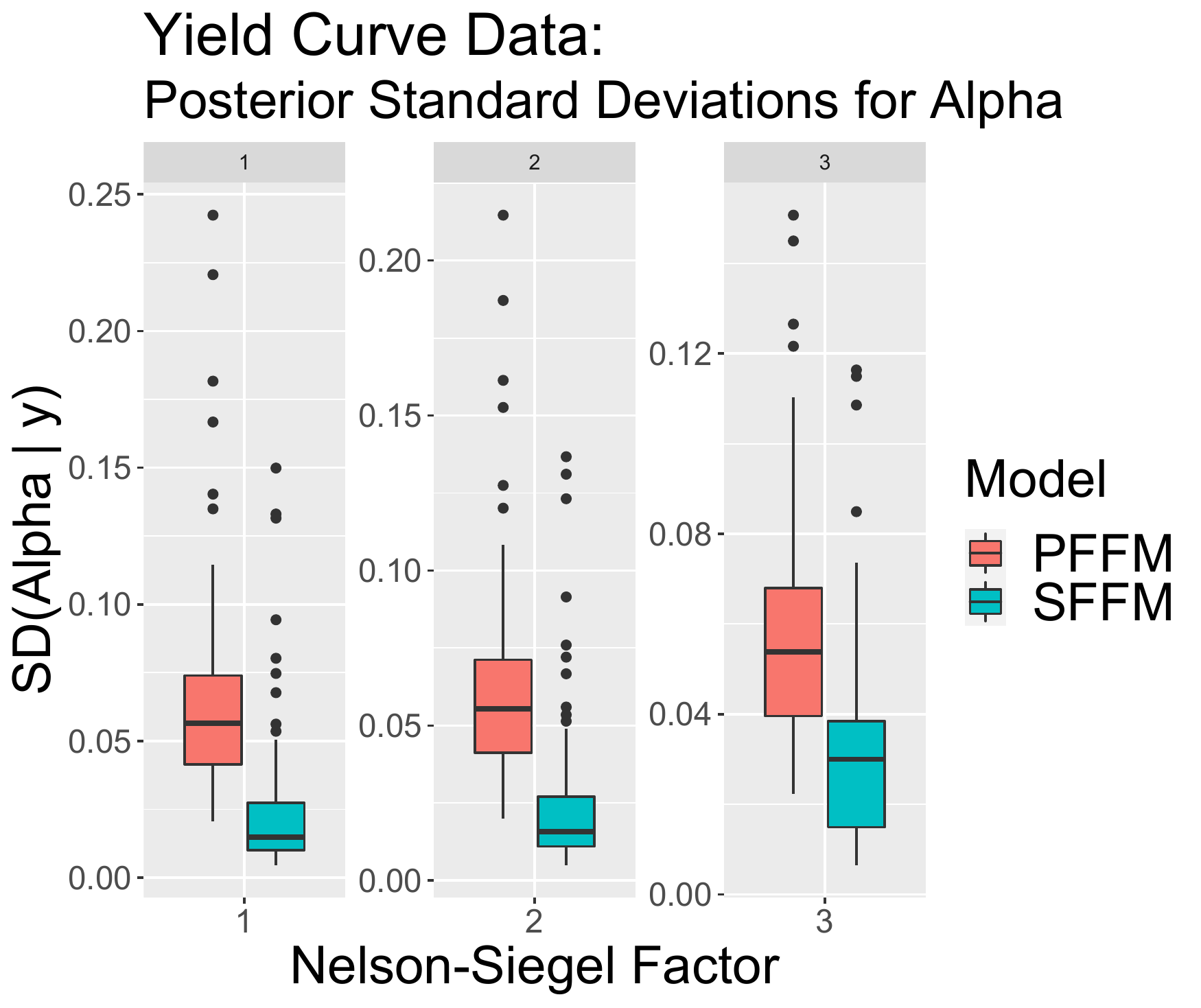}
\caption{\small Posterior standard deviations of $\{\alpha_{\ell,i}\}_{i=1}^n$ for $\ell=1,\ldots,L$  in the parametric (PFFM) and semiparametric functional factor model (SFFM) for the pinch force data (left) and yield curve data (right). There is a clear and consistent reduction in posterior uncertainty from the PFFM to the SFFM for these key parametric factors.
 \label{fig:sd}}
\end{center}
\end{figure}

Additional evidence in favor of the SFFM is presented in Table~\ref{tab-ppk}, which estimates the posterior distribution of the effective number of nonparametric terms $K^*$. Clearly, there is substantially posterior probability for at least two nonparametric terms, with limited evidence that more than three terms are needed. WAIC 
similarly favors the SFFM: the estimates are $-922$, $-2693$, and $-3005$ for the PFFM, PFFM+gp, and SFFM, respectively.

\begin{table}[ht]
\centering
\begin{tabular}{ c | c | c  | c | c | c | c  | c   }
$k$ &  0    &  1   &   2  &   3   &  4   &  5 & $\ge 6$ \\ \hline
$\mathbb{P}(K^* = k \mid \bm y)$ & 0.003 & 0.147  &0.707 &0.135 & 0.007 & 0.001 & 0\\ 
 \end{tabular}
\caption{\small Posterior probabilities $\mathbb{P}(K^* = k \mid \bm y)$ for the pinch force data. There is strong evidence for a nonparametric component.  
\label{tab-ppk}}
\end{table}

We highlight the excellent MCMC performance of the SFFM.  For all models, we retain 10000 MCMC samples after discarding 5000 iterations as a burn-in. Traceplots of $\{Y_i\}$, $\{\alpha_{\ell,i}\}$, and $K^*$ (not shown) demonstrate excellent mixing and suggest convergence.  
The computing time per 1000 iterations is 1.59s, 4.35s, and 1.98s   for the PFFM, PFFM+gp, and SFFM,  respectively (using \texttt{R} on a MacBook Pro, 2.8 GHz Intel Core i7). Comparing the effective sample sizes for all $\{\alpha_{\ell,i}\}$,  the first  quartiles are 10000, 1361, and 9508 for the PFFM, PFFM+gp, and SFFM, respectively. In aggregate, the SFFM not only provides superior modeling  capabilities, but also rivals the simpler PFFM in computing time and MCMC efficiency---while the PFFM+gp lags far behind in both. 

\subsection{Dynamic yield curves}\label{yields}
Yield curves play a fundamental role in economic and financial analyses: they provide essential information about current and future economic conditions, including inflation, business cycles, and monetary policies, 
and are used to price fixed-income securities  and construct forward curves. The yield curve $Y_i(\tau)$ describes how interest rates vary as a function of the length of the borrowing period, or time to maturity $\tau$, at each time $i$. Naturally, yield curves can be modeled as functional data that evolve dynamically over time.

Yield curve models most commonly employ
the parametric  \cite{nelson1987parsimonious} basis: 
\begin{equation}\label{par-yields}
g_1(\tau; \gamma) = 1, \quad
 g_2(\tau; \gamma) = \{1 - \exp(-\tau\gamma)\}/(\tau\gamma), \quad
  g_3(\tau; \gamma) =  g_2(\tau; \gamma)  - \exp(-\tau\gamma)
\end{equation}
where $g_1$ is the \emph{level}, $g_2$ is the \emph{slope}, and $g_3$ is the \emph{curvature}. The nonlinear parameter $\gamma > 0$ is commonly treated as fixed, such as $\gamma = 0.0609$ \citep{diebold2006forecasting}, but otherwise may be estimated. We use a weakly informative Gamma prior for  $\gamma$ with prior mean 0.0609 and prior variance 0.5. Both the PFFM and the SFFM with this prior on $\gamma$ are favored by WAIC over their respective counterparts with fixed $\gamma =0.0609$. 


To capture the yield curve dynamics, we model the parametric Nelson-Siegel factors as an AR(1):
\begin{equation}\label{areg}
\alpha_{\ell,i} = \mu_\ell + \phi_\ell (\alpha_{\ell, i-1} - \mu_\ell) + \zeta_{\ell,i}, \quad \zeta_{\ell,i} \stackrel{indep}{\sim} N(0, \sigma_{\zeta_\ell}^2).
\end{equation}
The PFFM with \eqref{par-yields} and \eqref{areg} is also known as the \emph{dynamic Nelson-Siegel model} \citep{diebold2006forecasting}.  The dynamic Nelson-Siegel factors $\{\alpha_{\ell,i}\}$ may be viewed as the state variables in a dynamic linear model. Using the orthogonality constraints  on $\{g_\ell\}$ and $\{f_k\}$, Corollary~\ref{cor-alpha} provides a  convenient and low-dimensional state  simulation algorithm for efficient joint sampling of the dynamic factors $\{\alpha_{\ell,i}\}$; see the supplementary material for details. Similar  simplifications are \emph{not} available for the PFFM+gp, which is omitted.

The dynamic model \eqref{areg} is accompanied by a diffuse prior  
 $\mu_\ell \stackrel{iid}{\sim} N(0, 10^6)$ and a weakly informative prior $(\phi_\ell + 1)/2 \stackrel{iid}{\sim} \mbox{Beta}(5,2)$---which ensures stationarity of the dynamic factors $\{\alpha_{\ell, i}\}$ and therefore $\{Y_i\}$ and $\{\bm y_i\}$---along with $\sigma_{\zeta_\ell} \stackrel{iid}{\sim} C^+(0,1)$ for $\ell=1,\ldots,L$. In addition, we generalize \eqref{obs} to accommodate stochastic volatility in the observation error variance, $
\epsilon_{i,j} \stackrel{indep}{\sim}N(0, \sigma_{\epsilon_i}^2)$, with  $ \log \sigma_{\epsilon_i}^2 \sim \mbox{AR(1)}$, 
which is an essential component in many economic and financial models  \citep{kim1998stochastic}. For these AR(1) parameters, we adopt the priors and sampling algorithm from \cite{Kowal2020b}.


We evaluate the suitability of the Nelson-Siegel basis for modeling monthly unsmoothed Fama-Bliss US government bond yields from 2000-2009 ($n=120$) provided by \cite{van2014forecasting}. These data are available for maturities of 3, 6, 9, 12, 15, 18, 21, 24, 30, 36, 48, 60, 72, 84, 96, 108 and 120 months ($m = 17$). 
The inferential targets are the latent curves $Y_i$, the dynamic Nelson-Siegel factors $\{\alpha_{i,\ell}\}$, and the effective number of nonparametric terms $K^*$. The  SFFM hyperparameter and MCMC specifications from Section~\ref{pinch} are adopted here, again with excellent mixing and convergence for these key parameters.


Figure~\ref{fig:fit} shows the PFFM and SFFM fitted values with 95\% simultaneous prediction bands for the yield curve in September 2008 
during the onset of the Great Recession. The PFFM produces a reasonable shape for the yield curve, yet---as with the pinch force data---suffers from clear bias and overconservative prediction bands. The SFFM corrects these deficiencies without distorting the general shape of the curve or overfitting to the data.

 Similar results are obtained  in Figure~\ref{fig:sd} for the dynamic Nelson-Siegel factors. The SFFM offers substantial reductions in posterior standard deviation for all three dynamic Nelson-Siegel factors. These factors are fundamental for describing the shape of the yield curve; reducing the posterior uncertainty is a crucial advantage of the SFFM. Importantly, the simulation study confirms that the reduced posterior uncertainty quantification from the SFFM nonetheless retains valid calibration, or more specifically, correct nominal coverage of the posterior credible intervals.


We summarize the posterior distribution $\mathbb{P}(K^* \mid \bm y)$ of the effective number of nonparametric terms in Table~\ref{tab-ppk-yields}. The evidence for the nonparametric factors is moderate: we estimate $\mathbb{P}(K^* > 0 \mid \bm y) = 0.09$. 

\begin{table}[ht]
\centering
\begin{tabular}{ c | c | c  | c | c  }
$k$ &  0    &  1   &   2  & $\ge 3$ \\ \hline
$\mathbb{P}(K^* = k | \bm y)$ & 0.910 &  0.087  & 0.003 & 0\\ 
 \end{tabular}
\caption{\small Posterior probabilities $\mathbb{P}(K^* = k \mid \bm y)$ for the yield curve data. There is moderate evidence for a nonparametric component.  
\label{tab-ppk-yields}}
\end{table}

There are other compelling reasons to include the nonparametric factors in a dynamic Nelson-Siegel model. First, the uncertainty quantification for $\{Y_i\}$ and $\{\alpha_{\ell,i}\}$  is more precise in the SFFM than in the PFFM (see Figures~\ref{fig:fit}~and~\ref{fig:sd}). Second, WAIC decisively prefers the SFFM (-10229) over the PFFM (-5516), which suggests potential improvements in out-of-sample predictive capabilities.
Third, the additional computational burdens of the SFFM are minimal. Due to the orthogonality constraints and Corollory~\ref{cor-alpha}, the sampling steps for $\{\alpha_{\ell, i}\}$ and associated parameters are identical for the PFFM and the SFFM, which permits full and robust model development for the parametric factors. Further, the computing times per 1000 iterations are similar:  21s and 23s for the PFFM and SFFM, respectively (using \texttt{R} on a MacBook Pro, 2.8 GHz Intel Core i7).




 \section{Discussion}\label{dis}
We proposed a Bayesian semiparametric model for functional data. The semiparametric functional factor model (SFFM) augmented a parametric template with an infinite-dimensional nonparametric functional basis. The nonparametric basis was treated as unknown and learned from the data to correct the biases of the parametric template while appropriately incorporating relevant uncertainties into the posterior distribution. Distinctness between the parametric and nonparametric terms was achieved by conditioning upon an orthogonality constraint, which simultaneously prevented \emph{functional confounding} and admitted highly convenient simplifications for efficient MCMC sampling.  The nonparametric component was regularized with an ordered spike-and-slab prior, which implicitly provided rank selection for infinite-dimensional models and satisfied appealing theoretical properties. This prior was accompanied by a parameter expansion scheme customized to boost MCMC efficiency, and is broadly applicable for Bayesian factor models. Our analyses of synthetic data, human motor control data, and dynamic interest rates demonstrated clear advantages of the semiparametric modeling framework relative to both parametric and Gaussian process alternatives. The proposed approach eliminated bias, reduced excessive posterior and predictive uncertainty, and provided reliable inference on the effective number of nonparametric terms---all with minimal computational costs. 


 There are several promising extensions that remain for future work. First,  formulation of the parametric template in \eqref{sffm} can be generalized, for example to include a functional regression term with scalar or functional covariates. Second, the ordered spike-and-slab prior currently uses independent and identically distributed variables in the parameter expansion \eqref{px}. Adaptations to include dependence among these variables, such as regression models, clustering, and spatio-temporal dependence, would broaden the applicability of the prior. Lastly, we studied only a small subset of many possible parametric templates. For applications that rely on such parametric models, the proposed semiparametric modeling framework can directly assess  the adequacy of these models---and perhaps suggest improvements.

\ifblinded
\else 
\subsection*{Acknowledgements} 
Research was sponsored by the Army Research Office and was accomplished under Grant Number W911NF-20-1-0184 (Kowal). The views and conclusions contained in this document are those of the authors and should not be interpreted as representing the official policies, either expressed or implied, of the Army Research Office or the U.S. Government. The U.S. Government is authorized to reproduce and distribute reprints for Government purposes notwithstanding any copyright notation herein. 
\fi

\bibliographystyle{apalike}
\bibliography{refs}

\clearpage

\section{Proofs}\label{a-proofs}

\begin{lemma} 
When $\bm G_\gamma ' \bm f_k = \bm 0_{L}$ for all $k=1,2,\ldots$ and conditional on $\sigma_\epsilon^2$, the likelihood \eqref{like} factorizes:
$p(\bm y \mid  \gamma , \{\bm \alpha_i\}, \bm F,  \{\bm \beta_i\}) =  p_{0}(\bm y \mid  \gamma ,\{\bm \alpha_i\}) p_{1}(\bm y \mid  \bm F , \{\bm \beta_i\})$,
where $p_0$ depends \emph{only} on the parametric terms $\gamma$  and $\{\bm \alpha_i\}$ and  $p_1$ depends \emph{only} on  the nonparametric terms $\bm F$  and $\{\bm \beta_i\}$. 
\end{lemma}
\begin{proof}
Observing that $\Vert\bm y_i - \bm G_\gamma \bm \alpha_i - \bm F\bm \beta_i\Vert^2 = \Vert\bm y_i\Vert^2 + \Vert \bm G_\gamma \bm \alpha_i \Vert^2 + \Vert \bm F\bm \beta_i\Vert^2- 2 \bm \alpha_i \bm G_\gamma ' \bm y_i - 2 \bm \alpha_i \bm G_\gamma ' \bm F \bm \beta_i  - 2 \bm \beta_i \bm F ' \bm y_i = \Vert\bm y_i\Vert^2 + \big(\Vert \bm G_\gamma \bm \alpha_i \Vert^2- 2 \bm \alpha_i \bm G_\gamma ' \bm y_i  \big) + \big(\Vert \bm F\bm \beta_i\Vert^2 - 2 \bm \beta_i \bm F ' \bm y_i \big)$, the joint likelihood factorizes into 
\begin{align*}
 p(\bm y | -) &= \prod_{i=1}^n p(\bm y_i |  -) =  \prod_{i=1}^n  (2\pi \sigma_\epsilon^2)^{-m/2}\exp\{-\Vert\bm y_i - \bm G_\gamma \bm \alpha_i - \bm F\bm \beta_i\Vert^2/(2\sigma_\epsilon^2) \}\\
&= \prod_{i=1}^n  (2\pi \sigma_\epsilon^2)^{-m/2}\exp[- \{\Vert\bm y_i\Vert^2 + \big(\Vert \bm G_\gamma \bm \alpha_i \Vert^2- 2 \bm \alpha_i \bm G_\gamma ' \bm y_i  \big) + \big(\Vert \bm F\bm \beta_i\Vert^2 - 2 \bm \beta_i \bm F ' \bm y_i \big)\}/(2\sigma_\epsilon^2) ]\\
&= (2\pi \sigma_\epsilon^2)^{-nm/2}\exp\Big\{-\sum_{i=1}^n \Vert\bm y_i\Vert^2/(2\sigma_\epsilon^2)\Big\}
\exp\Big\{-\sum_{i=1}^n \big(\Vert \bm G_\gamma \bm \alpha_i \Vert^2- 2 \bm \alpha_i' \bm G_\gamma ' \bm y_i\big)/(2\sigma_\epsilon^2)\Big\} \\
& \quad \quad \quad 
\exp\Big\{-\sum_{i=1}^n \big(\Vert \bm F\bm \beta_i\Vert^2 - 2 \bm \beta_i' \bm F ' \bm y_i \big)/(2\sigma_\epsilon^2)\Big\}
\end{align*}
which produces the necessary product factorization. 
\end{proof}

\begin{corollary} 
Under model \eqref{like} with prior independence $p(\{\bm \alpha_i\},  \{\bm \beta_i\} ) =  p( \{\bm \alpha_i\}) p(\{\bm\beta_i\})$ and  $\bm G_\gamma ' \bm f_k = \bm 0_{L}$ for all $k=1,2,\ldots$, the parametric and nonparametric factors are \emph{a posteriori} independent: 
$p( \{\bm \alpha_i\}, \{\bm \beta_i\} \vert \bm y, \sigma_\epsilon^2,  \bm F, \gamma) =  p(\{\bm \alpha_i\}\vert \bm y, \sigma_\epsilon^2,  \bm F, \gamma) p( \{\bm\beta_i\}\vert \bm y, \sigma_\epsilon^2,  \bm F, \gamma)$. 
\end{corollary}
\begin{proof}
The result follows as a direct consequence of (i) the prior independence assumption and (ii) the likelihood factorizations from Lemma~\ref{like-decomp}.
\end{proof}

  \begin{lemma}
	Suppose $\bm G_\gamma$ is $L\times m$ with rank $L$. Denote $\bm\psi_k^0 \sim N\left(\bm Q_{\psi_k}^{-1} \bm \ell_{\psi_k}, \bm Q_{\psi_k}^{-1}\right)$ with $\bm f_k^0 = \bm B \bm \psi_k^0$. The shifted term  $\bm f_k = \bm B \bm \psi_k$ with $\bm \psi_k = \bm \psi_k^0 - \bm Q_{\psi_k}^{-1} \bm B' \bm G_\gamma \left(\bm G_\gamma' \bm B\bm Q_{\psi_k}^{-1}\bm B' \bm G_\gamma\right)^{-1} \bm G_\gamma' \bm B\bm \psi_k^0$  satisfies  $\bm f_k \stackrel{d}{=}[\bm f_k^0 \mid \bm G_\gamma' \bm f_k^0 = \bm 0]$ and $ \mathbb{P}(\bm G_\gamma' \bm f_k = \bm 0) = 1$. 
\end{lemma}
\begin{proof}
To show the result, we apply a modification of Theorem 1 in \cite{Kowal2020b}. The full conditional distribution of the unconstrained term is $\bm f_k^0 \sim N\big(\bm \mu_{f_{k}^0}, \bm \Sigma_{f_{k}^0}\big)$ with $\bm \mu_{f_{k}^0} = \bm B \bm Q_{\psi_k}^{-1} \bm \ell_{\psi_k}$ and $ \bm \Sigma_{f_{k}^0} = \bm B  \bm Q_{\psi_k}^{-1} \bm B'$. Since $\bm f_k^0$ is Gaussian, so is $\bm G_\gamma' \bm f_k^0$, and consequently $[\bm f_k^0 | \bm G_\gamma' \bm f_k^0 = \bm 0] \sim N\big(\bm \mu_{f_{k}}, \bm \Sigma_{f_{k}}\big)$ with $\bm \mu_{f_{k}} = \bm \mu_{f_{k}^0} -  \bm \Sigma_{f_{k}^0}\bm G_\gamma(\bm G_\gamma'  \bm \Sigma_{f_{k}^0} \bm G_\gamma)^{-1}\bm G_\gamma' \bm \mu_{f_{k}^0} $ and $\bm \Sigma_{f_{k}} =  \bm \Sigma_{f_{k}^0} -  \bm \Sigma_{f_{k}^0} \bm G_\gamma (\bm G_\gamma'  \bm \Sigma_{f_{k}^0}\bm G_\gamma)^{-1}  \bm \Sigma_{f_{k}^0}$. Direct calculation shows that the shifted term $\bm f_k$ has the same mean and covariance. Pre-multiplying a random variable with distribution $[\bm f_k^0 | \bm G_\gamma' \bm f_k^0 = \bm 0] $ by $\bm G_\gamma'$ produces a random variable with degenerate distribution at $\bm 0_L$.
\end{proof}

\begin{corollary} 
Under model \eqref{like} and subject to $\bm G_\gamma'\bm f_k = \bm 0_L$ for all $k$ and $\bm G_\gamma'\bm G_\gamma = \bm I_L$, the likelihood for the parametric factors is proportional to $p(\bm y \mid  \gamma , \{\bm \alpha_i\}, \bm F,  \{\bm \beta_i\}) \propto  p_{0, G}(\bm y \mid  \gamma ,\{\bm \alpha_i\})$ where $p_{0,G}$ is the likelihood defined by $[\bm g_{\ell; \gamma}'\bm y_i \mid -] \stackrel{indep}{\sim} N(\alpha_{\ell,i} , \sigma_\epsilon^2)$ for $\ell=1,\ldots,L$.
\end{corollary}
\begin{proof}
From the proof of Lemma~\ref{like-decomp} and considering terms involving only the parametric factors, we have 
\begin{align*}
p(\bm y | -) &\propto \exp\Big\{-\sum_{i=1}^n \big(\Vert \bm G_\gamma \bm \alpha_i \Vert^2- 2 \bm \alpha_i \bm G_\gamma ' \bm y_i\big)/(2\sigma_\epsilon^2)\Big\} \\
&= \exp\Big\{-\sum_{i=1}^n \big(\Vert  \bm \alpha_i \Vert^2- 2 \bm \alpha_i' \bm G_\gamma ' \bm y_i\big)/(2\sigma_\epsilon^2)\Big\} \\
&\propto  p_{0, G}(\bm y \mid  \gamma ,\{\bm \alpha_i\})
\end{align*}
using the orthonormality  $\bm G_\gamma'\bm G_\gamma = \bm I_L$.  
\end{proof}

 \begin{corollary}
Under model \eqref{like} and subject to $\bm G_\gamma'\bm f_k = \bm 0_L$ for all $k$ and $\bm F'\bm F = \bm I_\infty$, the likelihood for the nonparametric factors is proportional to $p(\bm y \mid  \gamma , \{\bm \alpha_i\}, \bm F,  \{\bm \beta_i\})  \propto  p_{1,F}(\bm y \mid  \bm F , \{\bm \beta_i\}) $ where $p_{1,F}$ is the likelihood defined by $[\bm f_k'\bm y_i \mid -] \stackrel{indep}{\sim} N(\beta_{k,i}, \sigma_\epsilon^2)$ for all $k$.
\end{corollary}
\begin{proof}
The result follows using the same arguments as in the proof of Corollary~\ref{cor-alpha}. 
\end{proof}

\begin{proposition}
For $\varepsilon>0$ and fixed $\theta_0$, let $\mathbb B_\varepsilon(\theta_0) = \{\theta_k: | \theta_k - \theta_0| < \varepsilon\}$. Prior  \eqref{csp}-\eqref{pi-prior} implies that
$\mathbb{P}(|\theta_k - \theta_0| \le \varepsilon) < \mathbb{P}(|\theta_{k+1} - \theta_0| \le \varepsilon) $
whenever $P_{slab}\{\mathbb B_\varepsilon(\theta_0)\} < P_{spike}\{\mathbb B_\varepsilon(\theta_0)\}$. 
\end{proposition}
\begin{proof}
The proof proceeds directly: \vspace{-5mm}
\begin{align*}
\mathbb{P}(|\theta_k - \theta_0| \le \varepsilon) &= \mathbb{E}[P_k\{\mathbb B_\varepsilon(\theta_0)\}] =   \mathbb{E}[(1-\pi_k)P_{slab}\{\mathbb B_\varepsilon(\theta_0)\} + \pi_k P_{spike}\{\mathbb B_\varepsilon(\theta_0)\}] \\
&=  \mathbb{E}\{(1-\pi_k)\}P_{slab}\{\mathbb B_\varepsilon(\theta_0)\}  + \mathbb{E}(\pi_k) P_{spike}\{\mathbb B_\varepsilon(\theta_0)\}\\
&= P_{slab}\{\mathbb B_\varepsilon(\theta_0)\} \{\kappa/(1+\kappa)\}^k + P_{spike}\{\mathbb B_\varepsilon(\theta_0)\} [1 - \{\kappa/(1+\kappa)\}^k ] \\
&= P_{spike}\{\mathbb B_\varepsilon(\theta_0)\}  + \{\kappa/(1+\kappa)\}^k [P_{slab}\{\mathbb B_\varepsilon(\theta_0)\} - P_{spike}\{\mathbb B_\varepsilon(\theta_0)\}]
\end{align*}
noting   $\mathbb{E}(\pi_k) = 1 - \{\iota \kappa/(\iota+\iota \kappa)\}^k = 1 - \{\kappa/(1+\kappa)\}^k$.  Since $\{\kappa/(1+\kappa)\}^k \in [0,1]$ is decreasing in $k$ and the remaining terms are invariant to $k$, the result follows.
\end{proof}

   \begin{corollary}
  For $\varepsilon>0$, $\mathbb{P}(|\eta_k| \le \varepsilon) < \mathbb{P}(|\eta_{k+1}| \le \varepsilon)$ whenever $v_0 < 1$.
\end{corollary}
\begin{proof}
Corollary~\ref{eta-shrink} follows from Proposition~\ref{csp-shrink} by observing that $[\eta_k | \pi_k] \sim (1-\pi_k) t_{2a_1}(0, \sqrt{a_2/a_1}) +\pi_k t_{2a_1}(0, \sqrt{v_0a_2/a_1})$, where $t_d(m, s)$ denotes a $t$-distribution with mean $m$, standard deviation $s$, and degrees of freedom $d$. This mixture is a special case of \eqref{csp}, and the densities of the $t$-distributions  place greater mass near zero when the scale parameter is smaller. The condition  $v_0 < 1$ ensures that the spike distribution is indeed more concentrated around zero. 
\end{proof}

\begin{proposition}
The CUSP \eqref{pi-prior} satisfies $(1 - \pi_k) = \mu_{(k)}$ where $\mu_{(1)} > \cdots > \mu_{(K)}$ are the ordered (slab) probabilities from the \cite{Teh2007} stick-breaking construction of the IBP, i.e., $\mu_{(k)} = \prod_{\ell = 1}^k \nu_\ell'$ with   $\nu_\ell' \stackrel{iid}{\sim} \mbox{Beta}(\iota \kappa, \iota).$
\end{proposition}
\begin{proof}
Recall that  IBPs are obtained by 
establishing the conditional distributions  
$ [b_{j,k} \mid \mu_k] \stackrel{indep}{\sim}  \mbox{Bernoulli}(\mu_k) $ and  $ [\mu_k] \stackrel{iid}{\sim} \mbox{Beta}(\iota \kappa/K, \iota)$ and then integrating over $\{\mu_k\}$ with $K\rightarrow \infty$. Now
let $\nu_\ell' = (1-\nu_\ell) \stackrel{iid}{\sim}\mbox{Beta}(\iota \kappa, \iota)$, so $\omega_h = (1-\nu_h') \prod_{\ell=1}^{h-1} \nu_\ell' = \prod_{\ell=1}^{h-1} \nu_\ell' - \prod_{\ell=1}^{h} \nu_\ell'$. By cancelling terms in the cumulative summation $\pi_k = \sum_{h=1}^k \omega_h$, the result for $1-\pi_k$ follows.
\end{proof}

\begin{lemma}
	\label{lem:priorR}
	Let $\varepsilon_n\to 0$ with $\varepsilon_n^{1/K_{0n}} > \kappa/(\kappa+1)$. For the CUSP prior  \eqref{csp}-\eqref{pi-prior} and a positive constant  $C>1$, the remainder term $R_n = \sum_{k\geq  K_{0n}} \omega_k$ satisfies
	\begin{eqnarray}
		\mathbb{P}(R_{n} >\varepsilon_n) \leq \exp(-C K_{0n}).
		\label{eq:priorR}
	\end{eqnarray}
\end{lemma}

\begin{proof}
	Rewriting the reminder as 
$
	R_{n} = \prod_{h=1}^{K_{0n}} (1 - \nu_h),
$
	the \emph{a priori} expectation of the remainder is
	\[
	\mathbb{E}(R_{n}) =  \prod_{h=1}^{K_{0n}} \mathbb{E}(1 - \nu_h) = \left( \frac{\kappa}{\kappa+1}\right)^{K_{0n}}
	\]
	from which we conclude that
		\[
	\mathbb{P}(R_{0n} >\varepsilon_n )  \leq \frac{1}{\varepsilon_n} \left( \frac{\kappa}{\kappa+1}\right)^{K_{0n}} = 
	 \left( \frac{1}{\varepsilon_n^{1/K_{0n}} }  \frac{\kappa}{\kappa+1}\right)^{K_{0n}} 
	\leq 
	\exp(-CK_{0n}),
	\]
	where the first inequality follows from  Markov's inequality, and the last one from the assumptions on the upper bound for $\varepsilon_n$.
\end{proof}

Before proving Theorem 1, we introduce the following lemma:
\begin{lemma}
	\label{lem:denominator}
	Let $1<r<\sqrt{n}$, $a_1>1/2$, and $a_1 \leq a_2$. Under the CUSP and NMIG in \eqref{csp}-\eqref{pi-prior} and  \eqref{nmig}, respectively, then \emph{a priori}
	\begin{eqnarray*}
	\mathbb{P}\left(|\eta_h| < \frac{r}{\sqrt{n}}\right) > \left\{1-\left(\frac{\kappa}{\kappa+1}\right)^h\right\} \left(1 - \frac{1}{1+ \frac{a_1r^2}{nv_0a_2}}\right) 
		\end{eqnarray*}
\end{lemma}
\begin{proof}
First, we consider the conditional probability 
\[ \mathbb{P}\left(|\eta_h| < \frac{r}{\sqrt{n}} \mid \pi_h \right)\] and then integrate $\pi_h$ with respect to the CUSP. 
We know that, in a neighborhood of zero, the density function of a Student $t$ random variable with degrees of freedom greater or equal to 1 is bounded below by the density of a Cauchy random variable. The  cumulative distribution function  of a Cauchy random variable is 
\[F(t) = \frac{1}{2} + \frac{1}{\pi}\arctan(t).
\]
Hence, for sufficiently large $n$ we can write
	\begin{eqnarray*}
\mathbb{P}\left(|\eta_h| < \frac{r}{\sqrt{n}}\mid \pi_h \right) & \geq & (1 - \pi_h) \frac{2}{\pi}\arctan\left(\frac{r\sqrt{a_1}}{\sqrt{a_2 n}}\right) + \pi_k \frac{2}{\pi}\arctan\left(\frac{r\sqrt{a_1}}{\sqrt{a_2 nv_0}}\right) \\
& \geq &   \frac{2\pi_h}{\pi}\arctan\left(\frac{r\sqrt{a_1}}{\sqrt{nv_0 a_2}}\right) \\
& \geq &   \pi_h \left(1 - \frac{1}{1+ \frac{a_1r^2}{nv_0a_2}}\right) 
	\end{eqnarray*}
and the last inequality holds since $2/\pi \arctan(t) > (1 - 1/(1+t^2))$ for $t \in (0,1)$.
Finally, integrating out $\pi_h$ yields
\[
	\mathbb{P}\left(|\eta_h| < \frac{r}{\sqrt{n}}\right) \geq
	\mathbb{E} \left(\pi_h\right) \left(1 - \frac{1}{1+ \frac{a_1r^2}{nv_0a_2}}\right) = \left\{1-\left(\frac{\kappa}{\kappa+1}\right)^h\right\} \left(1 - \frac{1}{1+ \frac{a_1r^2}{na_2}}\right) .
\]
\end{proof}

\begin{theorem}
Let $\varepsilon_n\to 0$ with $\varepsilon_n^{1/K_{0n}} > \kappa/(\kappa+1)$ and
	assume (C1)-(C3) and $C >2Ae$. For the CUSP  \eqref{csp}-\eqref{pi-prior} and NMIG \eqref{nmig} priors, the posterior distribution satisfies 
		\[
	\lim_{n \to \infty} \mathbb{E}_0 \left\{ \mathbb{P} \left( R_n> \varepsilon_n \mid y_1,\ldots, y_n \right) \right\} = 0.
\]
\end{theorem}
\begin{proof}
Let $y^{(n)} = (y_1,\ldots, y_n)'$ denote the observed data. For $\eta \in \mathbb{R}^{n}$, let $f_\eta(\cdot)$ denote the probability density function of a $N(\eta,I_n)$ distribution and $f_{\eta_i}(\cdot)$ denote the related univariate marginal. Let $f_0(\cdot)$ be the true density, i.e., a Gaussian distribution with mean $\eta_0$ and identity covariance. Recall that $\eta_{0i} \neq 0$ for $i\leq K_{0n}$ while the remaining  $\eta_{0i}$ for $i> K_{0n}$ are all null. Let $\Pi(\cdot)$ the CUSP prior measure on the $\eta$ parameters.
Let $E_n$ be the target event, 
\[
E_n = \{R_{K_{n}}> \varepsilon_n  \}.
\]
Following \citet{CvdV} and \citet{rock}, we write
\begin{equation}
\mathbb{P} \left( E_n \mid {y}^{(n)} \right)  = \frac{
\int_{E_n} \prod_{i=1}^n \frac{f_{\eta_i}(y_i)}{f_{0i}(y_i)} d\Pi(\eta_i)
}{
\int \prod_{i=1}^n \frac{f_{\eta_i}(y_i)}{f_{0i}(y_i)} d\Pi(\eta_i)
} = \frac{N_n}{D_n}.
\label{eq:NoverD}
\end{equation}
We introduce an event $A_n$  with large probability under the true data generating process. Specifically, let $R = \{h=K_{n0}+1, \dots, n\}$ and define 
\[
A_n := \{ D_n \geq e^{-r_n^2} \mathbb{P}(||\eta_R|| \leq r_n)\}.
\]
If we decompose the probability in \eqref{eq:NoverD} in the sum of two complementary conditional probabilities (conditioning on $A_n$ and its complement $A_n^c$) we can write
\[
\mathbb{E}_0 \left[ \mathbb{P} \left( E_n\mid {y}^{(n)} \right) \right]  \leq 
\mathbb{E}_0 \left[ \mathbb{P} \left( E_n\mid {y}^{(n)} \right) \mathbb{I}_{A_n} \right]+ \mathbb{P}_0(A_n^c)
\]
Now thanks to Lemma 5.2 of \citet{CvdV} the second summand has negligible probability for increasing $n$: $\mathbb{P}_0(A_n^c) \leq \exp\{-r^2_n\}$. Applying \eqref{eq:NoverD}, we have
\[
\mathbb{E}_0 \left[ \mathbb{P} \left( E_n\mid {y}^{(n)} \right) \right]  \leq 
\frac{
	\mathbb{P}(E_n)
}{
	 e^{-r_n^2} \mathbb{P}(||\eta_R|| \leq r_n)} + \exp\{-r^2_n\},
\]
where $\mathbb{P}(E_n)$ is simply the prior probability. 
We can now use Lemma  \ref{lem:priorR}  to bound the numerator and Lemma \ref{lem:denominator} to bound the denominator. Specifically, we have 
$\mathbb{P}(||\eta|| \leq r_n) \geq \mathbb{P}(|\eta_{K_{0n}}| \leq r_n/\sqrt{n})^{n-K_{n0}}$ and hence, Lemma \ref{lem:denominator} yields 
\[
\mathbb{P}(||\eta_R|| \leq r_n) \geq \left(1- \left(\frac{\kappa}{\kappa+1}\right)^{K_{0n}}\right)^{n-K_{n0}} \left(1 - \frac{1}{1+ \frac{a_1r_n^2}{na_2v_0}}\right)^{n-K_{n0}} .
\]
Now, we have $\left(\frac{\kappa}{\kappa+1}\right)^{K_{0n}}< A K_{0n}/n$ for high $n$ and a constant $A>1/2$. Thus, choosing $r^2_n = K_{0n}$ and considering (C1),  we can write
\[
\mathbb{P}(||\eta_R|| \leq r_n) \geq \left(1- \frac{A K_{0n}}{n}\right)^{n} \left(1 - \frac{1}{n}\right)^{n} \geq \left(\frac{1}{2e} \right)^{AK_{0n}+1},
\]
where the last inequality follows from $(1 - x)^{1/x} > 1/(2e)$ for $0 < x < 0.5$. Finally, using  Lemma \ref{lem:priorR}, we get 
\[
\mathbb{E}_0 \left[ \mathbb{P} \left( E_n\mid {y}^{(n)} \right) \right]  \leq 
2\exp\{-K_{0n} [C-1-A\log(2e)]+1 \} + \exp\{-K_{0n}\}.
\]
Then, since $C > 2Ae > 1 + A \log(2e)$ for $A > 1/2$, we conclude $\mathbb{E}_0\{\mathbb{P} \left( E_n \mid {y}^{(n)} \right) \} \to 0$.
\end{proof}

\begin{proposition} 
For $\varepsilon>0$, $\mathbb{P}(|\beta_{k,i}| \le \varepsilon) < \mathbb{P}(|\beta_{k+1,i}| \le \varepsilon)$ whenever $v_0 < 1$.
\end{proposition} 
\begin{proof}
Since  $\{\xi_{k,i}\}$ are iid (upon marginalization over $\{m_{\xi_{k,i}}\}$), it follows that $\mathbb{P}(|\beta_{k,i}| \le \varepsilon) = \mathbb{P}(|\eta_k\xi_{k,i}| \le \varepsilon) = \mathbb{P}(|\eta_k\xi_{k+1,i}| \le \varepsilon)$. Then we proceed directly: 
\begin{align*}
\mathbb{P}(|\beta_{k,i}| \le \varepsilon) &= \mathbb{P}(|\eta_k\xi_{k+1,i}| \le \varepsilon) = \mathbb{E}\{ \mathbb{I}(|\eta_k| |\xi_{k+1,i}| \le \varepsilon)\} \\
&= \mathbb{E}\big[  \mathbb{E}\big\{ \mathbb{I}(|\eta_k| |\xi_{k+1,i}| \le \varepsilon) \big\vert \xi_{k+1,i} \big\} \big] \\
&= \mathbb{E}\big[  \mathbb{P}\big(|\eta_k| |\xi_{k+1,i}|  \le  \varepsilon  \big\vert \xi_{k+1,i}\big)\big] \\ 
& < 
 \mathbb{E}\big[  \mathbb{P}\big(|\eta_{k+1}| |\xi_{k+1,i}|  \le  \varepsilon  \big\vert \xi_{k+1,i}\big)\big] \\
 &= \mathbb{P}(|\beta_{k+1,i}| \le \varepsilon),
\end{align*}
where the inequality follows from Corollary~\ref{eta-shrink}.
\end{proof}

\begin{proposition} 
Let $\theta^{(K)} =  \{\theta_k\}_{k=1}^K$ denote the sequence $\{\theta_k\}_{k=1}^\infty$ truncated at $K$. For $0 < v_0 < \varepsilon < 1$, we have $\mathbb{P}\left\{d_\infty(\theta, \theta^{(K)}) > \varepsilon\right\}  \le \kappa \{\kappa/(1+\kappa)\}^K$. 
\end{proposition}
\begin{proof}
For the prior \eqref{nmig}, note that $P_{slab} = \delta_1$ and $P_{spike} = \delta_{v_0}$, so $P_{slab}\{\bar{\mathbb B}_\varepsilon(0)\} = 1$ if $\varepsilon < 1$ and  $P_{spike}\{\bar{\mathbb B}_\varepsilon(0)\} = 0$ if $\varepsilon > v_0$, where $\bar{\mathbb B}_\varepsilon(0)$ denotes the complement of $\mathbb B_\varepsilon(0)$. Using the proof of Proposition~\ref{csp-shrink}, we have
$ \mathbb{P}\left\{d_\infty(\theta, \theta^{(K)}) > \varepsilon\right\}  = \mathbb{P}\left\{ \sup_{k > K} |\theta_k| > \varepsilon\right\} 
\le \sum_{k > K} \mathbb{P}\left(  |\theta_k| > \varepsilon\right)= \sum_{k > K} \{\kappa/(1+\kappa)\}^k = \kappa \{\kappa/(1+\kappa)\}^K$. 
\end{proof}

\section{MCMC algorithm}\label{a-mcmc}

We outline the Gibbs sampling algorithm for the semiparametric functional factor model (SFFM). Specifically, we present full conditional updates for the model with the following priors on the parametric components:
$$
[ \alpha_{\ell,i} | \sigma_{\alpha_\ell}] \stackrel{indep}{\sim} N(0, \sigma_{\alpha_\ell}^2), \quad \sigma_{\alpha_\ell} \stackrel{iid}{\sim} C^+(0,1), \quad p(\sigma_\epsilon^2) \propto 1/\sigma_\epsilon^2,
$$
and assume a generic prior $p(\gamma)$ for the nonlinear parameter $\gamma$ (if unknown). The sampling step for $\gamma$ is omitted if there are no unknown nonlinear parameters.

\subsection*{Gibbs sampling algorithm:}
\begin{enumerate}

\item {\bf Imputation:} for any unobserved $\tau^*$ for functional observation $i$, sample $$[y_i(\tau^*) | -] \stackrel{indep}{\sim} N\Big(\sum_{\ell=1}^L g_\ell(\tau^*; \gamma) \alpha_{\ell,i} + \sum_{k=1}^K f_k(\tau^*) \beta_{k,i}, \sigma_\epsilon^2\Big);$$

\item {\bf Nonlinear parameter:}  if $\gamma$ is unknown, sample  $$
p(\gamma | -) \propto p(\gamma) \exp(-\frac{1}{2} \sum_{i=1}^n || \bm y_i - \bm G_\gamma \bm \alpha_i - \bm F \bm \beta_i||^2/ \sigma_\epsilon^2 )$$ using a slice sampler \citep{neal2003slice};

\item {\bf Nonparametric parameters:} for $k=1,\ldots,K$ (in random order),
\begin{enumerate}
\item Sample the unconstrained coefficients 
$$\bm\psi_k^0 \sim N\left(\bm Q_{\psi_k}^{-1} \bm \ell_{\psi_k}, \bm Q_{\psi_k}^{-1}\right)$$  where
$\bm Q_{\psi_k} =  \sigma_\epsilon^{-2}(\bm B'\bm B)\sum_{i=1}^n \beta_{k,i}^2 +   \lambda_{f_k} \bm \Omega $ and $\bm \ell_{\psi_k} =   \sigma_\epsilon^{-2}\bm B' \sum_{i=1}^n\big\{\beta_{k,i} \big( \bm y_i  - \bm G_\gamma \bm\alpha_i - \sum_{\ell \ne k} \bm f_\ell \beta_{\ell,i}\big)\big\}$ (see \citealp{kowal2019bayesianfosr} for efficient techniques);

\item Compute the constrained vector 
$$
\bm \psi_k = \bm \psi_k^0 - \bm Q_{\psi_k}^{-1} \bm C_k' \left(\bm C_k\bm Q_{\psi_k}^{-1}\bm C_k'\right)^{-1}\bm C_k \bm \psi_k^0
$$
where $\bm C_k = (\bm G_\gamma, \bm f_1, \ldots, \bm f_{k-1}, \bm f_{k+1}, \ldots, \bm f_K)' \bm B $;

\item Normalize: $\beta_{k,i} \rightarrow \beta_{k,i} \Vert \bm B \bm \psi_k\Vert$, $\bm \psi_k \rightarrow \bm \psi_k \Vert \bm B \bm \psi_k\Vert$, and $\bm f_k =  \bm B \bm \psi_k$;

\item Sample the smoothing parameter
$$
[\lambda_{f_k} | -] \sim \mbox{Gamma}\big( (J + 1)/2, \bm \psi_k' \bm \Omega \bm \psi_k/2\big) \mbox{ truncated to } (10^{-8}, \infty)
$$
where $J$ is the number of columns of $\bm B$;

\end{enumerate}

\item {\bf Ordered spike-and-slab parameters:} let  $y_{k,i}^F = \bm f_k'\bm y_i$ for $i=1,\ldots,n$ and $k=1,\ldots,K$;

\begin{enumerate}
\item Sample $[m_{\xi_{k,i}} \vert -]$ from $\{-1,1\}$ with $\mathbb{P}(m_{\xi_{h,i}} = 1 \vert -) = 1/\{1 + \exp(-2\xi_{h,i})\}$; 

\item Sample $[\xi_{k,i} \vert -] \sim N(Q_{\xi_{k,i}}^{-1} \ell_{\xi_{k,i}},   Q_{\xi_{k,i}}^{-1})$ where $Q_{\xi_{k,i}} = \eta_k^2/\sigma_\epsilon^2 + 1$ and $\ell_{\xi_{k,i}} = \eta_k y_{k,i}^F/\sigma_\epsilon^2 + m_{\xi_{k,i}} $;

\item Sample $[\eta_{k} \vert -] \sim N(Q_{\eta_k}^{-1} \ell_{\eta_k},   Q_{\eta_k}^{-1} )$ where $Q_{\eta_k} = \sum_{i=1}^n \xi_{k,i}^2/\sigma_\epsilon^2 + (\theta_k \sigma_k^2)^{-1}$ and $\ell_{\eta_k} = \sum_{i=1}^n \xi_{k,i} y_{k,i}^F/\sigma_\epsilon^2$;

\item Rescale $\eta_k \rightarrow (\sum_{i=1}^n \vert\xi_{k,i}\vert/n) \eta_k$
 and $\bm \xi_k \rightarrow (n/\sum_{i=1}^n \vert\xi_{k,i}\vert) \bm \xi_k$ 
 and update  $\beta_{k,i} = \xi_{k,i}\eta_k$;

\item Sample $[\sigma_{k}^{-2} | -] \sim \mbox{Gamma}\left\{a_1 + 1/2, a_2 + \eta_k^2/(2\theta_k)\right\}$; 

\item Sample  $[\nu_k | -] \sim \mbox{Beta}(1 + \sum_{h=1}^K \mathbb{I}\{z_h = k\}, \kappa + \sum_{h=1}^K \mathbb{I}\{z_h > k\})$ for $k=1,\ldots,K-1$ 
and update $ \omega_h = \nu_h \prod_{\ell=1}^{h-1} (1-\nu_\ell)$ and $\pi_k = \sum_{h=1}^k \omega_h$;

\item Sample $[\kappa | -] \sim \mbox{Gamma}\{a_{\kappa}+ K-1, b_{\kappa} - \sum_{k=1}^{K-1} \log(1 - \nu_k)\}$;

\item Sample $[z_k | -]$ from
$$
\mathbb{P}(z_k=h | -) 
\propto 
\begin{cases}
\omega_h t_{2a_1}( \eta_k; 0, \sqrt{v_0a_2/a_1})&  h \le k \\
\omega_h t_{2a_1}( \eta_k; 0, \sqrt{a_2/a_1}) & h > k 
\end{cases}
$$
where $t_d(x; m, s)$ is the density of the $t$-distribution evaluated at $x$ with mean $m$, standard deviation $s$, and degrees of freedom $d$;

\item {\bf Update} $\theta_k =  1$ if $z_k >k$ and $\theta_k = v_0$ if $z_k \le k$.  
\end{enumerate}

\item {\bf Parametric linear coefficients:} sample  $$[\bm \alpha_i | -] \stackrel{indep}{\sim} N\big(\bm Q_{\alpha_i}^{-1} \bm \ell_{\alpha_i}, \bm Q_{\alpha_i}^{-1}\big), \quad i=1,\ldots,n$$ where  $\bm Q_{\alpha_i} = \sigma_\epsilon^{-2} \bm I_L + \mbox{diag}(\{\sigma_{\alpha_\ell}^{-2}\}_{\ell=1}^L)$ and $\bm \ell_{\alpha_i} = \sigma_\epsilon^{-2} \bm G_\gamma'\bm y_i$;

\item {\bf Parametric variance:} using a parameter-expansion of the half-Cauchy distribution \citep{wand2011mean}, 
\begin{enumerate}
\item Sample $[\sigma_{\alpha_\ell}^{-2} | -,  \xi_{\sigma_{\alpha_\ell}}] \stackrel{indep}{\sim} \mbox{Gamma}\big( (n+1)/2, \sum_{i=1}^n \alpha_{\ell,i}^2/2 + \xi_{\sigma_{\alpha_\ell}}\big)$ for $\ell = 1,\ldots,L$;
\item Sample $[\xi_{\sigma_{\alpha_\ell}} | \sigma_{\alpha_\ell}] \stackrel{indep}{\sim} \mbox{Gamma}\big(1, \sigma_{\alpha_\ell}^{-2} + 1\big)$ for $\ell = 1,\ldots,L$;
\end{enumerate}

\item {\bf Observation error variance:} sample 
$$[\sigma_\epsilon^{-2} | - ] \sim \mbox{Gamma}\Big(nm/2, \sum_{i=1}^n ||\bm y_i - \bm G_\gamma \bm\alpha_i - \bm F \bm \beta_i||^2/2\Big).$$

 \end{enumerate}

\begin{remark}
Recall that we orthogonalize $\bm G_\gamma'\bm G_\gamma = \bm I_L$ using a QR decomposition $\bm G_\gamma^0 = \bm Q_\gamma \bm R_\gamma$   of the initial basis matrix   $\bm G_{\gamma}^0 = (\bm g_{1; \gamma},\ldots, \bm g_{L; \gamma})$, and set  $\bm G_\gamma = \bm Q_\gamma$. When $\gamma$ is unknown and endowed with a prior distribution, the QR decomposition is incorporated into the likelihood evaluations of \eqref{like} for posterior sampling of $\gamma$. In addition, the parametric factors may be recovered on the original scale by setting $\bm \alpha_i^0 = \bm R_\gamma^{-1} \bm \alpha_i$, which can be computed draw by draw within the MCMC sampler.

\end{remark}

\subsection*{Hierarchical prior on the nonlinear parameter}
For the pinch force data (Section~\ref{pinch}), we include the hierarchical prior 
$$
\gamma_i \stackrel{iid}{\sim} N(\mu_\gamma, \sigma_\gamma^2), \quad 
\mu_\gamma \sim N(0, 10), \quad 
\sigma_\gamma \sim C^+(0,1)$$
which requires minor modifications to the Gibbs sampling algorithm. The nonlinear parameter is still sampled using the slice sampler \citep{neal2003slice}, but with the modified likelihood
$$
p(\gamma_i | -) \propto p(\gamma_i | \mu_\gamma, \sigma_\gamma^2) \exp(-\frac{1}{2} || \bm y_i - \bm G_\gamma \bm \alpha_i - \bm F \bm \beta_i||^2/ \sigma_\epsilon^2 )
$$
specific to each $i=1,\ldots,n$, while the (conditional) prior is the aforementioned Gaussian distribution. The mean parameter is sampled from its full conditional distribution,
$$
[\mu_\gamma | -] \sim N\big(Q_{\mu_\gamma}^{-1}\ell_{\mu_\gamma}, Q_{\mu_\gamma}^{-1}\big),
$$
where $Q_{\mu_\gamma} = n\sigma_\gamma^{-2} + 1/10$ and $\ell_{\mu_\gamma} = \sigma_\gamma^{-2} \sum_{i=1}^n \gamma_i$. The sampler for the standard deviation parameter uses a half-Cauchy parameter-expansion  \citep{wand2011mean}:
\begin{align*}
[\sigma_\gamma^{-2} | -,  \xi_{\sigma_\gamma}] &\sim \mbox{Gamma}\big( (n+1)/2, \sum_{i=1}^n (\gamma_i - \mu_\gamma)^2/2 + \xi_{\sigma_\gamma}\big)  \\
[\xi_{\sigma_\gamma} |\sigma_\gamma ] &\sim \mbox{Gamma}\big(1, \sigma_{\gamma}^{-2} + 1\big).
\end{align*}

\subsection*{Dynamic models for the parametric factors}
Since yield curve data (Section~\ref{yields}) are time-ordered for $i=1,\ldots,n$, we specified a dynamic model for the parametric factors
$$
\alpha_{\ell,i} = \mu_\ell + \phi_\ell (\alpha_{\ell, i-1} - \mu_\ell) + \zeta_{\ell,i}, \quad \zeta_{\ell,i} \stackrel{indep}{\sim} N(0, \sigma_{\zeta_\ell}^2)
$$
with the additional priors
$$
\mu_\ell \stackrel{iid}{\sim} N(0, 10^6), \quad (\phi_\ell + 1)/2 \stackrel{iid}{\sim} \mbox{Beta}(5,2), \quad \sigma_{\zeta_\ell} \stackrel{iid}{\sim} C^+(0,1)
$$
for $\ell = 1,\ldots,L$. Conditional on $\{\alpha_{\ell, i}\}$, the sampling steps for these parameters are straightforward; see \cite{Kowal2020b} for details. Note that the prior on the (shifted and scaled) autoregression coefficients implies that the dynamic factors $\{\alpha_{\ell,i}\}$ are stationary, and therefore so is the functional time series $\{Y_i\}$.

The dynamic factors can be sampled using state space simulation methods; we use \cite{durbin2002simple} implemented in the \texttt{KFAS} package in \texttt{R}. Specifically, Corollary~\ref{cor-alpha} implies that the observation equation reduces to 
\[
\bm g_{\ell; \gamma}'\bm y_i =   \alpha_{\ell,i} +  \epsilon_{\ell,i}, \quad  \epsilon_{\ell,i}  \stackrel{indep}{\sim}N(0, \sigma_{\epsilon_i}^2)
\]
for $\ell=1,\ldots,L$ and $i=1,\ldots,n$, while the evolution equation is given by the independent AR(1) models for the dynamic factors $\{\alpha_{\ell,i}\}_{i=1}^n$ for each $\ell=1,\ldots,L$. Using the state space construction, we can sample the factors $\{\alpha_{\ell,i}\}$ jointly in $O(nL)$ computational complexity. 

The yield curve model also includes a stochastic volatility model on the observation error variance. The priors and sampling steps are identical to those in \cite{Kowal2020b}.

\subsection*{Modifications for the Gaussian process alternative}
The PFFM+gp \eqref{pffm+gp} was introduced as a seemingly reasonable alternative to \eqref{sffm} that replaced the nonparametric term with a Gaussian process, $Y_i(\tau)  = \sum_{\ell=1}^L \alpha_{\ell,i} g_\ell(\tau; \gamma) + h_i(\tau)$, where $ h_i  \stackrel{iid}{\sim} \mathcal{GP}(0, \mathcal{K}_h)$ is a Gaussian process with mean zero and covariance function $\mathcal{K}_h$. 
Equivalently, each latent curve $Y_i$ is a Gaussian process centered at the parametric term, $Y_i \stackrel{indep}{\sim} \mathcal{GP}\{\sum_{\ell=1}^L \alpha_{\ell,i} g_\ell(\cdot; \gamma), \mathcal{K}_h\}$. 
Using the observation model \eqref{obs} with the evaluation points  $\{\tau_j\}_{j=1}^m$, the PFFM+gp modifies \eqref{like}  to be 
\[
\bm y_i = \bm G_\gamma \bm \alpha_i + \bm h_i + \bm \epsilon_i, \quad \bm \epsilon_i \stackrel{indep}{\sim}N(\bm 0, \sigma_{\epsilon_i}^2 \bm I_m)
\]
where $\bm h_i = (h_i(\tau_1),\ldots, h_i(\tau_m))' \sim N( \bm 0, \bm K_h)$  and $\bm K_h =\{\mathcal{K}_h(\tau_j, \tau_\ell)\}_{j,\ell=1}^m$. Informally, each $\bm h_i$ captures the functional (within-curve) variability that is unexplained by the parametric term.
We implement this model using a Mat\'ern covariance function with smoothness parameter 2.5, i.e., $\mathcal{K}_h(\tau_1, \tau_2) = \sigma_h^2 [1 + \sqrt{5}\vert \tau_1 - \tau_2\vert/\rho_h + 5 \vert\tau_1 - \tau_2\vert^2/(3\rho_h^2)\exp\{-\sqrt{5}\vert \tau_1 - \tau_2 \vert/\rho_h\}]$. We assign the prior  $\sigma_h \sim C^+(0,1)$ and  fix $\rho_h = 0.5$ to provide smooth paths and reduce the computational burden; the results are not sensitive to this choice.

 To adapt the SFFM Gibbs sampler for the PFFM+gp, we only require modifications of the sampler for $\{\bm\alpha_i\}$ and sampling steps for the Gaussian process parameters. Specifically, we sample $\{\bm h_i\}$ after integrating out the parametric factors $\{\bm \alpha_i\}$, which induces a joint draw of $\{\bm \alpha_i, \bm h_i\}$ to encourage MCMC efficiency. Let $\bm  R_h = \sigma_h^{-2}\bm K_h$ denote the (fixed) correlation matrix. Assuming $\alpha_{\ell,i} \stackrel{indep}{\sim}N(0, \sigma_{\alpha_\ell}^2)$, the requisite (marginalized) full conditional distributions are $[\bm h_i \mid - ] \stackrel{indep}{\sim} N(\bm Q_{h_i}^{-1} \bm \ell_{h_i}, \bm Q_{h_i}^{-1})$ for $i=1,\ldots,n$, where $ \bm Q_{h_i} = \sigma_h^{-2} \bm R_h^{-1} + \bm \Sigma_{h_i}^{-1}$, $\bm\ell_{h_i} = \bm\Sigma_{h_i}^{-1} \bm y_i$, and the marginal precision $ \bm\Sigma_{h_i}^{-1} = (\bm G_\gamma \bm \Sigma_\alpha \bm G_\gamma' + \sigma_\epsilon^2 \bm I_m)^{-1} = \sigma_\epsilon^{-2} \{\bm I_m - \bm G_\gamma ( \mbox{diag}\{\sigma_\epsilon^2/\sigma_{\alpha_\ell}^{2} + 1\}_{\ell=1}^L) \bm G_\gamma'\}$ is simplified due to the Woodbury identity and the orthogonality of $\bm G_\gamma$. The sampling steps for $\sigma_h$ proceed using the same parameter expansion from \cite{wand2011mean} as in Section~\ref{a-mcmc}. 
 
 We emphasize that \emph{without} marginalizing over $\{\bm \alpha_i\}$, the MCMC efficiency for both $\{\bm \alpha_i\} $ and $\{ \bm h_i\}$ deteriorates significantly. Hence, our implementation---using both the joint sampling step and the Woodbury simplifications above---presents a very favorable sampling algorithm for the PFFM+gp. Nonetheless, both the computational efficiency and the MCMC efficiency of the PFFM+gp lag far behind those for the SFFM and the PFFM, both in simulations and real data analyses. Lastly, we note that the key marginalizations needed for joint sampling of $\{\bm \alpha_i, \bm h_i\}$ become significantly more challenging when the model for the parametric factors is more complex (e.g., dynamic); by comparison, the proposed SFFM avoids this issue entirely.


\section{Additional simulation results}\label{a-sims}

\subsection{Linear template: $K_{true} = 8$}
To assess performance in the presence of many nonparametric factors, we reproduce the simulation analysis from Section~\ref{sims} in the case of $K_{true} = 8$. For the proposed SFFM, we increase the bound on the number of factors from $K=10$ to $K=15$. Figure~\ref{fig:linear-8} consolidates the point and interval estimation and prediction
 for the linear template; similar results are presented for the Nelson-Siegel template in Section~\ref{a-ns}. The results are similar to those for $K_{true} =3$ in the main paper, and confirm the improvements offered by the SFFM when $K_{true} > 0$. 

  \begin{figure}[h]
\begin{center}
\includegraphics[width=.49\textwidth]{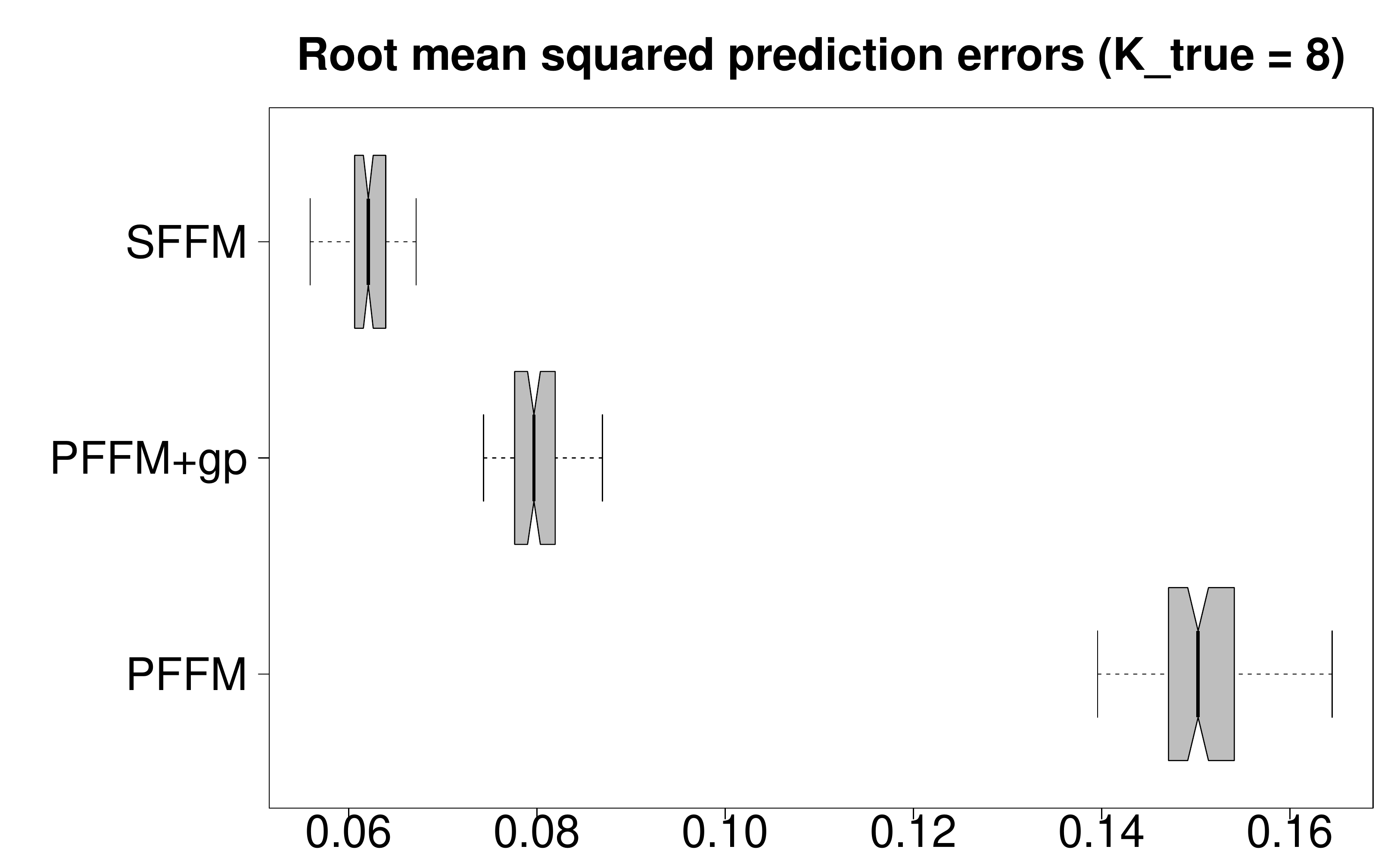}
\includegraphics[width=.49\textwidth]{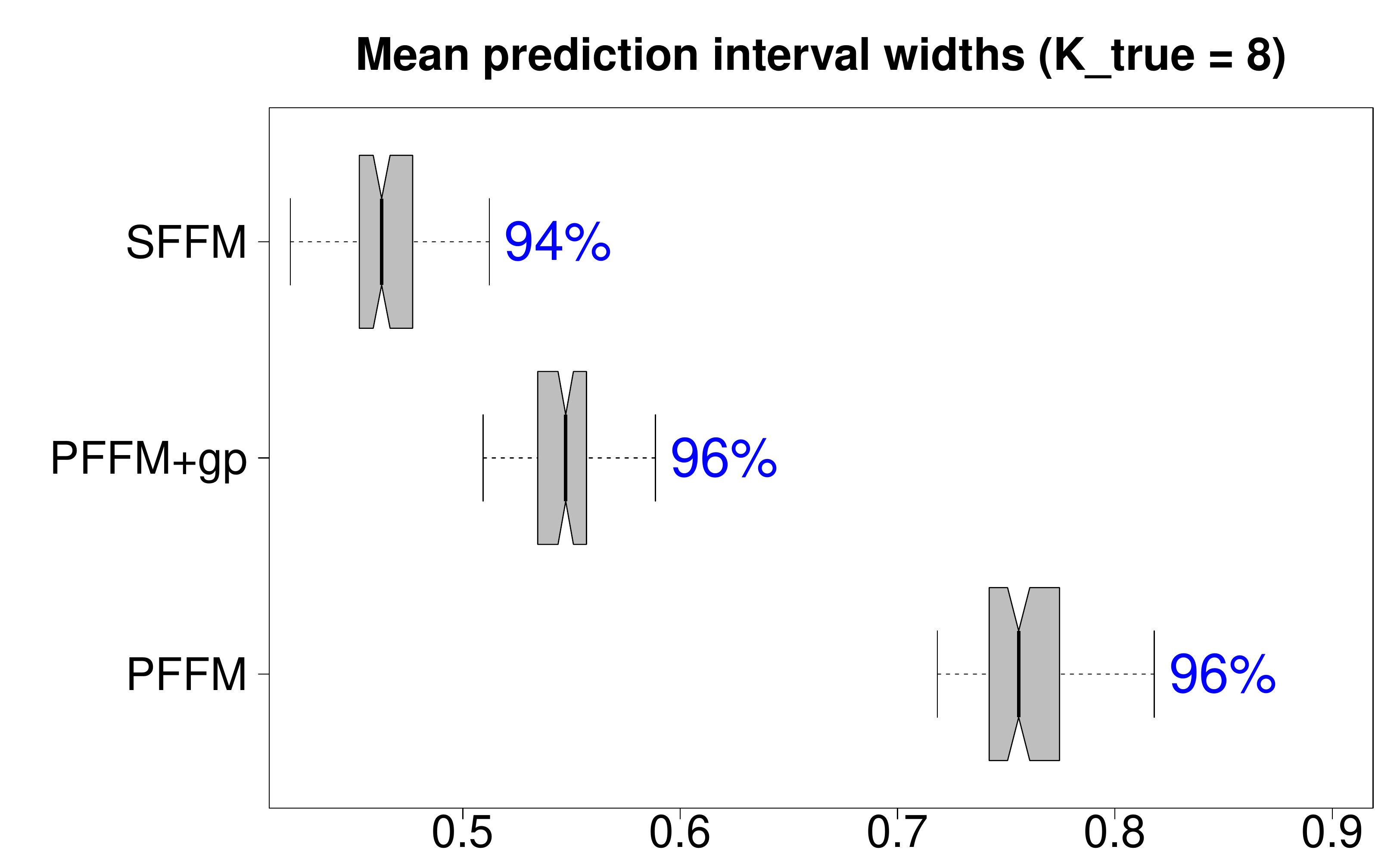}
\includegraphics[width=.49\textwidth]{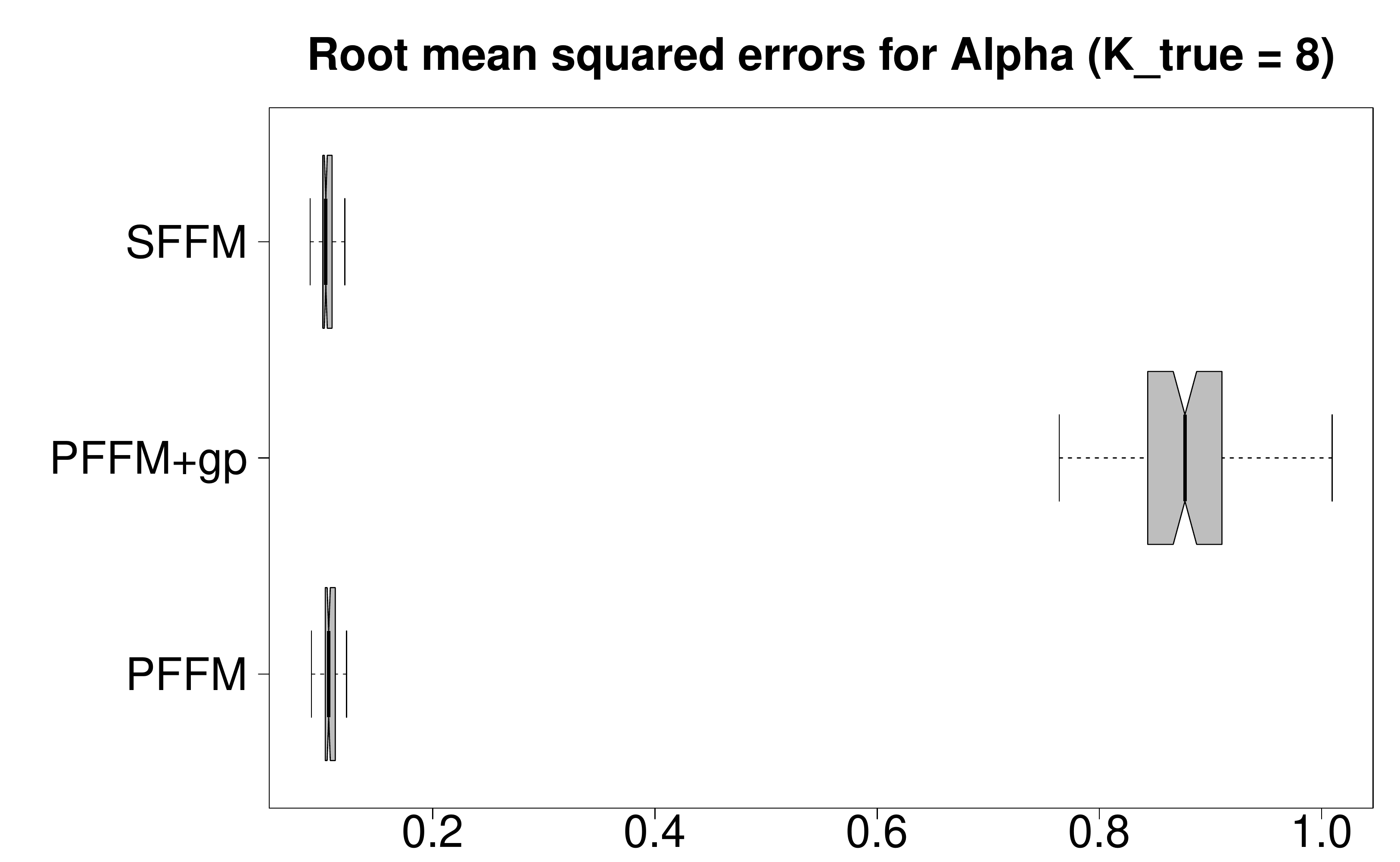}
\includegraphics[width=.49\textwidth]{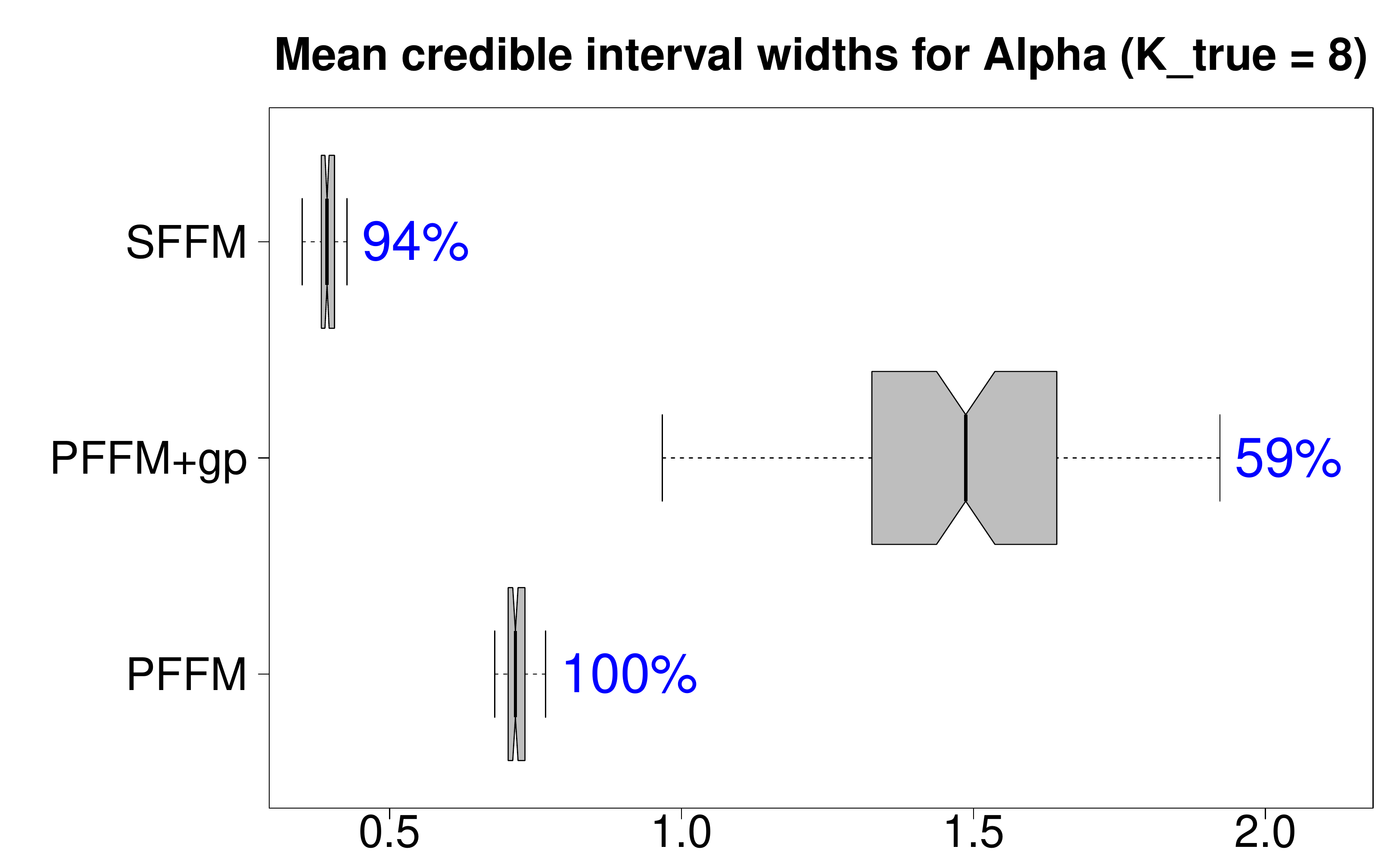}
\caption{\small Assessments of point prediction (left) and interval prediction (right)  for $K_{true}  = 8$. Results are similar to those for $K_{true} = 3$. 
\label{fig:linear-8}}
\end{center}
\end{figure}

\subsection{Nelson-Siegel template}\label{a-ns}
We augment our simulation study to include the \cite{nelson1987parsimonious} template, 
$$
g_1(\tau; \gamma) = 1, \quad
 g_2(\tau; \gamma) = \{1 - \exp(-\tau\gamma)\}/(\tau\gamma), \quad
  g_3(\tau; \gamma) =  g_2(\tau; \gamma)  - \exp(-\tau\gamma)
$$
with unknown $\gamma > 0$. The synthetic data-generating process proceeds as before, with following modifications: the true nonlinear parameter is set to $\gamma^* = 0.0609$ \citep{diebold2006forecasting,bianchi2009great} and the nonparametric basis matrix $\bm F^*$ is orthogonalized to the parametric basis $\bm G_{\gamma^*}^*$ using a QR-decomposition. As in the yield curve application, we use a Gamma prior for $\gamma$ with prior mean 0.0609 and prior variance 0.5. The results are given in Figures~\ref{fig:ns-rmse}, \ref{fig:ns-int}, and \ref{fig:ns-kprob} and confirm the results for the linear template in the main paper.

   \begin{figure}[h]
\begin{center}
\includegraphics[width=.49\textwidth]{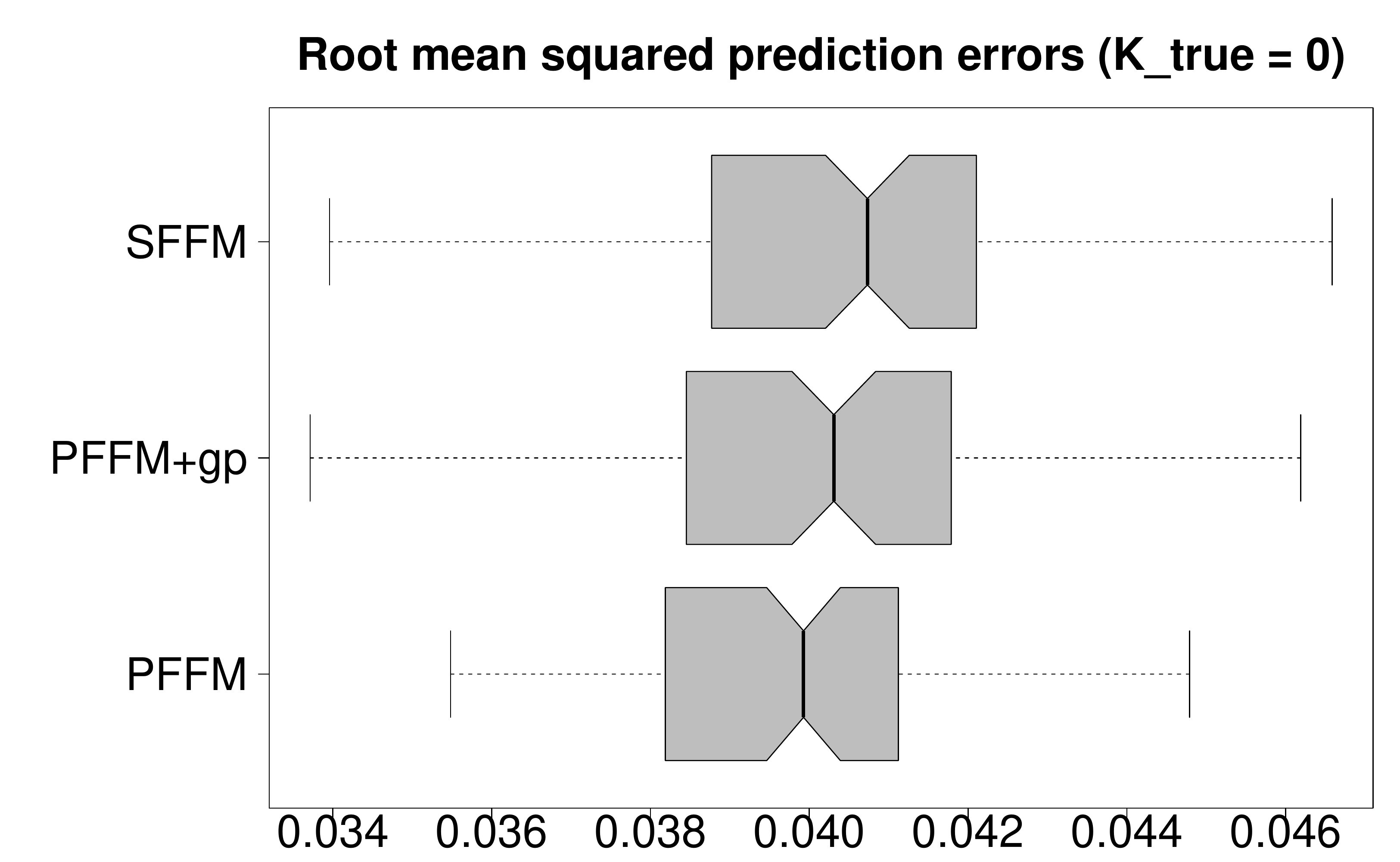}
\includegraphics[width=.49\textwidth]{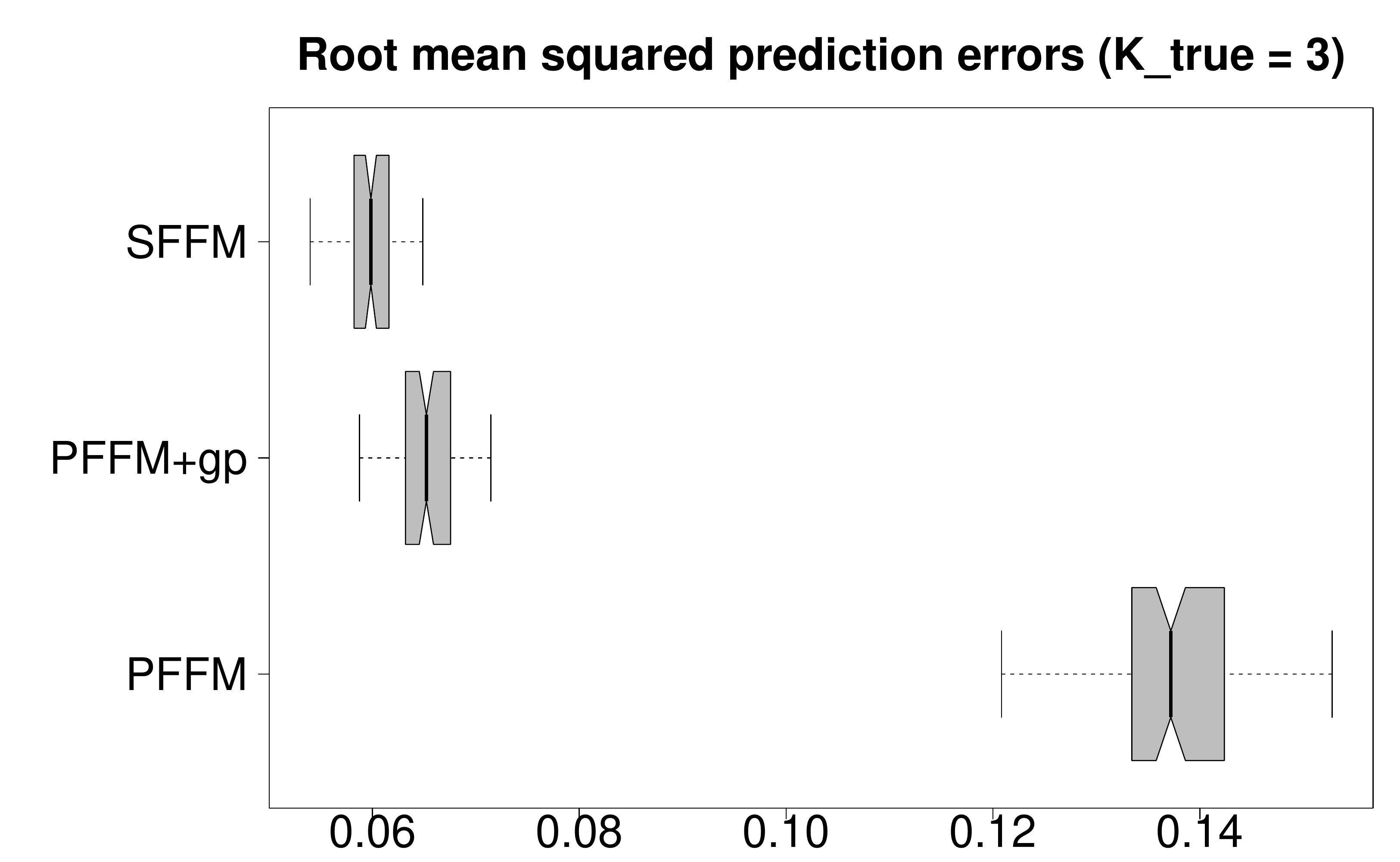}
\includegraphics[width=.49\textwidth]{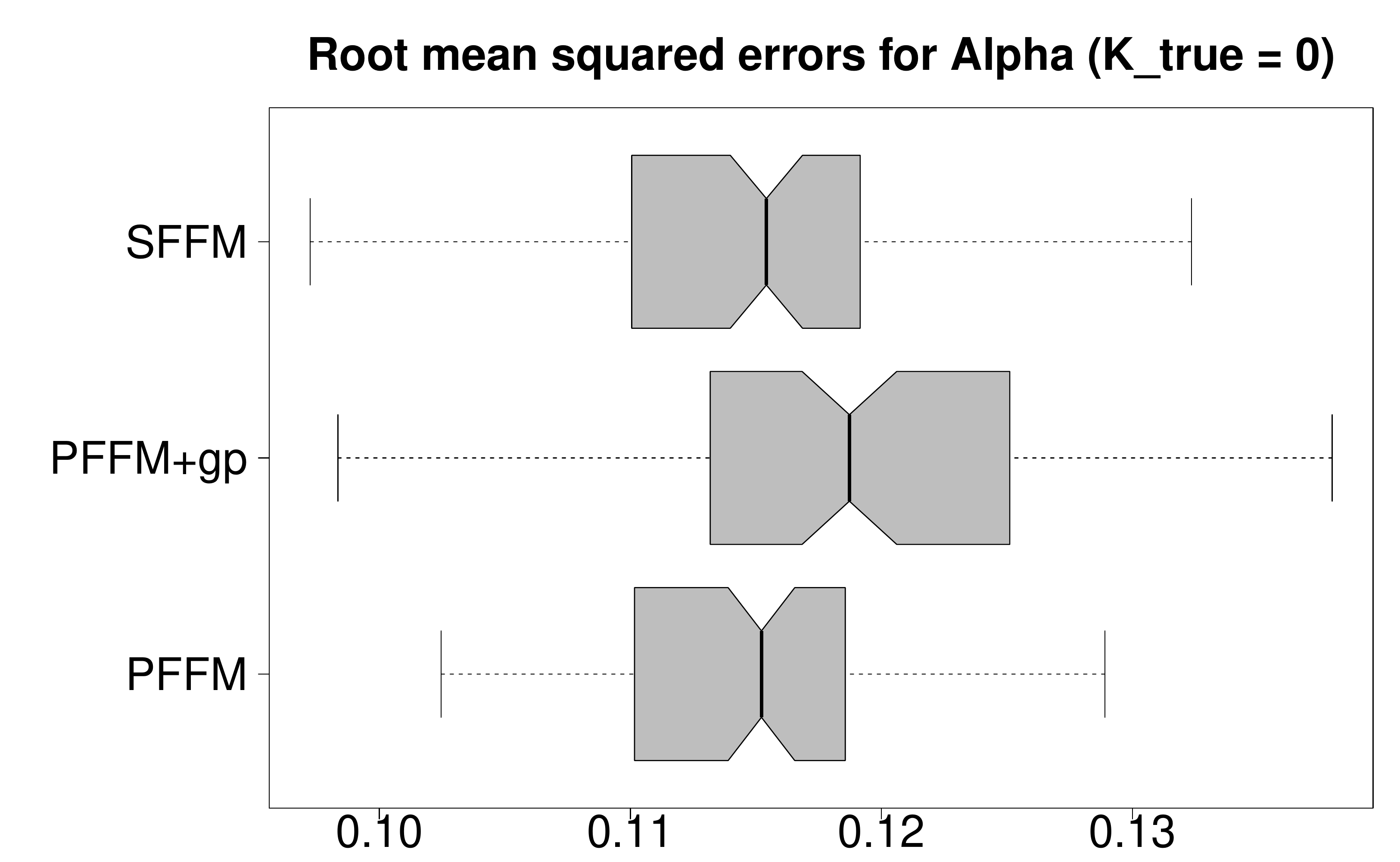}
\includegraphics[width=.49\textwidth]{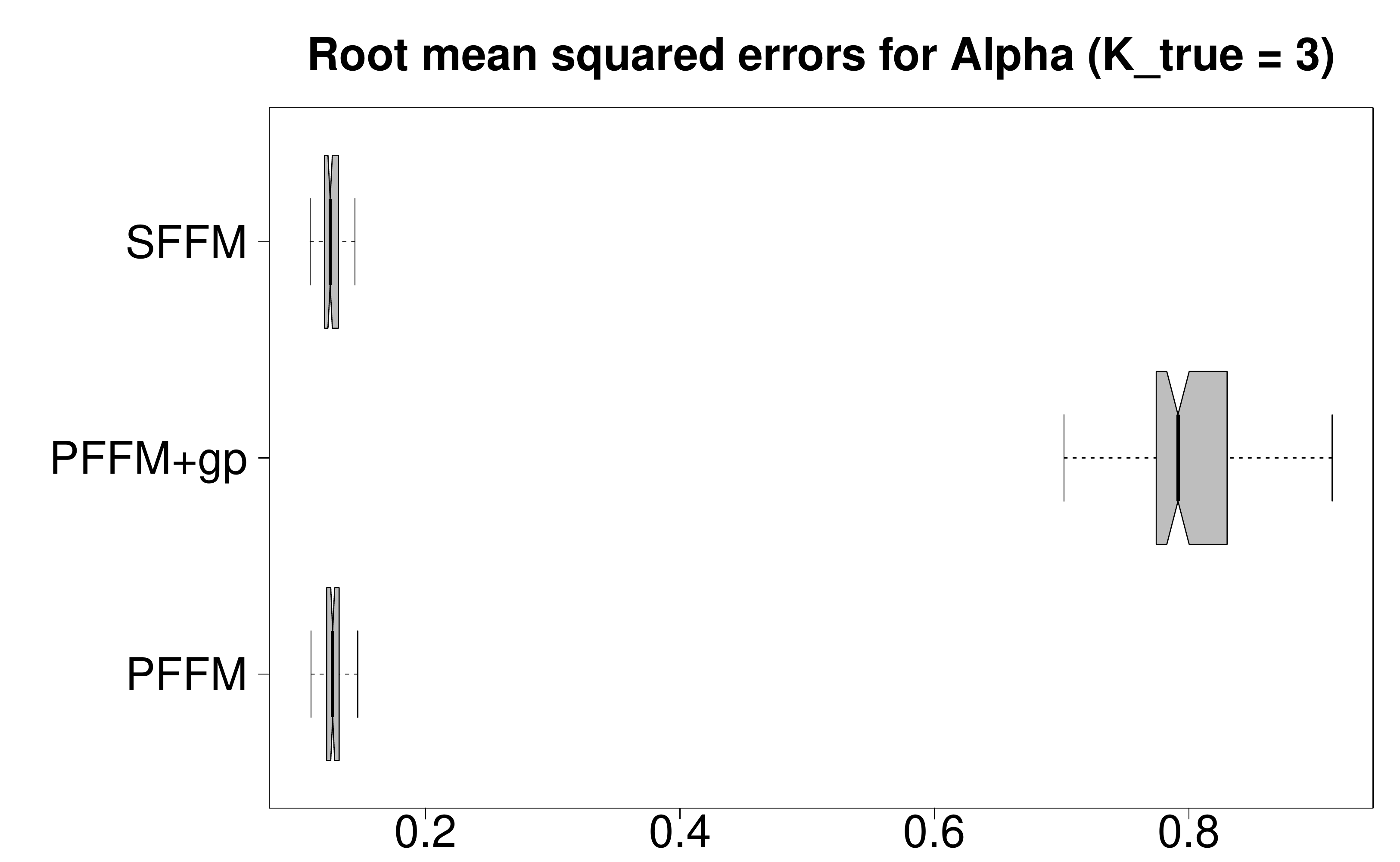}
\caption{\small Root mean squared (prediction) errors for $\bm Y^*$ (top) and   $\{\alpha_{\ell, i}^*\}$   (bottom) for $K_{true} = 0$ (left) and $K_{true} = 3$ (right) under the \cite{nelson1987parsimonious} template with unknown $\gamma$. All methods perform similarly for $K_{true} = 0$, while the SFFM outperforms all competitors when $K_{true}  >  0$. In this setting, the PFFM+gp offers better prediction than the  PFFM, but cannot accurately estimate the parametric factors.
\label{fig:ns-rmse}}
\end{center}
\end{figure}

   \begin{figure}[h]
\begin{center}
\includegraphics[width=.49\textwidth]{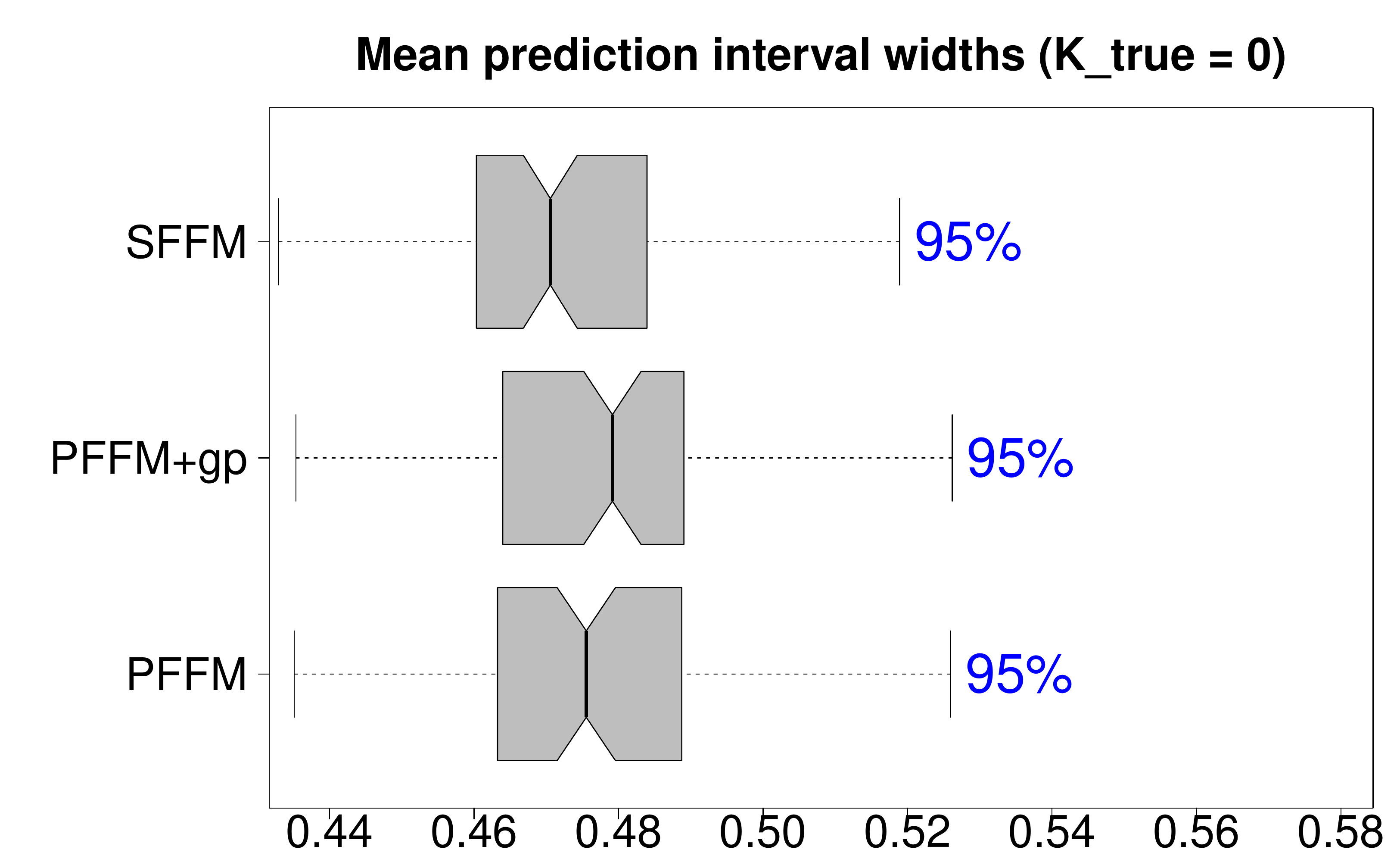}
\includegraphics[width=.49\textwidth]{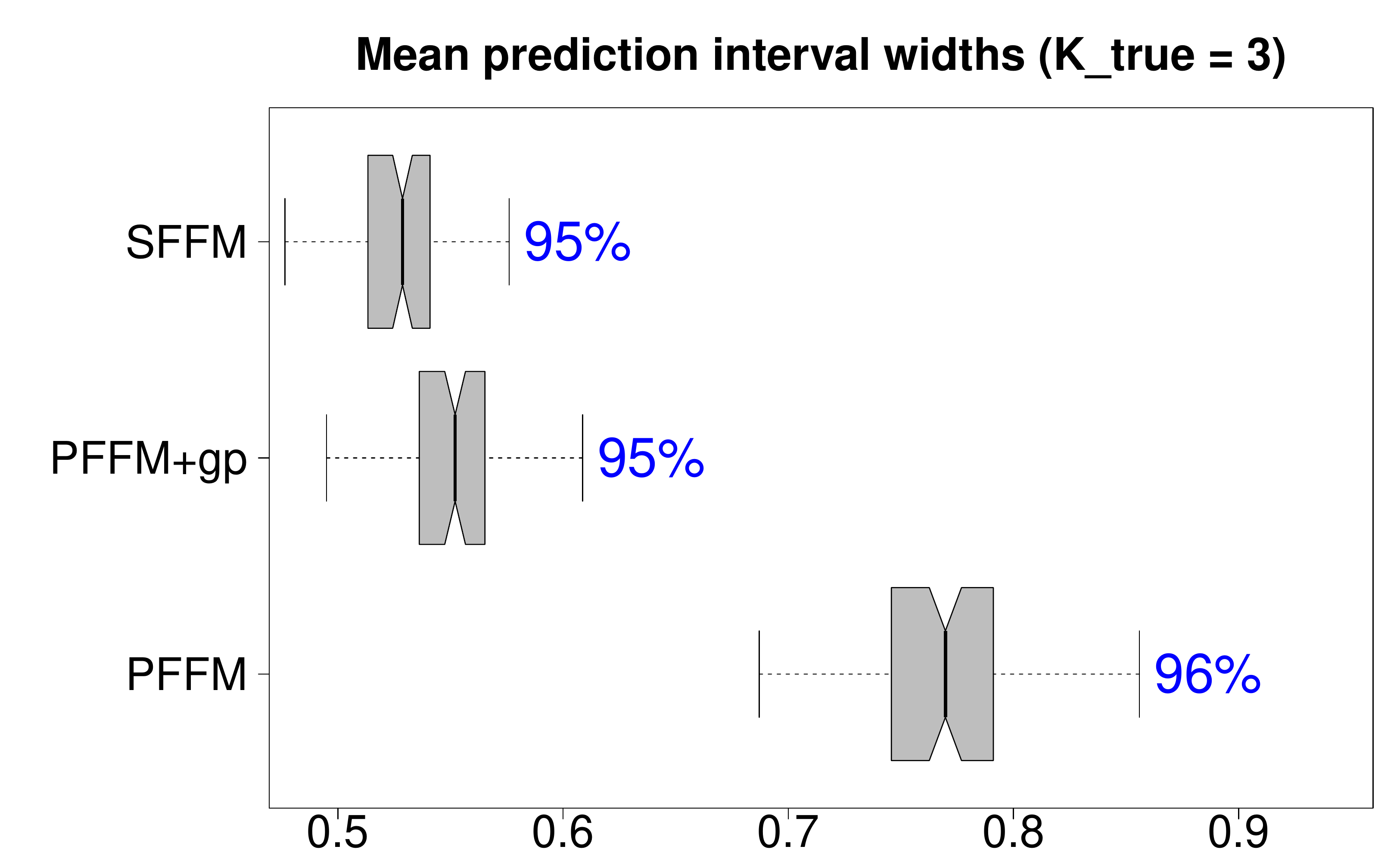}
\includegraphics[width=.49\textwidth]{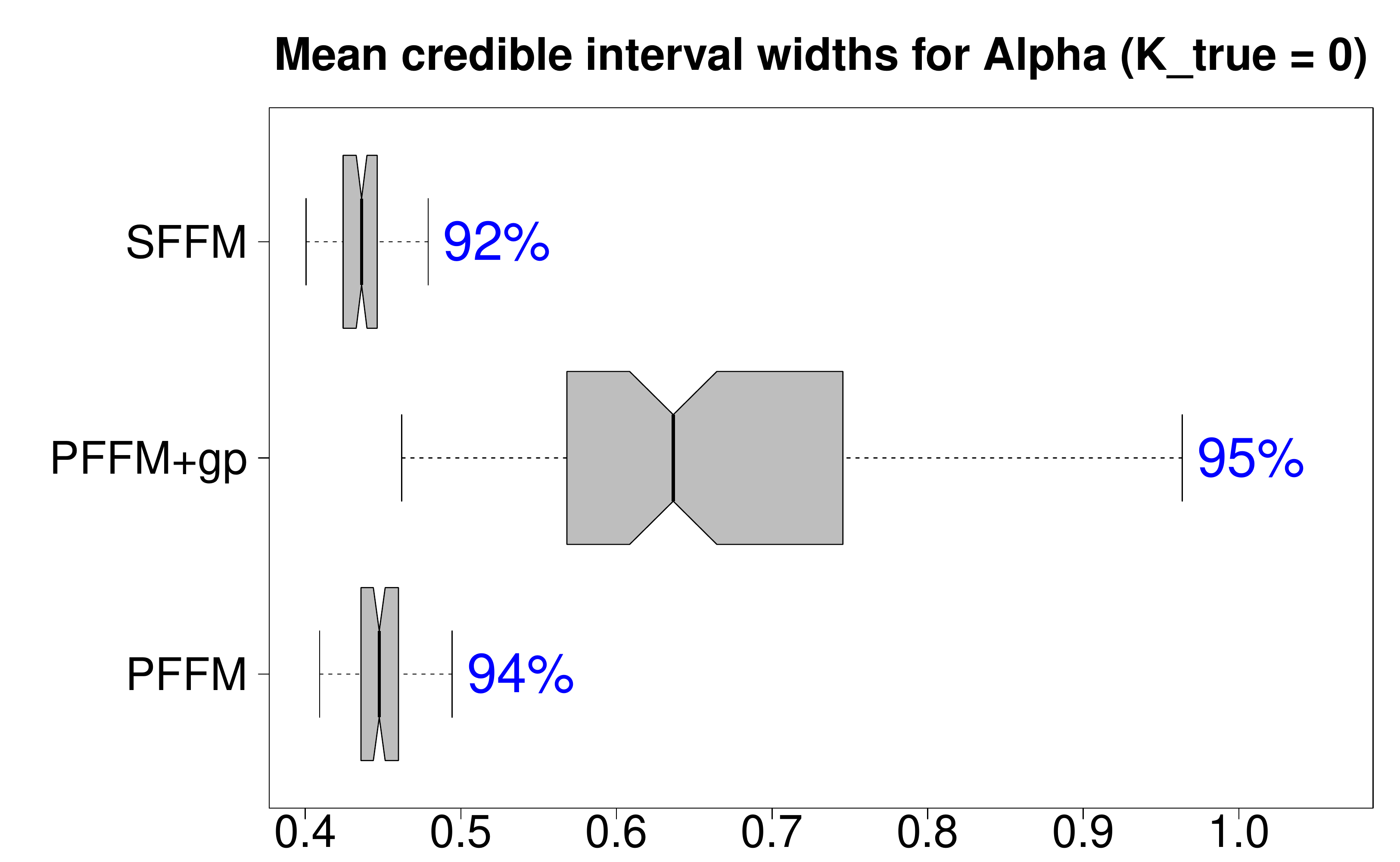}
\includegraphics[width=.49\textwidth]{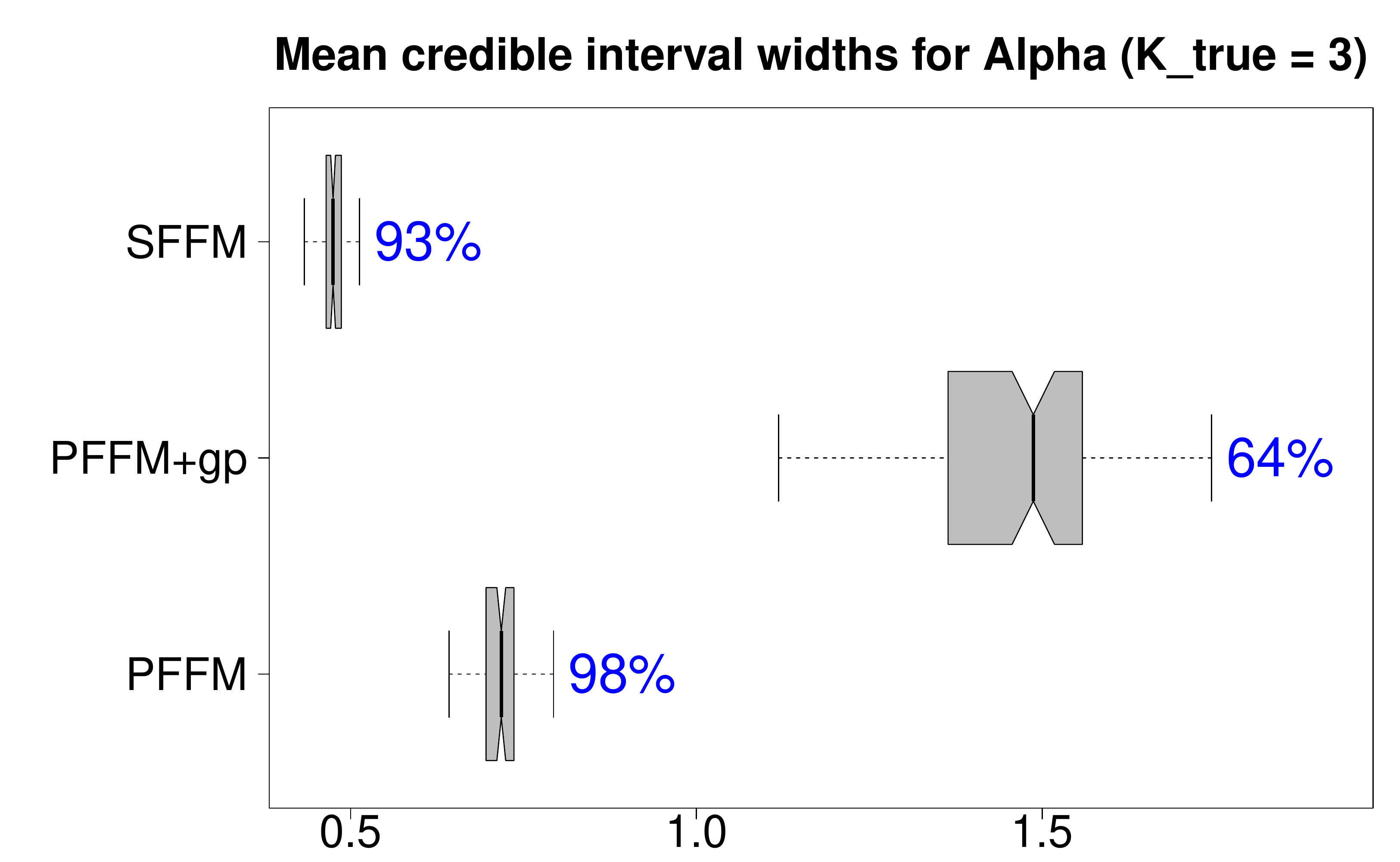}
\caption{\small Mean 95\% prediction (top) and credible (bottom) interval widths  for $\bm y_i^*$ and $\{\alpha_{\ell,i}^*\}$, respectively, with empirical coverage (annotations) for $K_{true} = 0$ (left) and $K_{true} = 3$ (right)  under the \cite{nelson1987parsimonious} template with unknown $\gamma$. The SFFM provides narrow interval estimates with (nearly) the correct nominal coverage in all cases. The PFFM intervals are excessively wide for $K_{true} > 0$, while the PFFM+gp performs poorly for $\{\alpha_{\ell,i}^*\}$ in all cases. 
\label{fig:ns-int}}
\end{center}
\end{figure}

 \begin{figure}[h]
\begin{center}
\includegraphics[width=.49\textwidth]{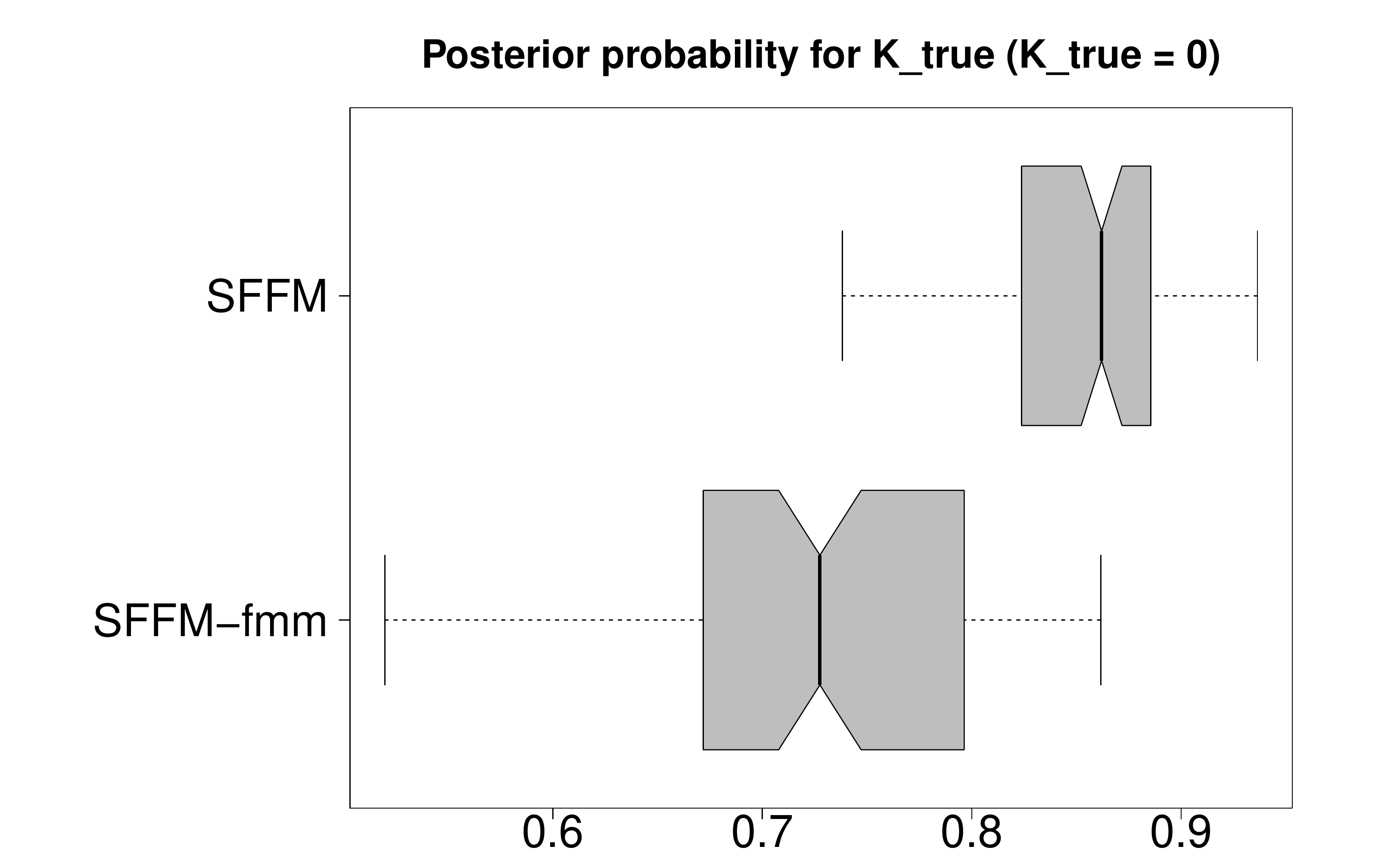}
\includegraphics[width=.49\textwidth]{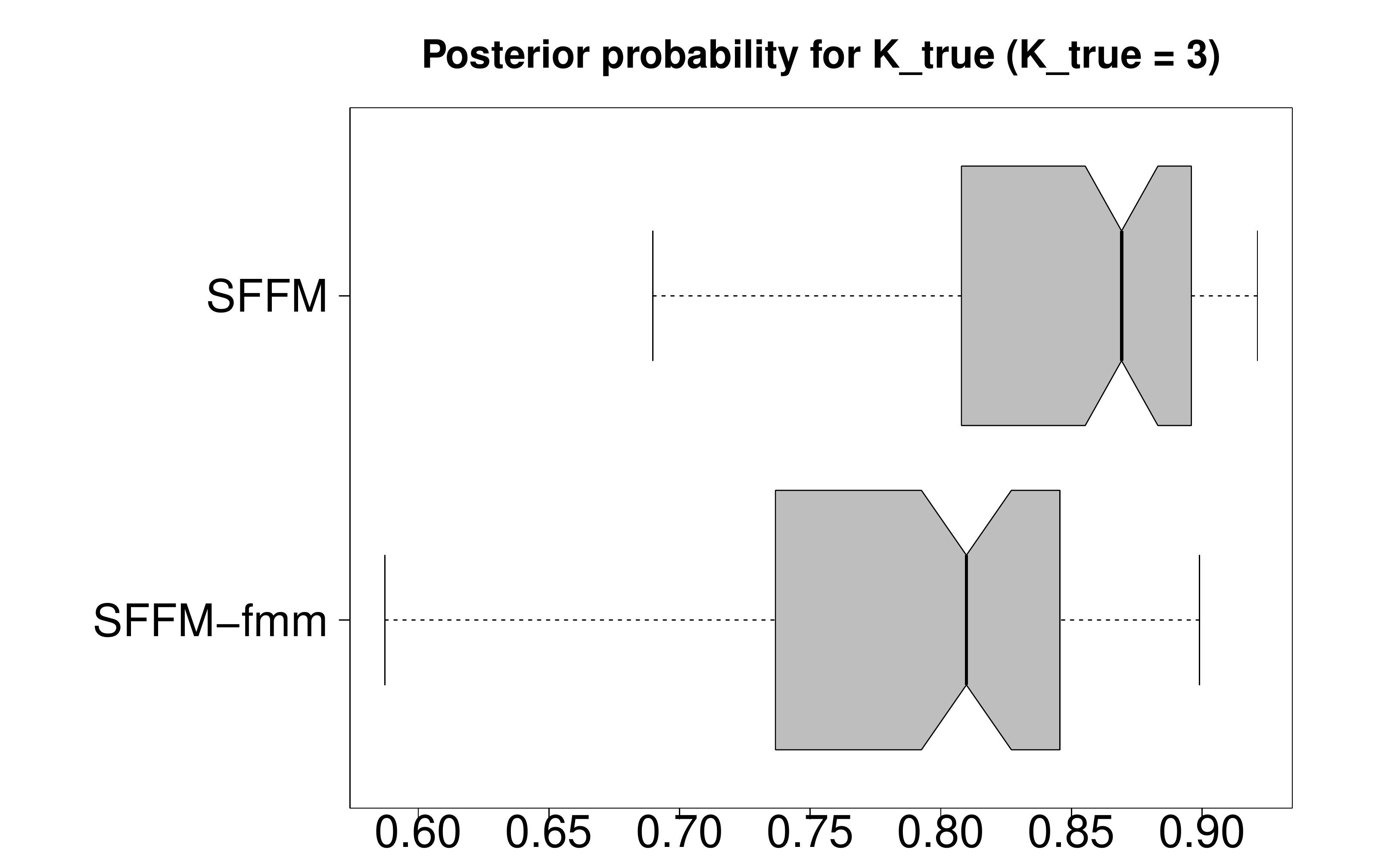}
\caption{\small   Probability score $\mathbb{P}(K^* = K_{true} \mid \bm y)$ 
for  $K_{true} = 0$ (left) and $K_{true} = 3$ (right) under the \cite{nelson1987parsimonious} template with unknown $\gamma$. The proposed ordered spike-and-slab prior provides better rank selection and inference than the finite mixture alternative.  
\label{fig:ns-kprob}}
\end{center}
\end{figure}

\subsection{Assessing inclusion probabilities for the nonparametric term}

To further evaluate the posterior distribution for the effective number of nonparametric terms, we report ranked probability scores for  $\mathbb{P}(K^*\mid \bm y)$ under the SFFM and SFFM-fmm in Figure~\ref{fig:kcrps}. Both methods perform similarly for $K_{true} = 8$, while the proposed SFFM dominates in the remaining cases. Next, we study the probability of inclusion for \emph{any} nonparametric terms,  $\mathbb{P}(K^*  > 0| \bm y)$. When $K_{true} > 0$, we find that $\mathbb{P}(K^*  > 0| \bm y) = 1$ in all scenarios, so the SFFMs both offer excellent power for detecting the presence of a nonparametric addition. For the more challenging case of the true parametric model with $K_{true} = 0$, we report the summary statistics of $\mathbb{P}(K^*  > 0| \bm y)$ across 100 simulations in Table~\ref{tab-pnz}. The overestimation issues for the finite mixture alternative are apparent: compared to the SFFM, the SFFM-fmm assigns far greater inclusion probabilities for the nonparametric term when no such term is needed.


 \begin{figure}[h]
\begin{center}
\includegraphics[width=.32\textwidth]{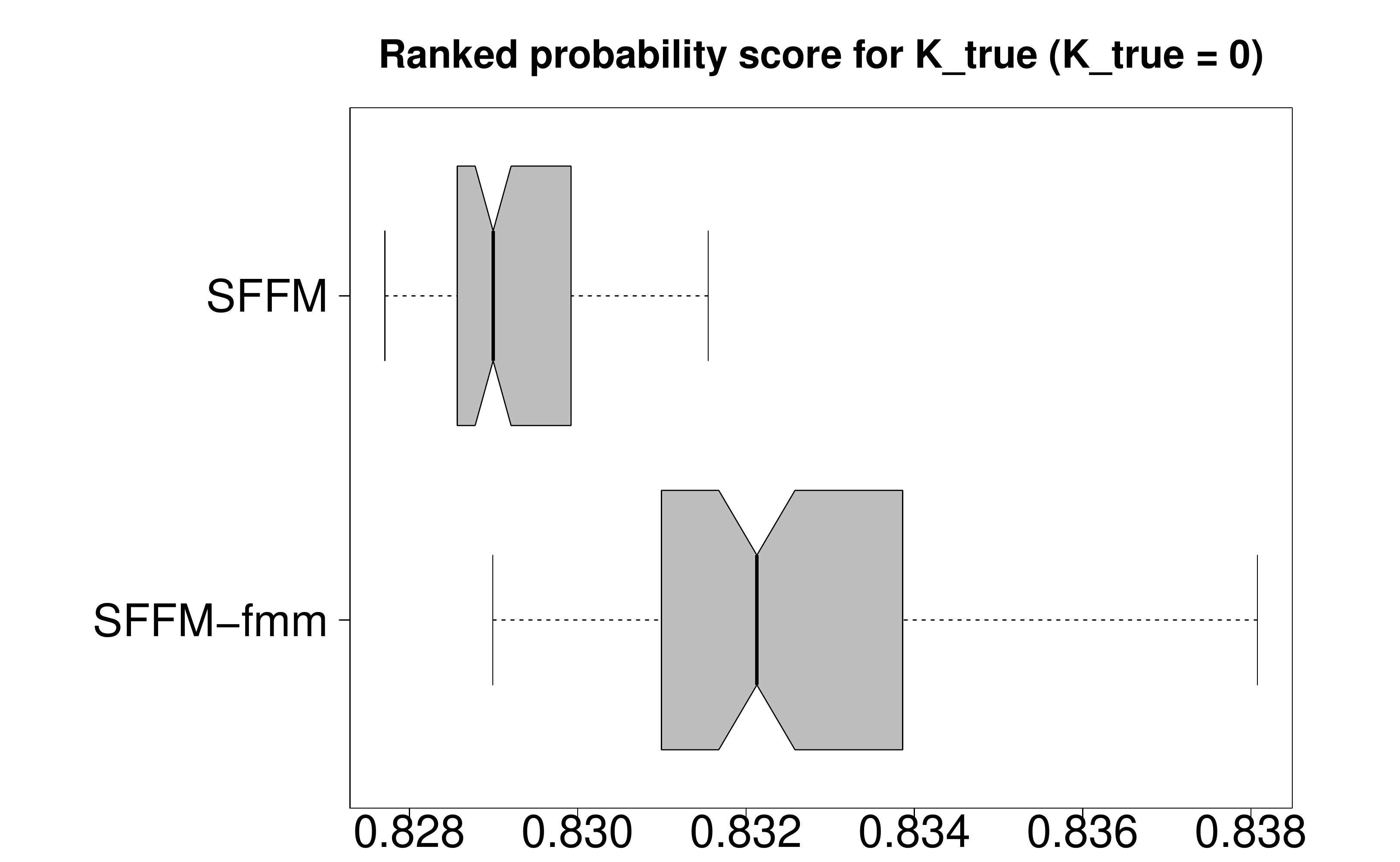}
\includegraphics[width=.32\textwidth]{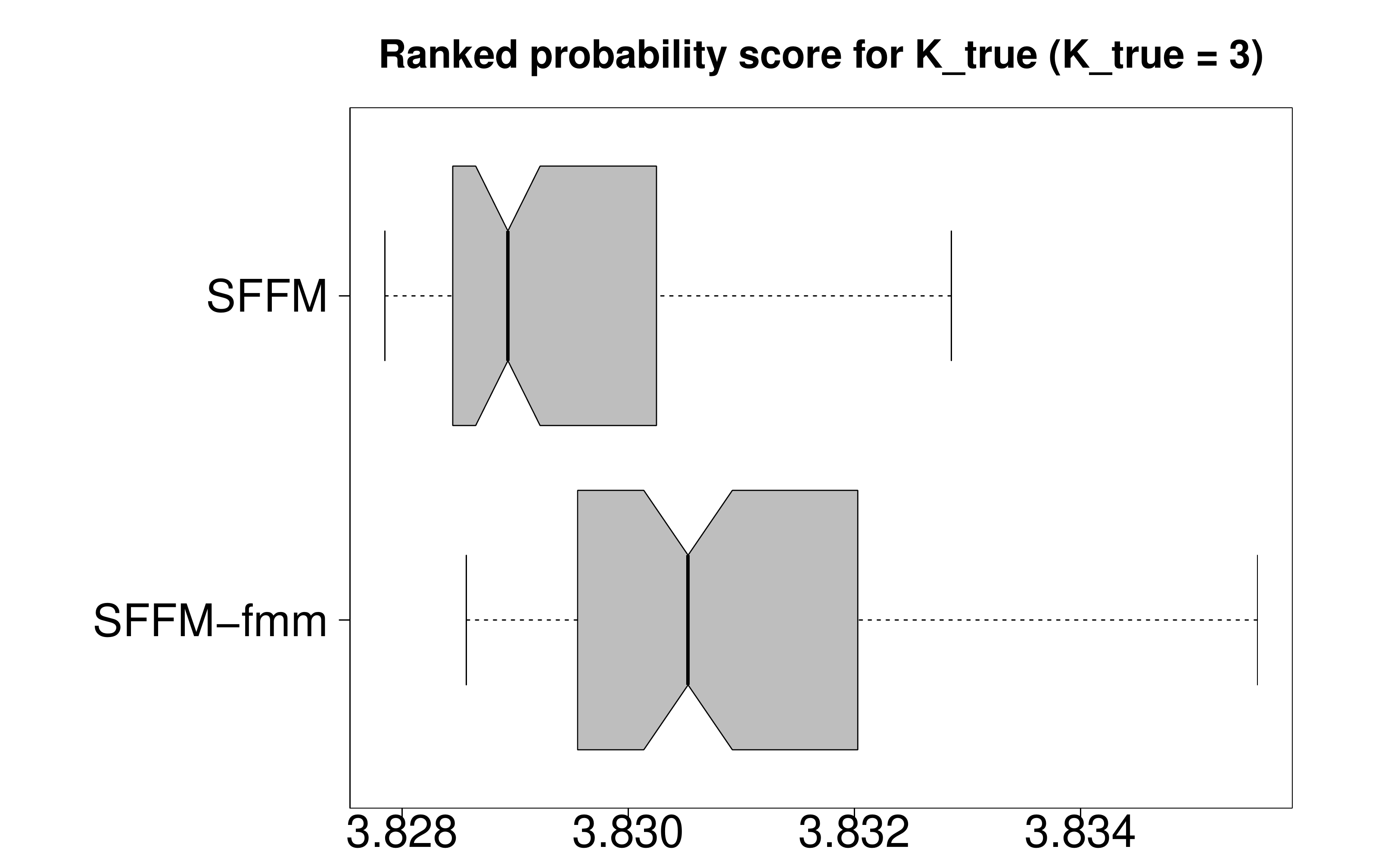}
\includegraphics[width=.32\textwidth]{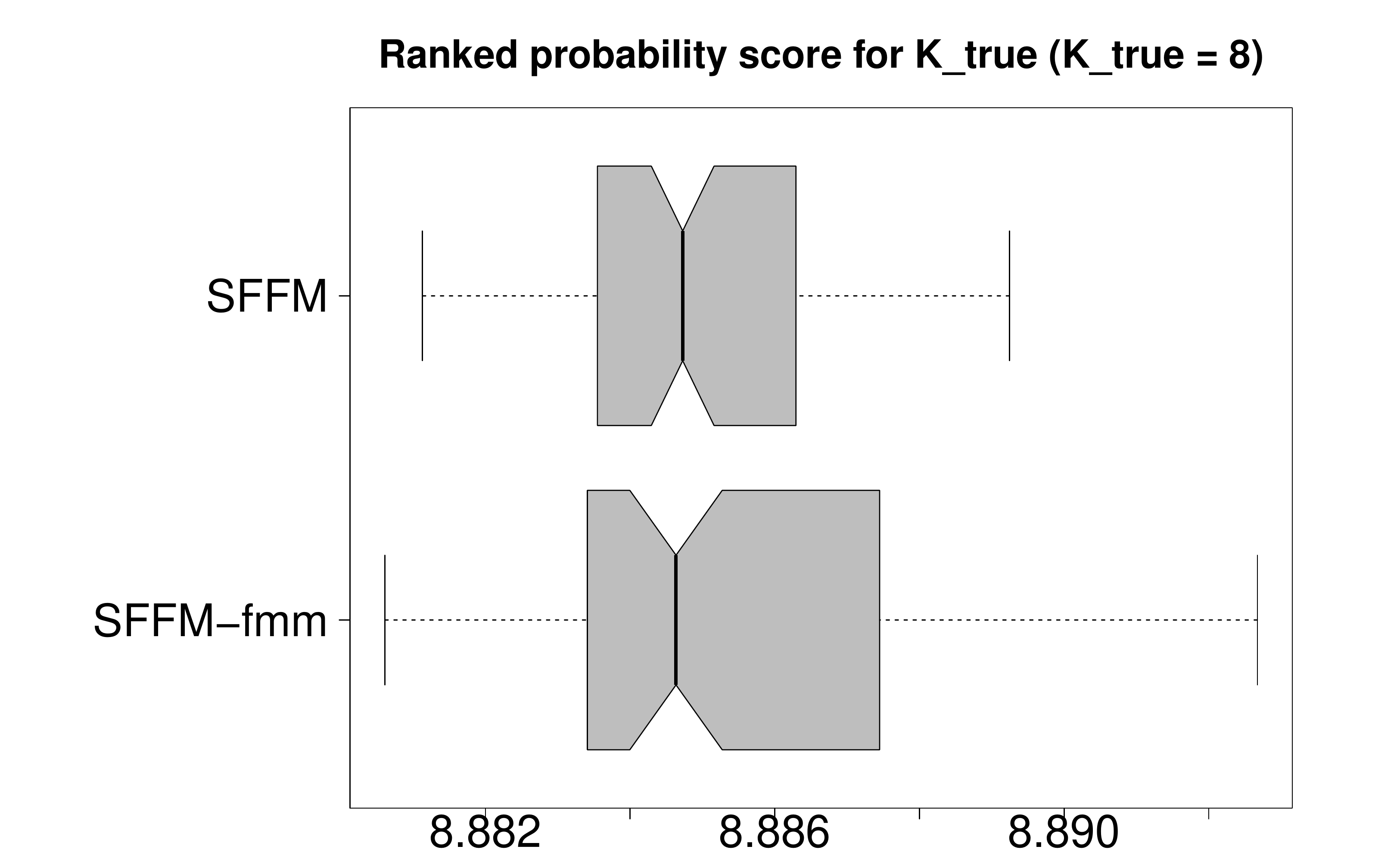}
\caption{\small  Negatively-oriented ranked probability score for $\mathbb{P}(K^* \mid \bm y)$   for  $K_{true} = 0$ (left), $K_{true} = 3$ (center), and $K_{true} = 8$ (right).  These results are for the linear template; results for the \cite{nelson1987parsimonious} template with unknown $\gamma$ are nearly identical and are omitted.
\label{fig:kcrps}}
\end{center}
\end{figure}

\begin{table}[ht]
\centering
\begin{tabular}{ c | c | c | c  | c | c | c   }
Template & Model & Min & 1st quartile &  Median &    3rd quartile  & Max \\ \hline
  \multirow{2}{*}{Linear} & 
  SFFM  & 0.070 &  0.117  & 0.142  &  0.183 &   0.370 \\ 
& SFFM-fmm  & 0.123  & 0.229  & 0.280  &  0.354  & 0.642  \\  \hline
  \multirow{2}{*}{Nelson-Siegel} & 
  SFFM  & 0.075 &   0.118  &  0.138    &  0.180  &  1.000 \\ 
& SFFM-fmm  &  0.146  &  0.219  &  0.265   &   0.314 &   1.000  \\  
 \end{tabular}
\caption{\small Summary statistics of $\mathbb{P}(K^* > 0 | \bm y)$ across 100 simulations.\label{tab-pnz}}
\end{table}

\end{document}